\documentclass{article}
\usepackage[dvipsnames]{xcolor}
\usepackage{colortbl}
\usepackage[final]{changes}

\usepackage{PRIMEarxiv}
\usepackage{amsmath}
\usepackage{amsthm}
\usepackage{amssymb}
\DeclareMathOperator*{\argmin}{arg\,min}
\RequirePackage{tgtermes}
\RequirePackage{newtxtext}
\RequirePackage{newtxmath}
\RequirePackage{bm}
\RequirePackage{endnotes}
\usepackage{algorithm}
\usepackage{algpseudocode}
\usepackage{tikz}
\usetikzlibrary{arrows.meta}
\usetikzlibrary{arrows}
\usepackage{caption}
\usepackage{subcaption} % (will load caption)
\usepackage{blindtext}
\usepackage{hyperref}

\usepackage{booktabs}
\usepackage{diagbox}
\usepackage{rotating}
\usepackage{mathtools}% Loads amsmath
\usepackage{graphicx}
\usepackage[utf8]{inputenc} % allow utf-8 input
\usepackage[T1]{fontenc}    % use 8-bit T1 fonts
\usepackage{url}            % simple URL typesetting
\usepackage{booktabs}       % professional-quality tables
\usepackage{nicefrac}       % compact symbols for 1/2, etc.
\usepackage{microtype}      % microtypography
\usepackage{lipsum}
\usepackage{fancyhdr}       % header
\usepackage{graphicx}       % graphics
\graphicspath{{media/}}     % organize your images and other figures under media/ folder
\usepackage{amsmath}
\usepackage{cases}
\usepackage{enumitem}
\usepackage{comment}
\usepackage{subfiles}
\usepackage[normalem]{ulem}

\newtheorem{theorem}{Theorem}

\newtheorem{lemma}[theorem]{Lemma}
\newtheorem{proposition}[theorem]{Proposition}

\newtheorem{assumption}{Assumption}
\newtheorem{remark}{Remark}
\newtheorem{example}{Example}
\hypersetup{
    colorlinks=true,
    linkcolor=blue,
    filecolor=cyan,      
    urlcolor=red,
    citecolor = red
    }

\usepackage{newtxmath}
\renewcommand{\vec}[1]{\overrightarrow{\mkern-4mu #1 }} % longer shaft, default head
\renewcommand{\vv}[1]{\vec{#1}}

\newcommand{\gen}[1][]{s_{#1}}
\newcommand{\genLim}[1][]{\overline{\mkern-1mu s}_{#1}}
\newcommand{\dem}[1][]{t_{#1}}
\newcommand{\flow}[1][]{f_{#1}}
\newcommand{\flowTot}{\bm{f}}
\newcommand{\flowLim}[1]{\overline{\,\mkern-5mu f}_{#1}}
\newcommand{\flowLimVec}{\bm{\overline{\,\mkern-4mu f}}}
\newcommand{\PR}[1]{\mathbb{P}\left(#1\right)}
\newcommand{\R}[0]{\mathbb{R}}
\newcommand{\N}[0]{\mathbb{N}}

\newcommand{\EE}[0]{\mathcal{E}}
\newcommand{\V}[0]{\mathcal{V}}
\newcommand{\G}[0]{\mathcal{G}}
\newcommand{\K}[0]{\mathcal{K}}

\newcommand{\nv}[0]{|\V|}
\newcommand{\nee}[0]{|\EE|}
\newcommand{\nk}[0]{|\K|}
\newcommand{\demM}[0]{\bm{T}}
\newcommand{\genM}[0]{\bm{S}}
\newcommand{\flowM}[0]{\bm{F}}
\newcommand{\incM}[0]{\bm{B}}
\newcommand{\e}[1]{\textbf{e}_{#1}}
\newcommand{\paretoI}[1]{\vv{X}_{\mkern-7mu#1}}
\newcommand{\pareto}[0]{\vv{X}\mkern-2mu}
\newcommand{\flowCost}[0]{c_f}
\newcommand{\flowMOpt}[1]{\bm{F}^*(#1)}

\newcommand{\genMLim}[0]{\bm{\overline{\mkern-2muS}}}
\newcommand{\OO}[1]{\mathcal{O}\mkern-2mu\left(#1\right)}

\newcommand{\exc}[1]{\phi_{#1}}

\newcommand{\constCostFunction}[1]{w_{#1}}
\newcommand{\boundZ}{m_Z}
\newcommand{\detZ}{z}
\newcommand{\failFactor}{\psi}
\newcommand{\cascadeSeqRand}{C}
\newcommand{\cascadeSet}{\mathcal{\cascadeSeqRand}}
\newcommand{\cascadeSeqDet}{c}
\newcommand{\cascadeCost}{Z}
\usepackage{natbib}

\title{Heavy tails in dynamic flow networks\\
\large universal explanation of their emergence

}

\author{
  Agnieszka Janicka$^1$, Fiona Sloothaak$^{1}$, Maria Vlasiou$^{2}$, and Bert Zwart$^{1,3}$ \\
  $^1$\textit{Eindhoven University of Technology, 5612 AZ Eindhoven, the Netherlands}\\
  $^2$\textit{University of Twente, 7522 NB Enschede, the Netherlands}\\
  $^3$\textit{Centrum Wiskunde and Informatica, 1098 XG Amsterdam, the Netherlands}
}

\begin{document}
\maketitle
%%%%%%%%%%%%%%%%%%%%%%%%%%%%%%%%%%%%%%%%%%%%%%%%%%%%%%%%%%%%%%%%%%%%%%
\begin{abstract}%
% Enter your abstract
Overload-induced cascading failures can cause extreme disruptions in a wide range of networked systems, such as power grids, transportation networks, or financial systems. Empirical studies across domains report that the size of such disruptions often follows a Pareto- or heavy-tailed distribution. While many models reproduce this scaling behavior, they are either tailored to specific domains or based on simplified mechanisms that overlook key aspects of overload cascading behavior. Hence, a general understanding of the mechanisms driving scale-free behavior in these settings remains incomplete. In this paper, we develop a universal and analytically tractable model of overload cascading failures on flow networks, offering a new perspective on how Pareto-tailed disruptions emerge across networks. Our framework shows, under mild assumptions, that heavy-tailed disruptions can arise naturally from Pareto-tailed external inputs, and it establishes a transformation law linking the input and output tail exponents. We further identify broad conditions under which the resulting cascade cost exhibits a heavy-tailed distribution and show that the mechanism is robust across several domains, including power transmission, traffic networks, and processing systems. Our results provide a unified explanation for the emergence of scale-free failures in overload-driven systems and connect previously disparate, application-specific models under a unified framework.
\end{abstract}%
% Text of your paper here


\section{Introduction}

%\textcolor{red}{General introduction to cascading failures}\\
Network models are often used to analyze the behavior of complex systems such as energy grids, financial markets, or transportation networks. Subsequent component failures, called \textit{cascading failures}, are prevalent in such systems and widely studied \citep{Guo2017,Huang_Vodenska_Havlin_Stanley_2013,Daqing_2014, Watts2002, Crucitti2004,Mahdi2017}. \deleted{Cascading failures also occur in multilayered systems, e.g., railway and highway networks, where failures in one system cause increased load or congestion in the other and may lead to their failure.} \replaced{Such f}{F}ailures are typically classified into two distinct groups: \textit{structural} and \textit{overload} failures. Structural failures spread locally when the failure of one component destabilizes its neighborhood, causing subsequent failures of the weakened components. These type of failures are typically seen in epidemics or wildfires \citep{manzano2014epidemic, Pastor2001, Manzano2013, VC2011, Turcotte_Malamud_Guzzetti_Reichenbach_2002}. Overload failures, which are the subject of this work, need not be local and occur when the capacity of a component is exceeded; that is, an overload occurs, causing redistribution of flow and consequent overloads \citep{CNCFreview}. Such cascades \added{commonly }emerge in man-made networks such as power systems, traffic networks, financial markets, communication networks, and more \citep{Ren2018, Elliott_Golub_Jackson_2012, IoTReview2021}. In this work, we develop a model of overload cascades to study how local component failures may lead to critical system-wide disruptions.

%\textcolor{red}{Heavy-tailed disruptions often emerge}\\
\deleted{Cascading failures can \replaced{lead to}{cause} critical system-wide disruptions. }In many complex systems, there is empirical evidence suggesting that the size of disruptions exhibits a Pareto (a.k.a.\ power-law) or heavy tail. A random variable $X$ is Pareto-tailed with tail parameter $\alpha>0$ if
\begin{equation}
\label{def:pareto-tailedInformal} \PR{X>x} \approx Kx^{-\alpha}\quad\text{as}\quad x\rightarrow \infty,
\end{equation}
for some $K>0$.
%\begin{equation}
%    \label{def:pareto-tailed}\PR{X>x} \sim K x^{-\alpha} \quad\text{as}\quad x\rightarrow \infty, 
%\end{equation}for some $K>0$. Here we imply that $f(x)\sim g(x)$ as $x\rightarrow \infty \iff \lim_{x\rightarrow \infty} f(x)/g(x) = 1$. 
Such scaling behavior has been observed in several domains. In power transmission systems, cascading failures lead to Pareto-tailed  blackouts \citep{chen2001analysis}. In traffic networks, failures represent road congestion and cause widespread traffic jams. As seen in the literature, various performance measures of traffic jams, such as one-dimensional average traffic jam length \citep{Janicka_Sloothaak_Vlasiou_Zwart_2025} or multidimensional spatio-temporal performance measures \citep{Zhang_Zeng_Li_Huang_Stanley_Havlin_2019}, have a Pareto-tail distribution. Other examples of heavy-tailed disruptions include reinsurance losses \citep{ ibragimov2015heavy}, systemic risk in finance \citep{Liu_Yang_2021,Thurner_Farmer_Geanakoplos_2010}, delay times of flights, \citep{Fleurquin_Ramasco_Eguiluz_2013}, or area burned in wildfires \citep{Moritz_Morais_Summerell_Carlson_Doyle_2005}. Under these scaling laws, the occurrence of extreme disruptions is non-negligible; therefore, understanding the emergence of these phenomena is imperative.

%\textcolor{red}{Models explaining heavy-tailed disruptions} \\
Many models have been developed to explain the emergence of Pareto-tailed disruptions in complex systems. Some aim to provide a universal explanation, while others are application-specific. The latter models seek to capture the particular dynamics of the system in question, such as the physical laws of electricity distribution in power grids or human driving behavior in traffic networks \citep{Dobson_Carreras_Lynch_Newman_2007, Gang_Jia_Herniter_2022, Beggs_Plenz_2003, PhysRevLett.68.1244, Bak_Paczuski_Shubik_1997, PhysRevLett.69.1629, LI20211, Blanchet_Shi_2012}. In contrast, universal models aim to uncover global mechanisms without fine-tuning to specific dynamics of the system. Prominent examples include the critical branching model \citep{Branching2002}, the sandpile model \citep{SoC1987,Dhar1999}, percolation models \citep{Barabasi2000,LI20211,Watts2002}, and load-capacity models \citep{Motter2002}.

While the literature is rich, there is no consensus on whether existing universal models fully explain the observed scale-free behavior or if other unexplored mechanisms are involved. Existing models provide important insight, but have limitations --- especially in capturing the dynamics of overload failures. Specifically, the critical branching model is mathematically elegant and well understood, but assumes simple dynamics and is limited to the scaling exponent $\alpha = \frac{1}{2}$ in \eqref{def:pareto-tailedInformal}. The sandpile model shows how a system can organize itself to a critical state where cascade sizes follow a power law. While this approach inspired models in %physics, biology, finance, or geology 
many fields \citep{PhysRevLett.69.1629, Gang_Jia_Herniter_2022, Beggs_Plenz_2003, PhysRevLett.68.1244, Bak_Paczuski_Shubik_1997}, it cannot capture the non-local overload dynamics characteristic of many real-world systems. Percolation models explain the emergence of power-law cascade sizes from the underlying graph topology, using a simple percolation rule. %, which activates or deactivates vertices or edges
 This allows for in-depth analytical understanding. Again, cascades occur locally, making the model an excellent tool for studying structural cascading failures, such as the spread of a disease or wildfires \citep{LI20211}, but not suitable for a non-local overload process that requires a notion of capacity. Load-capacity models overcome this limitation by allowing overloaded components to fail and redistribute their load. Their scale-free cascade properties arise from the assumed scale-free graph topology. However, because these models depend on predetermined network structures and a simple load-redistribution rule (e.g., shortest paths), they may be unrealistic for some real-world applications. Thus, while Pareto-tailed disruptions are well studied in the context of structural failures and scale-free topologies, a universal understanding of overload-driven disruptions remains incomplete.

%\textcolor{red}{Recently new explanation with shortcomings}\\
Recently, an alternative explanation of Pareto-tailed disruptions has been proposed. \cite{Nesti_2020} developed an application-specific model of power transmission networks, where the Pareto-tailed distribution of blackouts is inherited from the Pareto-tailed distribution of power demand, driven by city sizes. This simple yet novel explanation matches the empirical findings on city sizes and blackouts \citep{Nesti_2020}. Inspired by this model, a second application-specific model has been developed to explain Pareto-tailed traffic jams, where the notion of demand --- traffic intensity --- also follows an empirical power law \citep{Janicka_Sloothaak_Vlasiou_Zwart_2025}. Notably, while the specific dynamics of the two models differ significantly, both exhibit the same emergence phenomenon, suggesting the presence of a more general underlying principle. However, because each model is tailored to a particular domain, they cannot, on their own, establish the generality of this explanation. In this paper, we offer a universal model for overload cascading failures based on the same underlying principle, explaining how Pareto-tailed disruptions emerge from Pareto-tailed external inputs.

%\textcolor{red}{Our approach}\\
Our approach balances abstraction with application-level realism, making our model both broadly applicable and analytically tractable. We capture essential properties of flow and allocation of resources in networks using the principle of least action \citep{landau1960mechanics, Singh2022}, which states that systems tend to be organized in a manner minimizing total energy dissipation \citep{hess14, Doyle_Snell_1984, Murray1926}. This approach can model versatile flow patterns like Kirchhoff's laws \cite[ch.~8]{wood2013power} in power networks or Wardrop equilibria \citep{Janicka_Sloothaak_Vlasiou_Zwart_2025} in traffic, and holds for any network topology. This overcomes the limitations of other universal models, such as sandpile, percolation, or load-capacity models. Crucially, we leverage the persistence of the Pareto-tailed input behavior, which implies that detailed microdynamics of the model do not substantially affect the Pareto-tailed nature of large disruptions. This approach enables us to derive rigorous distributional results that apply across diverse types of networks. Consequently, we strengthen the emergence hypothesis of \cite{Nesti_2020, Janicka_Sloothaak_Vlasiou_Zwart_2025}, elevating it from an application-specific explanation to a universal principle for systems with cascading overloads.

%\textcolor{red}{Our key contributions}\\
Our main contribution is to explain, via a universal and analytically tractable model, how heavy-tailed disruptions can emerge in a wide range of networked systems. Specifically, we:
\begin{enumerate}
    \item develop a universal cascading framework for overload flow networks, applicable across domains;
    \item prove that the disruptions inherit a Pareto tail from the input distribution, and identify the transformation law governing the new tail exponent;
    \item establish widely applicable sufficient conditions for the universal scale-free behavior;
    \item demonstrate the versatility of the framework through applications to power\deleted{ transmission}, traffic, and processing networks --- highlighting a possible existence of a unifying mechanism behind extreme disruptions.
\end{enumerate}
These contributions have important implications for mitigating extreme cascades. Since the heavy-tailed behavior is inherited from the input distribution, reducing the tail heaviness is difficult without modifying that distribution itself. Network upgrades --- such as capacity increases or dynamic routing policies --- can reduce the multiplicative constant $K$ in \eqref{def:pareto-tailedInformal}, but generally do not affect the tail exponent $\alpha$, so the potential for extreme events persists. Nevertheless, decreasing this constant can still yield meaningful improvements in large networks. Furthermore, our results can be used to rank network components by their likelihood of triggering large cascades, which can guide the placement of early detection and targeted intervention mechanisms in the most critical parts of the system.

%\textcolor{red}{Describe our unifying framework}\\
In the remainder of this section, we discuss our contributions in more detail. Our model assumes a fixed graph, where commodities flow through the edges from source vertices to sink vertices. \textit{Sink requirements}, that is, the demand for a commodity at each sink, are functions of \textit{vertex weights} that follow a Pareto-tailed distribution. \textit{Source production}, derived as some function of sink requirements, satisfies three constraints: \textit{production-requirement balance}, \textit{production capacity}, and \textit{flow capacity}.
The flow is distributed according to some parametrized mechanism that ensures commodity balance at each vertex. A cascade of failures is a discrete stochastic process on this flow network. We assume a probabilistic failure mechanism, where overloaded edges can fail, resulting in a capacity decrease of the edge. The disruptive impact of the cascade on the system is quantified using the \textit{cascade cost} function, which is determined after the cascade terminates. A detailed description of this framework is provided in Section \ref{modelDescription}.

%\textcolor{red}{Key asymptotic results}\\
To make our framework analytically tractable, we impose mild continuity and scalability assumptions on the cascade cost and on the cascade probability (see Assumption \ref{assumption}). Our first key result, stated in Proposition~\ref{sbj}, identifies the most likely cause of a large cascade cost. In particular, it shows that large cascade costs typically occur when one vertex has a significantly larger weight than the others. Our second key result, Theorem~\ref{mainThm}, builds on this notion and shows that, under a certain condition on the vertex weight vector, the tail distribution of the cascade cost inherits the scale-free tail from the vertex weight distribution. However, the corresponding tail exponent in \eqref{def:pareto-tailedInformal} changes from $\alpha$ to $\alpha/\delta$, where $\delta$ is a parameter that depends on the choice of the cascade cost function. This aligns with the results in \citep{Nesti_2020, janicka2024scalefree} for power networks and in \citep{Janicka_Sloothaak_Vlasiou_Zwart_2025} for highway traffic networks, both of which show such a scaling behavior for particular cost functions with $\delta = 1$. However, the results in this paper are more general: they hold for a class of parameterized cascade mechanisms and apply to a wider range of network types and cost functions.

%\textcolor{red}{Results on sufficient conditions}\\
While our results are broadly applicable, Assumption~\ref{assumption} can be nontrivial to verify, as it requires understanding how the cascade cost depends on the initial conditions. To address this, in Section~\ref{sec:satisfiability}, we show that this assumption holds for several broad classes of parameterized modeling choices for source production, flow distribution, and cascade cost functions. Importantly, the proposed modeling choices are both physically interpretable and consistent with mechanisms commonly used in flow network models. The corresponding results are stated in Propositions~\ref{propClassFunc1} and~\ref{propClassFunc2}, with detailed proofs provided in Appendices~\ref{proofs} and \ref{subsec:proofsPropositions}. These proofs are technical in nature and analyze the behavior of optimal solutions to a family of parametric convex optimization problems arising at each stage of the cascade. Scalability is established through transformations of the optimization problems, while continuity is proved via a contradiction argument based on compactness and a carefully constructed sequence of feasible solutions. The complete argument proceeds inductively and characterizes how the cascade dynamics respond to changes in the input distribution. This proof strategy may also be employed to verify Assumption~\ref{assumption} for other classes of modeling choices beyond those considered in this paper.

%\textcolor{red}{Flexibility of modeling choices}\\
%Importantly, the proposed modeling choices are both physically interpretable and consistent with mechanisms commonly used in flow network models. For example, our flow distribution mechanism is based on the principle of least action \citep{landau1960mechanics, Singh2022}, which states that flows tend to minimize total energy dissipation \citep{hess14, Doyle_Snell_1984, Murray1926}. With appropriate energy functions, this principle captures realistic flow patterns \citep{Whittle_2007}. Incorporating such modeling techniques ensures that our framework remains applicable to a wide range of real-world networked systems.

%\textcolor{red}{Application examples of our framework}\\
Lastly, to illustrate the universality of our modeling framework, Section~\ref{Applications} presents three examples: a power transmission network, a highway traffic network, and a processing network. For each case, we detail how the system can be modeled within our framework, emphasizing the specific choices of source production, flow distribution, and cascade cost functions that best capture the network's behavior. We then verify that the sufficient conditions from Section~\ref{sec:satisfiability} are met%, and apply Theorem~\ref{mainThm} to show that the cascade cost \(\cascadeCost\) exhibits a Pareto tail across all examples
. These examples not only confirm the versatility of our framework but also demonstrate that heavy-tailed cascade costs can emerge robustly across domains.

\section{Multi-commodity cascade model of flow network} \label{modelDescription}
In this section, we introduce a universal overload cascading failure model for flow networks. We first present notational conventions used throughout the paper. Next, we define the flow network architecture, where we introduce the static structure of the network. Finally, we explain the temporal dynamics of the system, detailing how the network operates at each time step. The dynamics are divided into three stages: a planning stage, where flow capacities are set based on anticipated demand; an operational stage, representing the typical network behavior; and an emergency stage, where cascading failures occur due to overloads. %The model is described in full generality to ensure its applicability across various scenarios. Lastly, we state an additional assumption under which our results are derived and define classes of modeling choices that satisfy this assumption. 
\subsection{Notational conventions} \label{notation}
Throughout this paper, we adopt the following notation and conventions.
Let $\N$ denote the set of natural numbers excluding 0, $\R$ the set of real numbers, and $\R_+=\{x\in \R:x\geq 0\}$ the set of non-negative reals. For any positive integer $n$, we denote by $[n]$ the set $\{1,2,\dots,n\}$. Other sets are typically denoted using calligraphic capital letters, i.e., $\mathcal{A}$. The set size is denoted by $|\mathcal{A}|$, and the power set of $\mathcal{A}$ by $2^{\mathcal{A}}$. 

Vectors in $\R^n$ are denoted by bold lowercase letters, such as $\bm{x} = (x_1,\dots, x_n)$, where $x_i$ represents the $i$-th component of $\bm{x}$. The diagonal matrix with the vector $\bm{x}$ on the main diagonal is denoted by $\text{diag}(\bm{x})$. Matrices are represented by bold uppercase letters, such as $\bm{A} = [a_{i,j}]\in \R^{n\times m}$ with $i\in [n]$ and $j\in [m]$. %The only exception to this rule applies to $\pareto\in \R_+^n$, which represents a \textit{random} vector. 
The transpose of a vector or matrix $\bm{A}$ is written as $\bm{A}^T$. The identity matrix of size $n$ is denoted by $\bm{I}_n$, while $\e{n}$ denotes the all-ones column vector of size $n$ and $\e{{i,n}}$ denotes the $i$-th unit vector of size $n$, that is, the column vector that has a 1 in the $i$-th position and 0 otherwise. 

We say that a function $h:$ $\R^n \rightarrow \R^m$ is \textit{$\delta$-scale-invariant} with respect to (w.r.t.) $\bm{x}$ if for any $\bm{x}\in \R^n$,
\begin{equation}
h(\omega \bm{x}) = \omega^\delta h(\bm{x}) \quad \text{for all}\quad \omega>0.        \label{eq:deltaScaleInv}
    \end{equation}
In this paper, the terms \textit{1-scale-invariant} and \textit{scale-invariant} are used interchangeably. Informally, a function is scale-invariant, if a change in scale only affects its output proportionally, without changing the fundamental nature of the function itself. %In our analysis, $\delta$-scale invariance is used to ensure that the network features (e.g., sink requirements, source production, edge flow capacities) scale consistently with vertex weights.   
Furthermore, we say that a function $h:$ $\mathcal{D}\subseteq \R^n \rightarrow \R^m$  is \textit{sequentially right continuous (SRC)} if for every $\bm{x}^*\in \mathcal{D}$ and  and any sequence $\{\bm{x}_i\}_{i \in \mathbb{N}}$ satisfying  $\bm{x}_i\geq \bm{x}^*$ for all $i\in \N$ and $\bm{x}_i\rightarrow \bm{x}^*$ as $i\rightarrow \infty$, 
\begin{equation}\label{eq:SRC}\lim_{i\rightarrow \infty} h(\bm{x}_i) = h(\bm{x}^*).\end{equation}
Note that \textit{sequential right continuity} is a weaker form of \textit{continuity}, meaning that every continuous function is SRC. Furthermore, for two functions $g(x)$ and $h(x)$ we say that $g(x) = \OO{h(x)}$ as $x\rightarrow \infty$ if and only if $\exists~N$, $x^*>0$ such that $|g(x)|\leq N|h(x)|$ for all $x\geq x^*$.

Finally, to distinguish deterministic evaluations of a random variable from the random variable itself, we adopt the following notation: for any random variable \(A\) that is \textit{fully determined} by random variables \(B_1\) and \(B_2\), we denote its deterministic outcome, given $B_1 = b_1$ and $B_2 = b_2$, by $a(b_1,b_2)$. In particular, using slightly informal notation,
\begin{equation}
    a(b_1, b_2) := A \mid (B_1 = b_1, B_2 = b_2).
    \label{conditionalRV}
\end{equation}
\subsection{Flow network architecture} \label{basicComp}
In this section, we define the flow network architecture; that is, we introduce the static structure of the network and formally describe its components. The flow network, over which some commodities are transported, is represented by a directed connected graph $\G = (\V, \EE)$. Here, $\V$ and $\EE$ represent the set of vertices and edges of $\G$, respectively. The graph is described by its \textit{incidence matrix} $\incM\in \{0,1,-1\}^{\nv\times\nee}$, given by
\[b_{v,e} := \begin{cases} 1, & \text{if edge } e\in \EE \text{ enters vertex } v\in \V,\\
-1, & \text{if edge } e\in \EE \text{ exits vertex } v\in \V,\\
0, &\text{otherwise.}
\end{cases}\]

The set of all commodities is denoted by $\K$. Each vertex $v\in \V$ has a stochastic \textit{weight} $\paretoI{v}$ and the vector of all weights is denoted as $\pareto\in \R_+^{\nv}$.  %, requires $\dem[v,k]$, produces $\gen[v,k]$, and has production capacity of $\genLim[v,k]$ units of commodity $k\in \K$. We assume that vertex weights $\paretoI{v}$, $v\in \V$ 
We assume that all $\paretoI{v}$'s are mutually independent and follow a Pareto-tailed distribution with scale parameter $\alpha>0$ and constant $K>0$, i.e.,
\begin{equation}
\PR{\paretoI{v}>x} \sim Kx^{-\alpha}, \quad x\rightarrow \infty,\quad \text{ for all } v\in \V,     \label{def:pareto-tailed}
\end{equation}
where we imply that $f(x)\sim g(x)$ as $x\rightarrow \infty \iff \lim_{x\rightarrow \infty} f(x)/g(x) = 1$. This choice is motivated by the fact that in many physical networks, the vertices correspond to quantities whose distribution often has a Pareto tail, such as city sizes, data packets, or insurance claims \citep{Berry1961, crovella1998heavy, nair_wierman_zwart_2022}. 
%In  Motivated by this, we assume that vertex sizes correspond to city sizes, which often have a scale-free tail. Hence, 

Each vertex $v\in \V$ both \textit{requires} and \textit{produces} a certain amount of commodity $k\in \K$, given by $\dem[v,k]$ and $\gen[v,k]$, respectively. Moreover, the production $\gen[v,k]$ cannot exceed the production capacity $\genLim[v,k]$. Inspired by the terminology in optimization theory, we refer to $\demM = [\dem[v,k]]_{v\in\V, k\in \K}$ and $\genM = [\gen[v,k]]_{v\in \V, k\in \K}$ as the \textit{sink requirement} and the \textit{source production matrices}, respectively, while $\genMLim = [\genLim[v,k]]_{v\in\V,k\in\K}$ is called the \textit{production capacity matrix}. \replaced{Importantly,}{It is important to note that} the matrices $\demM$, $\genM$, and $\genMLim$ are not exogenously given but all depend on the vertex weight vector $\pareto$. The exact dynamics\deleted{,} as \deleted{explicit or implicit }functions of $\pareto$\deleted{,} are stated in Section~\ref{subsec:temporalDynamics}.  %This terminology is inspired by flow networks in optimization theory where the vertices with positive requirements (also known as demands) are called sink vertices; otherwise, they are called source vertices. Finally, we define the production capacity matrix as $\genMLim:=  [\genLim[v,k]]_{v\in \V, k\in \K}\in \R_+^{\nv\times\nk}$. 

At each vertex \( v \in \V \) and for each commodity \( k \in \K \), the netput requirement is defined as
\(
u_{v,k} := \dem[v,k] - \gen[v,k],
\)
 which yields the \textit{netput matrix} $\bm{U}:= \demM - \genM$.
 If \( u_{v,k} > 0 \), vertex \( v \) has a \emph{shortage} of commodity \( k \); if \( u_{v,k} < 0 \), it has a \emph{surplus}. These imbalances give rise to a network flow: commodities are transported from surplus vertices to those with shortages. The \emph{flow matrix} \( \flowM \in \mathbb{R}^{\nee \times \nk} \) specifies the amount of each commodity flowing along each edge. The \emph{total flow vector} \( \flowTot := \flowM \e{_{|\K|}} \in \mathbb{R}^{\nee} \) captures the aggregate flow on each edge. Each edge has a capacity, represented by the \emph{edge capacity vector} \( \flowLimVec \in \mathbb{R}_+^{\nee} \). If the flow on an edge exceeds its capacity, the edge may fail, as discussed in Section~\ref{emergencyPhase}.

%The network flow is induced by the netput matrix $\bm{U}:= \demM - \genM$. If for $v\in \V$ and $k\in \K$, the netput requirement $u_{v,k}>0$, then vertex $v$ has a shortage of commodity $k$; conversely, if $u_{v,k}<0$, it has a surplus. Network flow arises when the commodities from the surplus vertices are transported through the network to those with shortages. The flow matrix $\flowM\in \R^{\nee\times\nk}$ specifies the amount of flow of each commodity on each edge, while the total flow vector $\flowTot:= \flowM \e{_{|\K|}}\in \R^{\nee}$ represents the aggregate amount of flow on each edge. Finally, we define the edge capacity vector $\flowLimVec\in \R_+^{\nee}$. If a flow on an edge exceeds its capacity, it may lead to an edge failure as described in Section \ref{emergencyPhase}.

%for each vertex $v\in \V$ and each commodity $k\in \K$, the difference between the flow of commodity $k$ into and out of vertex $v$ is equal to the netput requirement of commodity $k$ at this vertex.

Having defined the main components of the flow network, we are now ready to describe the temporal dynamics, that is, how the components behave at each time step $t\in \N$.

\subsection{Temporal dynamics} \label{subsec:temporalDynamics}
In this section, we describe the dynamics of the model at a given time step $t\in \N$. The model consists of three phases: planning, operational, and emergency. The discrete time $t$ represents the current phase of the process. The \textit{planning} phase occurs at $t = 0$ and mimics the network design process. The \textit{operational} phase occurs at $t = 1$ and reflects how the underlying network typically operates. At the end of this phase, a random disruption triggers a cascade, marking the start of the \textit{emergency} phase, which occurs at $t\geq 2$ and models the behavior of the network during the cascade. Note that at each time step, the model components defined in Section~\ref{basicComp}, such as $\demM$, $\genM$, or $\flowM$, take different values; hence, we use the superscript notation $(t)$ to mark the current time step. Here, we emphasize that all these components can be viewed as \textit{deterministic} mappings of the exogenously given \textit{probabilistic} vertex weight vector $\pareto$ and the \textit{probabilistic} cascade sequence $\cascadeSeqRand$, defined in Section~\ref{emergencyPhase}. Specifically, for $t=0$ and $t=1$, the behavior of the model depends solely on $\pareto$, while for $t\geq 2$, it also depends on the cascade sequence $\cascadeSeqRand$, which governs the progression of the emergency phase.

We assume that the flow matrix $\flowM^{(t)}\in\R^{\nee\times \nk}$ is given by a matrix function $\flowM^*$ of the netput matrix $\bm{U}^{(t)}\in\R^{\nv\times \nk}_+$ and the edge capacity vector $\flowLimVec^{(t)}\in\R^{\nee\times \nk}$. Specifically,
\begin{subequations} \begin{gather}\tag{F}\flowM^*: \R^{\nv\times \nk}_+\times \R^{\nee} \rightarrow \R^{\nee\times \nk}  \label{flow} \\
\text{s.t.:}\quad \incM \flowM^* = \bm{U}, \tag{F1} \label{flowBalance}
\end{gather}
\end{subequations}
where Constraint \eqref{flowBalance} asserts that each vertex $v$ receives precisely the amount of commodity $k$ it requires.

Furthermore, the production capacity is considered constant with respect to time $t$ and non-pathological, meaning that, for each commodity, the system always has sufficient capacity to meet the demand at time $0$. Specifically, the production capacity matrix $\genMLim$ is some function $\genMLim(\demM^{(0)}): \R^{\nv\times \nk}_+ \rightarrow \R^{\nv\times \nk}_+$ satisfying
\begin{equation}\label{ineqGenLim}\e{_{\nv}}^T \genMLim(\demM^{(0)}) \geq \e{_{\nv}}^T\demM^{(0)}.\end{equation}
\subsubsection{Planning and operational phase}
The dynamics in the planning and operational phases are similar, which is why we describe them simultaneously.  We assume that the sink matrix, representing the requirement for commodities, is the same in both the planning and operational phases, i.e.,  $\demM^{(0)} = \demM^{(1)}$. Furthermore, we assume that the requirement for commodities is proportional to the vertex sizes. Specifically, $\dem[v,k]^{(0)} = q_{v,k}\paretoI{v}$ for some $q_{v,k} \geq 0$, representing the fraction of the vertex size $v$ that requires commodity $k$. To express this compactly, we collect the coefficients $q_{v,k}$ into the matrix $\bm{Q}:= [q_{v,k}]_{v\in \V, k\in \K}\in \R_+^{\nv\times\nk}$, which encodes the relative requirement for each commodity at every vertex. Then, in matrix notation, we have that $\demM^{(0)}(\pareto) =
    \text{diag}(\pareto)\cdot \bm{Q}$, and in particular,
\begin{equation}\label{eq:defT0}\demM^{(0)}(\pareto) = \begin{pmatrix}
    q_{1,1}\paretoI{1} & \cdots & q_{1,\nk}\paretoI{1}\\
    q_{2,1}\paretoI{2} & \cdots & q_{2,\nk}\paretoI{2}\\
    \vdots &  & \vdots \\
    q_{\nv,1}\paretoI{\nv}& \cdots & q_{\nv, \nk}\paretoI{\nv} 
\end{pmatrix}.\end{equation}

 The source matrices $\genM^{(0)}$ and $\genM^{(1)}$, while not necessarily equal, are both given by some function $\genM^*$ of the sink matrix $\demM$, the production capacity matrix $\genMLim$, the edge capacity vector $\flowLimVec$, and the flow capacity slack parameter $\lambda>0$. In particular,
\begin{subequations}
\begin{gather}
\genM^*: \R^{\nv\times \nk}_+\times \R^{\nv\times \nk}_+  \times \R_+^{\nee} \times \R_+/\{0\}\rightarrow \R^{\nv\times \nk}_+  \tag{S}\label{production}\\
\text{ s.t.:}\quad\e{_{\nv}}^T \genM^*= \e{_{\nv}}^T \demM,\tag{S1}\label{prodBalance}\\
\qquad~~\genM^*\leq \genMLim, \tag{S2} \label{productionCap}     \\
\qquad~~|\flowM(\demM - \genM^*, \flowLimVec)\e{_{|\K|}}|\leq \lambda \flowLimVec.\tag{S3} \label{flowCap}
\end{gather}
\end{subequations}
Constraint \eqref{prodBalance} ensures that for each commodity $k\in \K$, the total source production is equal to the total sink requirements. %This way, the requirement at each vertex can be satisfied, by transporting the commodities from the vertices with surplus to the vertices with shortage, ensuring commodity balance in the network. 
Constraint \eqref{productionCap} states that the source production of each commodity at each vertex does not exceed \replaced{its}{the corresponding} production limit. Lastly, Constraint \eqref{flowCap} guarantees that the total flow on edges does not exceed a fraction $\lambda$ of the corresponding edge capacities. Note that if $\lambda$ is not sufficiently large, Constraints \eqref{productionCap} and  \eqref{flowCap} can be mutually exclusive. Hence, we assume that $\lambda$ is chosen such that $\genM$ exists for all values of $\pareto$. 

The source production matrix $\genM^{(0)}$ is determined as a proxy for $\genM^{(1)}$ in order to choose the operational edge capacities in the network in a manner that anticipates the edge flows in the operational phase. In the literature, such a proxy is typically given as the source production, subject to no edge capacity constraints \citep{Nesti_2020, Motter2002}. Therefore, we set $\flowLimVec^{(0)} = \infty$, $\lambda^{(0)} = 1$, and \begin{align}\begin{split}
\genM^{(0)}:&= \genM^*(\demM^{(0)}, \genMLim(\demM^{(0)}),\flowLimVec^{(0)}, \lambda^{(0)}) \\&= \genM^*(\demM^{(0)}, \genMLim(\demM^{(0)}),\infty, 1),    \label{eq:defPlanningProd}\end{split}
\end{align} as this choice implies that Constraint \eqref{flowCap} is satisfied for all matrices $\genM$, making the constraint irrelevant.

It remains to determine the operational edge flow capacity vector $\flowLimVec^{(1)}$ and the source production matrix $\genM^{(1)}$. Following the approach in \citep{Janicka_Sloothaak_Vlasiou_Zwart_2025}, the capacities (i) should satisfy some minimum capacity requirement $\overline{\mkern-3mu f}_{\min}$ and (ii) need to be large enough to sustain the flow that the graph is expected to experience in the operational phase. To satify these criteria, for every $e\in \EE$, we set \begin{equation}
    \label{flowCapacity}
    \flowLim{e}^{(1)} = \max\{\tau f_e^{(0)}, \overline{\,\mkern-4mu f}_{\min}\},
\end{equation}
where $\flowTot^{(0)}$ is the total planning flow vector, i.e., $\flowTot^{(0)} = \flowM^{(0)}\e{_{|\K|}}$. Further, $\tau\geq 1$ is the planning slack parameter and $\flowLim{\min} := \varepsilon_{\min}\sum_{v\in \V} \paretoI{v}$ for some $\varepsilon_{\min}>0$. The first criterion is met by setting the minimum capacity $\flowLim{\min}$ to scale linearly with the total vertex weights. Moreover, edge capacities scale proportionally with the planning flows, which is a standard approach in the literature \citep{Motter2002} and satisfies the second criterion. 

Finally, the operational source production matrix $\genM^{(1)}$ is given by
\(\genM^{(1)}:= \genM^*(\demM^{(1)}, \genMLim(\demM^{(1)}), \flowLimVec^{(1)}, \lambda^{(1)}).\) Note that the parameter $\lambda^{(1)}$ has a versatile modeling role. If $\lambda^{(1)}\in (0,1)$, it acts as a safety tuning parameter, ensuring that the edges are utilized at most up to a fraction $\lambda^{(1)}$ of their capacity. On the other hand, if $\lambda^{(1)}$ is chosen sufficiently large, then Constraint \eqref{flowCap} is always satisfied. This choice is suitable for applications where the edge capacity constraint is not taken into account when choosing the source production locations (see, for example, the highway traffic application detailed in Section \ref{trafficApp}).  

 The network operates in stability following the flows given by $\flowM^{(1)} = \flowM^*(\bm{U}^{(1)}, \flowLimVec^{(1)})$, until some disruption occurs, initiating a cascade of edge failures, i.e., the emergency phase.

%Matrix $\genM^{(pl)}$ is obtained by solving Problem \eqref{optProductionProblem} with sink requirements matrix $\demM$, generation limits matrix $\genMLim$, flow cost function $\flowCost$ given in \eqref{flowCost}, and edge flow capacity vector $M\e{\EE}$, where $M$ is at least as large as the total amount of sink requirements in the graph, i.e., $M\geq \sum_{v\in\V, k\in \K} \dem[v,k]$. More specifically, 
%\begin{equation}
%    \genM^{(pl)} := \genM^*\left(\demM, M\e{\EE}, \genMLim, \flowCost\right) \label{genPlan}.
%\end{equation}
%Note that the choice of the flow capacity vector virtually yields the case of unlimited flow capacity because the flow in the network can never exceed the total amount of sink requirements. 

%We assume that the network is designed to be able to withstand the flows inducted by the source productions matrix $\genM^{(pl)}$ and the sink requirements $\demM$. Hence, we set $\flowLimVec^{(pl)} = \flow^*(\genM^{(pl)}, \demM, \flowCost) := \flowM^*(\genM^{(pl)}, \demM, \flowCost)\e{\K}$,. This concludes the planning phase. 

\subsubsection{Emergency phase} \label{emergencyPhase}
In the emergency phase (${t} > 1$), the cascade failure process takes place. Each failure causes either a partial or a full edge capacity decrease. The former can be interpreted as congestion, causing a flow cost increase, while the latter is a complete component failure, rendering it unusable. This may result in flow redistribution and consequent failures, due to overloaded edges, i.e., with flow exceeding the capacity limit. This allows the failure process to propagate in the graph in the form of a cascade. At each time step, the following actions are taken:
\begin{enumerate}
    \item Update the capacity $\flowLimVec$ of edges that failed in the previous step. 
    \item If the graph is disconnected, restore the balance of source production and sink requirement (Constraint~\ref{prodBalance}) in each component. We say that a pair of vertices lies in the same component if there exists a path between them, consisting of edges with positive flow capacity. As a result of this rebalancing, any excess production or unsatisfied requirement within a component is discarded.
    \item Redistribute flows by computing the flow matrix $\flowM^*$ in the current setting, using Equation~\eqref{flow}. 
    \item Determine which \textit{overloaded} edges fail next. We say that an edge is overloaded if its relative flow exceedance, defined in \eqref{exceedance}, is greater than 1. A failure can be \textit{partial} (congestion) or \textit{complete} (edge removal). Conditional on the exceedance levels, failing edges are selected according to the per-edge probability rules $p_{e,c}$ (partial) and $p_{e,r}$ (complete), as specified in the detailed description of Step~4. If no edges fail, the cascade terminates; otherwise, return to Step~1.
\end{enumerate} 
%Before formally explaining Steps 1 -- 4, we introduce necessary notation. Let $\genM^{(t)}$, $\demM^{(t)}$, $\flowM^{(t)}$, and $\flowLimVec^{(t)}$ denote the updated source, sink, and flow matrices and edge flow capacities at time step $t$.  
%We assume that every edge has $n_e+1$ \textit{failure states}. Let $u^{(t)}_e \in \{0,1,\dots,n_e\}$ denote the failure state at time $t\in \N$. If $u_e^{(t)}\in [n_e] - 1$, the edge can fail, resulting in the increase of its failure state. Consequently, the capacity of edge $e$ decreases by a failure factor $\phi^{(t)}_e(\psi_e^{(t)}) \in [0,1)$. We assume that $\phi_e^{(t)}$ is a continuous, increasing function. \textcolor{red}{Note that the dependence on $t$ allows us to introduce different capacity impacts for initial failures, compared with failures in the later stages of the cascade. This is desired because initial failures are caused by external factors, while later failures are caused by the network's dynamic response to the initial disruptions.} The cascade cost function quantifies the effect of the cascade on the network and is denoted by $\cascadeCost$. Finally, the set of all possible cascade sequences $\cascadeSeqRand = (\cascadeSeqRand^{(t)})_{t\in \N}$ is denoted by $\cascadeSet$, where $\cascadeSeqRand^{(t)}\subseteq \EE$ is the set of edges that failed at time $t$. With this notation, we proceed to describe the exact failure mechanism. 

Next, we formally introduce the exact cascade dynamics in the emergency phase. At the end of time step $t = 1$, a cascade is initiated by a failure of one or more edges according to some probability mass function $p^{(1)}$: $\{0,\mathrm{P},\mathrm{R}\}^{\nee} \rightarrow [0,1]$. Each vector $\bm{x}^{(1)} \in \{0,\mathrm{P},\mathrm{R}\}^{\nee}$ specifies a failure configuration: $x_i^{(1)}=0$ means edge $i$ does not fail at time $t=1$; $x_i^{(1)}=\mathrm{P}$ means edge $i$ undergoes a partial failure; and $x_i^{(1)}=\mathrm{R}$ means edge $i$ is removed from the graph (complete failure). Thus, $p^{(1)}(\bm{x}^{(1)})$ is the probability of the configuration $\bm{x}^{(1)}$: precisely the edges with $x_i^{(1)}=\mathrm{P}$ suffer a partial failure, those with $x_i^{(1)}=\mathrm{R}$ are removed, and those with $x_i^{(1)}=0$ do not fail at time $t=1$. To ensure that at least one edge fails, we require that $p^{(1)}(0, \dots, 0) = 0$. This constitutes the first stage of the cascade. 

Let $\cascadeSeqRand^{(t)}=P^{(t)}\cup R^{(t)}$ denote the set of failed edges at the end of time step $t$, where $P^{(t)}$ is the set of edges that suffer a partial failure and $R^{(t)}$ is the set of edges that are removed.
 This implies that $\cascadeSeqRand^{(1)} = P^{(1)} \cup R^{(1)}$ is the set of initial edge failures under the probability law $p^{(1)}$. For $t\geq 2$, the following cascade steps are taken:

\textbf{Step 1:}\quad  For all $e\in R^{(t-1)}$, we remove edge $e$ from the graph, which is equivalent to setting the capacity of edge $e$ to 0. For all $e\in P^{(t-1)}$, the edge flow capacity is updated; that is, the current capacity $\flowLim{e}^{(t)}$ is lowered using a capacity-decrease factor. Specifically, we introduce a non-increasing, continuous function $\failFactor_e^{(t)}:\R_+\to[l_e,1]$ with $l_e\in(0,1)$, and define the relative exceedance
\begin{equation}
\label{exceedance}
\exc{e}^{(t)} := \frac{\lvert \flow[e]^{(t)}\rvert}{\flowLim{e}^{(t)}}.
\end{equation}
Using this notation, capacities are updated according to
\begin{equation}
\label{flowCapFormula}
\flowLim{e}^{(t)} :=
\begin{cases}
0, & \text{if } e\in R^{(t-1)},\\[1mm]
\failFactor_{e}^{(t-1)}\!\big(\exc{e}^{(t-1)}\big)\,\cdot \flowLim{e}^{(t-1)}, & \text{if } e\in P^{(t-1)},\\[1mm]
\flowLim{e}^{(t-1)}, & \text{otherwise.}
\end{cases}
\end{equation}
The function $\failFactor_e^{(t-1)}$ ensures \deleted{that }the capacity reduction is bigger for larger exceedance levels, while remaining \deleted{strictly }positive for partial failures (\deleted{since }$\failFactor_e^{(t-1)}\geq l_e>0$); complete removals are handled \deleted{explicitly }through membership in~$R^{(t-1)}$. 

\textbf{Step 2:}\quad If the resulting graph is disconnected, we restore the balance of requirements and production in each connected component for each commodity (see Constraint \eqref{prodBalance}), by proportionally reducing the larger of the two quantities: that is, we lower production when it exceeds the total requirement, or reduce the requirement when it exceeds the available production. In particular, for a connected component $\widetilde{\G} = (\widetilde{\V}, \widetilde{\EE})$ of $\G$, every commodity $k\in \K$, and every $v\in \widetilde{\V}$, the sink and source matrices $\demM^{(t)}$ and $\genM^{(t)}$ are obtained using the following formula:
%\begin{equation} 
%    \left(\dem[v,k]^{(t)},\gen[v,k]^{(t)}\right) = \begin{cases}\left(\dem[v,k]^{(t-1)},\eta_k^{(t-1)}\left(\widetilde{\V}\right)\cdot \gen[v,k]^{(t-1)}\right), &\text{if } \sum_{v\in \widetilde{\V}}\dem[v,k]^{(t-1)}< \sum_{v\in \widetilde{\V}}\gen[v,k]^{(t-1)},\\
%    \left(\eta_k^{(t-1)}\left(\widetilde{\V}\right)^{-1}\cdot\dem[v,k]^{(t-1)},\gen[v,k]^{(t-1)}\right), & \text{if } \sum_{v\in \widetilde{\V}}\dem[v,k]^{(t-1)}> \sum_{v\in \widetilde{\V}}\gen[v,k]^{(t-1)},\\
%    (\dem[v,k]^{(t-1)}, \gen[v,k]^{(t-1)}),&\text{otherwise,}
%    \end{cases}
%\end{equation}

\begin{equation} \label{demGenCascade}
\left(\dem[v,k]^{(t)},\gen[v,k]^{(t)}\right) =
\begin{cases}
  \left(\dem[v,k]^{(t-1)}, \eta_k^{(t-1)}\left(\widetilde{\V}\right)\cdot\gen[v,k]^{(t-1)}\right)
    & \text{if }\sum_{v\in\widetilde{\V}}\dem[v,k]^{(t-1)}
            < \sum_{v\in\widetilde{\V}}\gen[v,k]^{(t-1)},
  \\
\left(\left(\eta_k^{(t-1)}\left(\widetilde{\V}\right)\right)^{-1}\cdot \dem[v,k]^{(t-1)}, \gen[v,k]^{(t-1)}\right),
    &\text{if }\sum_{v\in\widetilde{\V}}\dem[v,k]^{(t-1)}
            > \sum_{v\in\widetilde{\V}}\gen[v,k]^{(t-1)},
  \\
  \left(\dem[v,k]^{(t-1)}, \gen[v,k]^{(t-1)}\right),
  &\text{otherwise,}
\end{cases}
\end{equation}

where $\eta_k^{(t-1)}\left(\widetilde{\V}\right) = \left(\sum_{v\in \widetilde{\V}}\dem[v,k]^{(t-1)}\right)\Big/\left(\sum_{v\in \widetilde{\V}}\gen[v,k]^{(t-1)}\right).$

\begin{comment}
\textbf{Step 2:}\quad If the resulting graph is disconnected, we restore the balance of requirements and production in each connected component for each commodity (see Constraint \eqref{prodBalance}), by proportionally lowering whichever quantity is abundant. In particular, for every connected component $\widetilde{\G} = (\widetilde{\V}, \widetilde{\EE})$ of $\G$ and every commodity $k\in \K$, the sink and source matrices $\demM^{(t)}$ and $\genM^{(t)}$ are given by
\begin{equation} \label{demGenCascade}
    \left(\dem[v,k]^{(t)},\gen[v,k]^{(t)}\right) = \begin{cases}\left(\dem[v,k]^{(t-1)},\eta_{v,k}^{(t-1)}\gen[v,k]^{(t-1)}\right)  &\text{if } \sum_{v\in \widetilde{\V}}\dem[v,k]^{(t-1)}< \sum_{v\in \widetilde{\V}}\gen[v,k]^{(t-1)},\\
    \left(\left(\eta_{v,k}^{(t-1)}\right)^{-1}\dem[v,k]^{(t-1)},\gen[v,k]^{(t-1)}\right)& \text{if } \sum_{v\in \widetilde{\V}}\dem[v,k]^{(t-1)}> \sum_{v\in \widetilde{\V}}\gen[v,k]^{(t-1)},\\
    (\dem[v,k]^{(t-1)}, \gen[v,k]^{(t-1)})&\text{otherwise,}
    \end{cases}
\end{equation}
where $\eta^{(t-1)} = [\eta_{v,k}^{(t-1)}]_{v\in \V, k\in \K}$ is a matrix satisfying $$\sum_{v\in \widetilde{\V}} \eta_{v,k}^{(t-1)} = \frac{\sum_{v\in \widetilde{\V}}\dem[v,k]^{(t-1)}}{\sum_{v\in \widetilde{\V}}\gen[v,k]^{(t-1)}}$$
for every $\widetilde{\V}$ and $k\in \K$.
\end{comment}
\textbf{Step 3:}\quad The capacity decrease and balance restoration lead to flow redistribution. The new flow matrix $\flowM^{(t)}$ is given by the function $\flowM^*$ defined in \eqref{flow}, with inputs $\bm{U}^{(t)} = \demM^{(t)} - \genM^{(t)}$ and $\flowLimVec^{(t)}$. Specifically, $$\flowM^{(t)} = \flowM^*(\bm{U}^{(t)},\flowLimVec^{(t)}).$$ 

\textbf{Step 4:}\quad For each edge $e\in\EE$, let $x=\exc{e}^{(t)}$. If $x\le 1$, then $e$ does not fail at time $t$. If $x>1$, the edge may either become congested (partial failure) or be removed (full failure). We specify a categorical failure rule with probabilities
\begin{align*}
\mathbb{P}\{\text{congestion at }e\}=p_{e,c}(x), \quad
\mathbb{P}\{\text{removal at }e\}=p_{e,r}(x),\quad
\mathbb{P}\{\text{no failure at }e\}=1-p_{e,c}(x)-p_{e,r}(x),  
\end{align*}
where $p_{e,c},p_{e,r}:\R_+\to[0,1]$ are continuous functions, satisfying $p_{e,c}(x)=p_{e,r}(x)=0$ for $x\le 1$, and $p_{e,c}(x)+p_{e,r}(x)\le 1$ for all $x$. Note that congestion or a removal of edge $e$ occurs independently of other edges. The failure outcome for each edge is independent of those for all other edges. %If for edge $e\in \EE$ the flow magnitude exceeds the corresponding flow limit, i.e., $\exc{e}^{(t)}>1$, edge $e$ may fail, following a Bernoulli random variable with probability $p_e(x)$. We assume that $p_e(x) = 0$ for $x \leq 1$ and $p_e$ is a continuous, non-decreasing function. This implies that if the relative flow exceedance $\exc{e}^{(t)}\leq 1$, edge $e$ does not fail at time $t$; otherwise it fails with probability $p_e(\exc{e}^{(t)})$,  which increases with the level of exceedance.

The cascade terminates once no edge failures (either congestion or removal) occur during a given cascade stage. %Recall that edge failure can result in either a capacity decrease or to a complete loss of an edge (see also Equation~\eqref{flowCapFormula}), and thus an edge may fail multiple times. 
For simplicity, we assume that each edge can fail at most $n_e\in \N$ times. This means that in Step 4, we only consider edges that failed less than $n_e$ times.  This assumption ensures that the cascade always terminates in finite time and implies that the set of all possible cascade sequences, defined as $\cascadeSet := \left\{\cascadeSeqDet = \left(\cascadeSeqDet_P^{(t)}\cup \cascadeSeqDet_R^{(t)}\right)_{t\in \N}~:~ \cascadeSeqDet_P^{(t)}, \cascadeSeqDet_R^{(t)}\subseteq \EE, \cascadeSeqDet_P^{(t)}\cap \cascadeSeqDet_R^{(t)} = \emptyset \right\}$ is finite. 

The effect of the cascading failure process on the behavior of the system is denoted by the cascade cost $\cascadeCost\in \R$, which is a function of the flow network before and after the cascade. \textcolor{black}{Given the structure of our model, this means that $\cascadeCost$ is fully determined by two random variables: the vertex weight vector $\pareto$ and the cascade sequence $\cascadeSeqRand$. Thus, following~\eqref{conditionalRV}, $\cascadeCost$ given $\pareto = \bm{x}$ and $\cascadeSeqRand = \cascadeSeqDet$ is denoted by $\detZ(\bm{x}, \cascadeSeqDet)$.  }

\begin{table}[p!]
    \centering
    \renewcommand{\arraystretch}{1.16}
    \resizebox{\textwidth}{!}{
    \begin{tabular}{|c|l|}
    \hline
   \rowcolor{gray!15}
   \textbf{Variable} & \textbf{Description}\\
\hline\hline
    \multicolumn{2}{|c|}{\textbf{Network Structure}} \\ \hline
         $\G = (\V,\EE)$ & A connected directed graph representing the underlying network.\\
    \hline
         $\V$ & The set of vertices of $\G$.\\
    \hline
         $\EE$ & The set of edges of $\G$. \\

    \hline
         $\incM$ & The incidence matrix of graph $\G$. \\
    \hline
         $\paretoI{v}$ & The size of vertex $v\in \V$.\\
    \hline
         $\pareto = (\paretoI{v})_{v\in\V}$ & Vector of vertex sizes. \\
    \hline\hline
    \multicolumn{2}{|c|}{\textbf{Requirements and Production of Commodities}} \\     \hline
         $\K$ & Set of commodities that are transported in the network. \\    \hline
         $\dem[v,k]\in \R_+$ & Amount commodity $k\in \K$ required at the sink vertex $v\in \V$.\\
    \hline
         $\gen[v,k]\in \R_+$ & Amount of commodity $k\in \K$ produced at the source vertex $v\in \V$.\\
             \hline
        $q_{v,k}$ & Fraction of vertex $v\in\V$ that requires commodity $k\in \K$, with $q_{v,k}\geq 0$. \\
    \hline
         $\demM := [\dem[v,k]]_{v\in \V,k\in\K}$ & Sink matrix of size $\nv\times \nk$.\\
    \hline
         $\genM := [\gen[v,k]]_{v\in \V,k\in\K}$ & Source matrix of size $\nv\times \nk$.\\
                 \hline $\bm{Q} = [q_{v,k}]_{v\in\V,k\in\K}$ & Matrix encoding the relative requirement for each commodity at every vertex (see also \eqref{eq:defT0}). \\
    \hline
         $\genMLim := [\genLim[v,k]]_{v\in \V,k\in\K}$ & Production capacity matrix of size $\nv\times \nk$ with entries $\genLim[v,k]\geq 0$. \\   
    \hline $\bm{U} = [u_{v,k}]_{v\in\V, k\in \K}$ & Netput matrix $\bm{U} = \demM - \genM$, with netput requirement entries $u_{v,k} := \dem[v,k] - \gen[v,k]$. \\

             \hline\hline
    \multicolumn{2}{|c|}{\textbf{Flows and Flow Capacities}} \\
    \hline
         $\flow[e,k]$ & Flow of commodity $k\in \K$ on edge $e\in \EE$ with $\flow[e,k]\in \R$.\\
    \hline
         $\flowM := [\flow[e,k]]_{e\in\EE,k\in \K}$ & Flow matrix of size $\nee\times \nk$.\\

    \hline
         $\flow[e] := (\flowM\e{_{|\K|}})_e$ & Total flow on edge $e\in \EE$, i.e., the sum of flows of each commodity on edge $e$. \\
    \hline
         $\flowTot:= (\flow[e])_{e\in\EE}$ & Vector of total edge flows.\\
    \hline
         $\flowLim{e}$ & The flow capacity of edge $e\in \EE$.\\
    \hline
         $\flowLimVec :=(\flowLim{e})_{e\in\EE}$ & Vector of edge flow capacities. \\
             \hline
        $\tau$ & Planning slack parameter used to determine operational edge flow capacity (see \eqref{flowCapacity}); $\tau\geq 1$. \\
        \hline 
        $\flowLim{\min}$ & Minimal operational edge flow capacity (see \eqref{flowCapacity}).\\
             \hline\hline
    \multicolumn{2}{|c|}{\textbf{Flow and Production Mechanisms}} \\
        \hline
        $\flowMOpt{\bm{U}, \flowLimVec}$ & Function describing the flow distribution mechanism (see Problem \eqref{flow} and \eqref{optFlow}).\\
    \hline
         $\flowCost(\flowM, \flowLimVec)$ & Flow cost function of Problem \eqref{optFlow}, $\flowCost: \R^{\nee\times\nk}\times \R^{\nee}_+\rightarrow \R_+$. \\
             \hline
        $\genM^*(\demM,\genMLim, \flowLimVec, \lambda)$ & Function describing the production allocation mechanism (see Problem \eqref{production} and \eqref{optProd}).\\
    \hline
        $\lambda$ & Flow capacity slack parameter of \eqref{flowCap}; $\lambda > 0$.\\
                     \hline\hline
    \multicolumn{2}{|c|}{\textbf{Emergency Phase Variables and Functions}} \\
    \hline 
        $p^{(1)}$ & Initial failure probability function; $p^{(1)}: \{0,\mathrm{P}, \mathrm{R}\}^{\nee}\rightarrow [0,1]$.\\
    \hline 
    $\failFactor_e$ & Function that yields the capacity decrease of a failed edge $e\in \EE$.\\
    \hline
        $\exc{e}$ & Relative flow exceedance on edge $e\in \EE$; see \eqref{exceedance}.\\
    \hline
        $\eta_k$ & Production-requirement rebalancing factor for commodity $k$ (see \eqref{demGenCascade}).\\
    \hline 
            $p_{e,c}$ & Function describing the \textit{congestion} probability of edge $e\in \EE$; $p_{e,c}: \R_+\rightarrow [0,1].$\\
    \hline 
                $p_{e,r}$ & Function describing the \textit{removal} probability of edge $e\in \EE$; $p_{e,r}: \R_+\rightarrow [0,1].$\\
    \hline 
        $n_e$ & Maximum number of failures of edge $e\in \EE$. \\
    \hline 
    $\cascadeSet$ & Set of all possible cascade sequences.\\
    \hline 
$\cascadeSeqRand$ & Random variable describing the cascade sequence.\\
    \hline 

        $\cascadeCost$ & Cascade cost, given by either \textit{generalized cascade size} function \eqref{generalizedFailureCost} or \textit{cascade flow cost} function \eqref{deltaFlowCost}.\\
        \hline 
        $\detZ(\bm{x}, \cascadeSeqDet)$ & Deterministic cascade cost given $\pareto = \bm{x}$ and $\cascadeSeqRand = \cascadeSeqDet$.\\
                  \hline\hline
    \multicolumn{2}{|c|}{\textbf{Other notation}} \\
    \hline

         $\e{i}$ & All-ones vector of size $i\in \N$.\\        \hline
         $\e{i,n}$ & $i$-th unit vector of size $n$.\\
\hline

    \end{tabular}}
    \caption{Summary of notation. Note that some quantities may differ, depending on the current phase of the model. If that is the case, we use the superscript notation $(t)$, $t\in \N\cup \{0\}$, to specify the relevant phase. Notational conventions are provided in Section~\ref{notation}.}
    \label{tab:my_label}
\end{table}

\section{Emergence of Pareto-tailed cascade cost} \label{mainRes}
The objective of this paper is to find a universal explanation for the emergence of Pareto-tailed cascade costs $\cascadeCost$ in flow networks with overload cascades. Section~\ref{modelDescription} introduced a general model that has broad applicability across real-world flow networks. In this section, we state and derive asymptotic results on the cascade cost for this model. The results characterize the probability that the cascade cost $\cascadeCost$ exceeds $y$, as $y$ approaches infinity, and identify the most likely conditions that lead to high cascade costs. 

To ensure the analytical tractability of the cascade cost \( \cascadeCost \), two additional assumptions are required---one on the structure of \( \cascadeCost \), and one on the probability of observing a particular cascade \( \cascadeSeqDet\in \cascadeSet\). While these are technical assumptions that guarantee appropriate continuity and scalability, they are not restrictive. In Section~\ref{sec:satisfiability}, we demonstrate that many standard choices for the flow distribution mechanism \( \flowM^* \), source production \( \genM^* \), source production capacity \( \genMLim \), and cascade cost \( \cascadeCost \) satisfy the required conditions. %This includes a wide range of classical formulations used to model versatile flow networks, as illustrated with examples in Section~\ref{Applications}. These assumptions are stated below.

%The objective of this paper is to show how Pareto-tailed cascade costs $\cascadeCost$ can emerge in various flow networks with overload cascade. Hence in this section, we derive asymptotic results on the cascade cost in our model. These results characterize the probability that the cascade cost exceeds $y$, as $y$ approaches $\infty$, and determine the most likely circumstances that lead to high congestion costs. Before stating and proving the results, we first introduce an assumption, under which the results are derived.

%So far, in Section~\ref{modelDescription}we introduced our model in full generality to guarantee that the model is broadly applicable to many real-world flow networks. However, ensuring that the cascade cost is analytically tractable requires additional assumptions on functions $\flowM^*$, $\genM^*$, $\genMLim$, and $\cascadeCost$, under which we derive our results. 

%In order to state the assumptions, we first need to introduce key definitions: scale invariance, sequential right continuity and $\varepsilon$-weight allocation. Informally, a function is scale-invariant, if changing the scale does not change the function. The following definition makes this notion precise.

\begin{assumption}
    The flow distribution mechanism $\flowM^*$, source production $\genM^*$, source production capacity $\genMLim$, and cascade cost $\cascadeCost$ are chosen such that: \begin{enumerate}
        \item[\hypertarget{assumptionPart1}{1.}] For some $\boundZ>0$ and $\delta>0$, the random variable $\cascadeCost$ satisfies $\cascadeCost\leq \boundZ\cdot \sum_{v\in \V} \left(\paretoI{v}\right)^{\delta}$. \textcolor{black}{Moreover, given any realization $\pareto = \bm{x}$ and a particular cascade $\cascadeSeqRand = \cascadeSeqDet\in \cascadeSet$, the conditional cascade cost $\detZ(\bm{x}, \cascadeSeqDet)$ is $\delta$-scale-invariant and SRC w.r.t.\ the vertex weight vector $\bm{x}$}. 
        \item[\hypertarget{assumptionPart2}{2.}] The probability that a particular cascade $\cascadeSeqRand = \cascadeSeqDet$ occurs, given the vertex weight vector $\pareto$, is SRC w.r.t.\ $\pareto$ and equal for all scales $\omega$, i.e., $\forall~\omega>0$,
        \(\PR{\cascadeSeqRand = \cascadeSeqDet~|~\omega \pareto} = \PR{\cascadeSeqRand = \cascadeSeqDet~|~ \pareto}.\)
    \end{enumerate} \label{assumption}
\end{assumption}

Recall that the notions of $\delta$\textit{-scale invariance} and \textit{SRC} are defined in Section~\ref{notation}. Although Assumption~\ref{assumption} is satisfied by many functions, verifying this assumption is not always straightforward. To address this, in Section~\ref{sec:satisfiability}, we discuss various commonly used choices for functions and derive results on the satisfiability of Assumption~\ref{assumption} for these choices. This facilitates a straightforward application of our model and the asymptotic results in a wide range of settings, which we demonstrate with real-world examples in Section~\ref{Applications}.

In what follows, we state the main results of this paper, showing how Pareto-tailed cascade costs emerge in general flow networks as described in Section~\ref{modelDescription}. First, in Proposition \ref{sbj}, we identify which realizations of the vertex weight distribution are the most likely to cause a high cascade cost. %Then, in Theorem \ref{mainThm}, based on the insights from Proposition \ref{sbj}, we present an asymptotic result on the right-tail probability of a large cascade cost. 
In particular, we show that, with high probability, large cascade costs occur when one vertex has a significantly larger weight compared to all other vertices. Because a large disproportion between vertex weights is required, we refer to this result as the \textit{catastrophe principle}, as is commonly used in extreme value theory \citep{nair_wierman_zwart_2022}. Afterwards, we state the main result on the asymptotic behavior of the cascade cost $\cascadeCost$.
\begin{proposition}[Catastrophe principle]
Suppose that Assumption \ref{assumption} holds for some $\delta>0$. Fix $\varepsilon>0$, and let $\paretoI{\max} $ be the largest weight amoung all vertices, i.e., $\paretoI{\max} = \max_{1\leq i\leq \nv}\{\paretoI{i}\}$. Then,
\[\PR{\cascadeCost>y, \sum_{i = 1}^{\nv} \left(\paretoI{i}\right)^\delta > (1+\varepsilon)\left(\paretoI{\max}\right)^\delta} = \OO{y^{-2\alpha/\delta}} \quad \text{as }y\rightarrow \infty, \]
 where $\alpha>0$ is the tail parameter of $\pareto$ as given in \eqref{def:pareto-tailed}.\label{sbj}
\end{proposition}
In the proof of this result, we rely on the fact that the cascade cost $\cascadeCost$ is always bounded by $\boundZ\sum_{v\in \V} \left(\paretoI{v}\right) ^\delta$ for some $\boundZ>0$, and we use well-known bounds on the order statistics of Pareto-tailed distributions. The proof is provided in Appendix~\ref{EC:propProof}.

%Finally, we are ready to state and prove the main theorem that specifies the behavior of the probability that the cascade cost $\cascadeCost$ exceeds $y$ when $y\rightarrow \infty$. 
 
\newpage\begin{theorem}[Tail of the cascade cost]\label{mainThm} Suppose that Assumption \ref{assumption} holds for some $\delta>0$. If \textcolor{black}{ there exists a cascade $\cascadeSeqDet\in \cascadeSet$} such that \textcolor{black}{for a vertex weight vector $\bm{x} = \e{i,\nv} = (0,\dots,0,1,0, \dots, 0)$ for some $i\in \V$, the cascade $\cascadeSeqDet$:}
\begin{enumerate}
    \item occurs with positive probability, i.e., $\textcolor{black}{\PR{\cascadeSeqRand = \cascadeSeqDet \mid \pareto = \bm{x}}}>0$, and
    \item has a positive cost, \textcolor{black}{i.e., $\detZ(\bm{x}, \cascadeSeqDet)>0$,}
\end{enumerate}
then the cascade cost $\cascadeCost$ has a scale-free tail with parameter ${\alpha}/{\delta}$. Specifically,
\begin{equation}
    \label{thmCascadeSizeEq} \PR{\cascadeCost>y}\sim l_\cascadeCost\cdot y^{-\alpha/{\delta}}, \quad y\rightarrow \infty,
\end{equation}
with 
\begin{equation}
    \label{constantCS} l_\cascadeCost = \sum_{i = 1}^{\nv} \sum_{\cascadeSeqDet\in \cascadeSet}K \cdot  \PR{\cascadeSeqRand = \cascadeSeqDet~|~\pareto = \e{i,{\nv}}}\cdot \detZ(\e{{i,{\nv}}},\cascadeSeqDet)^{\alpha/\delta},
\end{equation}
where $K$ and $\alpha$ are given in Equation~\eqref{def:pareto-tailed}. Otherwise, 
\(\PR{\cascadeCost>y} = \OO{y^{-2\alpha/\delta}},  \text{ as }y\rightarrow \infty.\)
\end{theorem}
The theorem quantifies the tail probability of the cascade cost in terms of the cascade costs for scenarios where a single vertex has weight 1 and all other vertices have weight 0. This arises as a consequence of the catastrophe principle, stated in Proposition~\ref{sbj}. The $\delta$-scale invariance and the SRC assumptions also play an important role, ensuring that the limit behaves well and does not depend on~$y$. 

The proof of the theorem is provided in Appendix~\ref{EC:thmProof} and consists of the following steps:
\begin{enumerate}
    \item[\hypertarget{stepA}{\textbf{A)}}] Decompose the probability in \eqref{thmCascadeSizeEq} by considering all possible cascade trajectories $\cascadeSeqDet\in \cascadeSet$ and all possible vertices with the largest weights. 
    \item[\hypertarget{stepB}{\textbf{B)}}] For a fixed $\varepsilon>0$, split the probability into two cases, by distinguishing whether or not $\sum_{i = 1}^{\nv} \paretoI{i}^\delta\leq (1+\varepsilon)\max_{1\leq i\leq \nv}\{\paretoI{i}^\delta\}$, and apply Proposition \ref{sbj}.  
    \item[\hypertarget{stepC}{\textbf{C)}}] Rewrite the probability, by conditioning first on the value of the largest vertex weight, then on the remaining vertex weights, and finally on the occurrence of cascade $\cascadeSeqDet\in\cascadeSet$. 
    \item[\hypertarget{stepD}{\textbf{D)}}] By Assumption~\ref{assumption}, use the scale-invariance and SRC properties of the cascade cost and the conditional cascade probability.
    \item[\hypertarget{stepE}{\textbf{E)}}] Construct upper and lower bounds for the probability of a large cascade cost, as a function of $\varepsilon$ and $y$.
    \item[\hypertarget{stepF}{\textbf{F)}}] Show that these bounds are equal after taking $\lim y\rightarrow \infty$ and $\varepsilon \downarrow 0$ and conclude the final result.
\end{enumerate}

\section{Frameworks for cascading flow models} \label{sec:satisfiability}
Assumption~\ref{assumption} outlines the general conditions necessary for our results to hold. While these conditions provide a mathematically rigorous description, they offer limited intuition about the class of flow networks for which \replaced{this}{the} assumption may hold. In this section, we describe this class in more detail. In particular, we discuss common choices for the flow\deleted{ matrix}, source\deleted{ matrix}, and cost \replaced{matrices}{matrix} in Sections~\ref{subsec:FlowFunctions}--\ref{subsec:CascadeCost}, and in Section~\ref{subsec:classesSatisfy}, we \replaced{show}{demonstrate} that they satisfy Assumption~\ref{assumption}. Section~\ref{Applications} then presents concrete applications illustrating these ideas in practice.
\subsection{Flow matrices and cost functions} \label{subsec:FlowFunctions}
A well-known principle in flow networks asserts that the flow matrix function $\flowM^*$ typically minimizes the cost or the loss of some form of energy in the network, e.g., minimizing power loss in power networks, or minimizing travel time in traffic networks. Following this principle, the minimal-cost flow matrix $\flowM^*$ can be defined as an optimal solution to:
    \begin{equation} \begin{aligned}\flowM^*\left(
    \bm{U}, \flowLimVec
    \right)\in  &\argmin_{\flowM\in \textcolor{black}{\mathcal{D}}}\flowCost(\flowM, \flowLimVec): \\
    \incM &\flowM = \bm{U},\label{optFlow}\end{aligned}\end{equation}
where $\mathcal{D}$ denotes the optimization domain and $c_f$ \deleted{represents }the flow cost function. \textcolor{black}{We choose $\mathcal{D} = \R^{|\EE|\times|\K|}$ when flows against the assigned edge direction are allowed and $\mathcal{D} = \R_+^{|\EE|\times|\K|}$ otherwise. The cost function $c_f$ can depend on the flow matrix,} as well as the flow capacity. The choice of the cost function determines the\deleted{ fundamental} properties of the resulting network flows; thus, in the following, we review several \replaced{common}{commonly used} formulations.

A well-established cost function in network control and optimization theory to model flow in physical systems~\citep{Whittle_2007} is given by:%Typical assumptions where $c_f$ is a $\delta$-scale-invariant, continuous, and strictly convex function with respect to both $\flowM$ and $\bm{U}$. 
 \begin{equation}\flowCost\textcolor{black}{(\flowM, \flowLimVec)} := \sum_{e\in \EE}\left( \frac{1}{\beta} a_e b_e (|\flow[e]|/a_e)^\beta\right), \label{flowCost}\end{equation}
    where $\beta\geq 1$ and $a_e,b_e> 0$ for all $e\in \EE$. %{\color{red} Below needs to be explained in more detail. What are the parameters representing? Look at the Whittle book.} 
    This cost function is composed of cost contributions from each edge. Depending on the underlying application, these parameters can have different interpretations; however, from a physical perspective, one can associate $b_e$ with the length or weight of edge $e$, and $a_e$ with the cross-section of $e$, making the factor $a_eb_e$ the total volume of the edge.  The cost increase on edge $e$ is primarily driven by the flow density $|\flow[e]|/a_e$ raised to the power of $\beta$, which represents the cost of maintaining this flow density per unit of edge volume. This formulation implicitly asserts that it does not matter which commodities are transported through the edge, only their cumulative magnitude. 

   Common choices for the parameter $\beta$ are $\beta = 1$ and $\beta = 2$. The latter choice penalizes high utilization of edges, reflecting the fact that energy dissipation due to flow resistance in many systems grows superlinearly. Examples of such systems include electrical networks \citep{Montoya2019, Momoh1999}, water distribution systems \citep{Labadie2004}, and communication networks \citep{ahuja1993network}. In the setting\deleted{ of} $\beta = 2$ and $\mathcal{D} = \R^{|\EE|\times|\K|}$, the minimum-cost flow problem is well-understood, and, since $b_e>0$ for all $e\in \EE$, it has a unique total flow solution, namely $\flowTot^*$ is a linear function of $\bm{U}$ \cite[ch. 6]{boyd2004convex}. The choice of $\beta = 1$ yields a linear flow cost function. In this setting\textcolor{black}{, for both choices of domain $\mathcal{D}$}, the problem is \deleted{closely }related to shortest path flow, where the \deleted{length of an }edge\added{ length} is given by $b_e$; however, the latter is typically defined for a single source-sink pair. This problem may have a non-unique optimal solution, which can be problematic in some models. In \textcolor{black}{such a case}, additional conditions might be necessary to ensure the uniqueness of $\flowM^*$.
    
    Another typical cost function models scenarios where the flow cost \replaced{depends on}{is influenced by} the flow capacity and where the edges are one-directional, i.e., \textcolor{black}{$\mathcal{D} = \R_+^{|\EE|\times|\K|}$}, and is given by \citep{us1964traffic, 1}
        \begin{align}\begin{split}
        \label{wardrop} 
        c_f\textcolor{black}{(\flowM, \flowLimVec)}:= \sum_{e\in \EE}\int_0 ^{f_e}\left(d_e + b_e (x/\flowLim{e})^{\beta - 1} \right) dx = \sum_{e\in \EE}\left(d_e \flow[e] +\frac{1}{\beta} b_e \flowLim{e}(\flow[e]/\flowLim{e})^\beta \right), \quad d_e,b_e> 0,\beta> 1.\end{split}
    \end{align} 
     This setting is particularly relevant in traffic models, where Problem~\eqref{optFlow} with cost function \eqref{wardrop} is known as the Wardrop User Equilibrium \added{(UE) }\citep{Wardrop1952, Correa2011}.
    Compared to \eqref{flowCost}, the factor $a_e$ is replaced by $\flowLim{e}$, which leads to a high flow cost for flows that exceed the capacity on some edges. Moreover, this cost function has an additional linear term $d_e f_e$, representing the cost of free-flow travel, while the second term represents the additional travel cost due to congestion. Note that, \textcolor{black}{under  $\mathcal{D}= \R_+^{|\EE|\times|\K|}$, the flow problem remains feasible} because we assume that the graph is connected in a strong sense, implying the existence of a path between each pair of vertices. 

    Another common flow problem in traffic networks is called System Optimum \added{(SO)}. Contrary to the Wardrop \replaced{UE}{User Equilibrium}, where network flows arise from the selfish behavior of the drivers, the \replaced{SO}{System Optimum} flow matrix minimizes the cumulative flow cost in the network. System Optimum is another instance of \eqref{optFlow} with cost function
    \begin{align}\begin{split}
        \label{eq:SO}
        c_f\textcolor{black}{(\flowM, \flowLimVec)}:= \sum_{e\in \EE}f_e\left(d_e + b_e (\flow[e]/\flowLim{e})^{\beta - 1} \right) = \sum_{e\in \EE}\left(d_e \flow[e] +b_e \flowLim{e}(\flow[e]/\flowLim{e})^\beta \right), \quad d_e,b_e> 0,\beta> 1.
    \end{split}\end{align}
 We observe that \eqref{wardrop} and \eqref{eq:SO} differ only by a factor of $\beta$ in the second term, which means that\deleted{, without loss of generality,} System Optimum can be obtained using the cost function in \eqref{wardrop} by absorbing the factor $1/\beta$ into the parameter $b_e$.

As previously mentioned, the uniqueness of solutions to \eqref{optFlow} is an important notion. If $\beta>1$, then the cost functions \eqref{flowCost}--\eqref{eq:SO} are strictly convex functions in the total flow $\flowTot$, yielding a unique total flow vector $\flowTot^*$. However, when several commodities are considered, the optimal flow matrix $\flowM^*$ need not be unique. In our model, we focus on the total flow $\flowTot^*$; hence, the ambiguity of $\flowM^*$ is not problematic. Nevertheless, for completeness, we adopt the convention that $\flowM^*$ is the optimal solution that distributes each commodity uniformly over all possible paths. 
    %{\color{red} Paragraph needs to be written more to the point. Two typical choices are beta is one or two. The latter is often used for what and why? beta is one is relatively well-explained. Basically, you put two goals in one paragraph; uniqueness, and common choices of parameter beta. Moreover, is this only for the Wardrop case, or for the general thing in (21)?} Common choices for parameter $\beta$ is $\beta = 2$, which yields a quadratic cost function of the total flow $\flowTot$. In this setting, the minimum-cost flow problem is well-understood, and has a unique total flow solution, namely $\flowTot^*$ is a linear function of $\bm{U}$ \textcolor{red}{REF}.  If $\beta>1$, the cost function is strictly convex in $\flowTot$, yielding a unique total flow vector $\flowTot^*$. Note, however, that if several commodities are considered, the optimal flow matrix $\flowM^*$ need not be unique. For example, if two commodities are transported between vertices $i$ and $j$ and the optimal total flow $\flowTot^*$ utilizes two distinct $i\rightarrow j$ routes, then there are infinitely many distributions of the two commodities over the two routes that yield the optimal total flow. Typically, this non-uniqueness is not a problem as the quantity of interest in most models is the total flow vector. If uniqueness is required, again, additional assumptions about the distributions of commodities over the routes might be necessary.  
\subsection{Source matrices}\label{subsec:SourceFunctions}
In this section, we consider typical modeling choices for the source matrix.  In the literature, production modeling often follows either a centralized or a decentralized approach. In the centralized case, a single decision-maker determines all production quantities to minimize the total cost of meeting demand. By contrast, in the decentralized setting, each production node acts independently, making local decisions \citep{Saharidis11}. In what follows, we discuss the modeling of these two paradigms. 

In centralized optimization, a common approach is to minimize the total production cost subject to generation and flow constraints, mirroring the structure of flow optimization in Section~\ref{subsec:FlowFunctions}. Following this idea, we define the minimal-cost source matrix $\genM^*$ as 
\begin{subequations} \label{optProd} \begin{gather} \genM^*(\demM, \genMLim, \flowLimVec, \lambda) := \argmin_{\genM \in \R^{\nv}_+}  \sum_{v \in \V} \sum_{k \in \K} \frac{1}{\gamma} c_{v,k} \gen[v,k]^\gamma  \label{optProdCost} \\ \qquad\qquad\qquad\quad\text{s.t.:} \quad \e{_{\nv}}^T \genM =  \e{_{\nv}}^T \demM,\label{optProdBalance} \\ \qquad\qquad\qquad\quad\genM \leq \genMLim, \label{optProdConst}\\ \qquad\qquad\qquad\quad |\flowM^*(\demM - \genM, \flowLimVec)\e{_{|\K|}}| \leq \lambda \flowLimVec, \label{optFlowConst}\end{gather}  \end{subequations} where $c_{v,k}>0$ and $\gamma>1$. Here, $c_{v,k}$ denotes the cost coefficient associated with generating commodity $k$ at node $v$, capturing the cost of production resources and efficiency. The exponent $\gamma$ controls the marginal production cost; values $\gamma>1$ model increasing marginal costs and ensure strict convexity.

Problem \eqref{optProd} is a convex instance of the resource allocation problem \citep{katoh2024resource}, and arises in various applications \citep{PATRIKSSON20081}, including optimal power flow in transmission systems, where $\gamma = 2$ is commonly used \citep{MONTOYA201918} . The convexity \deleted{of the cost function }ensures uniqueness and supports tractable analysis.

The optimal resource allocation problem includes the optimal flow matrix $\flowM^*$ in one of the constraints. This means that Problem \eqref{optProd} is an instance of a bilevel optimization problem, where \eqref{optProd} is the upper level problem and \eqref{optFlow} is the lower level problem. \replaced{Such}{These types of} optimization problems are generally complex~\citep{Dempe_Kalashnikov_Pérez-Valdés_Kalashnykova_2015}, where the existence and uniqueness of optimal solutions for the upper level problem is only guaranteed under strict assumptions on the lower level problem. For our cascade model to be well defined, it is essential that \replaced{a unique}{the} solution of \eqref{optProd} exists\deleted{ and is unique}, allowing us to study $\genM^*$ as a function of the model input. A sufficient condition for uniqueness is \deleted{the }convexity of the feasible region and \deleted{the }strict convexity of the cost function~\citep{boyd2004convex}. The latter is guaranteed by our choice of the cost function; hence, it remains to ensure the former condition. \added{Since }Constraints~\eqref{optProdBalance} and \eqref{optProdConst} are linear in $\genM$\deleted{; therefore}, the only possible source of nonconvexity \replaced{in}{of} the feasible region \replaced{is}{arises from} Constraint~\eqref{optFlowConst}. \replaced{The feasible region is convex if}{It turns out that a sufficient condition on} $|\flowTot^*|$ is quasi-convex\deleted{ity} because \added{the }sublevel sets of \replaced{quasi-convex}{such} functions are convex\deleted{, resulting in a convex feasible region} \citep{boyd2004convex}. The function $|\flowTot^*|$ is quasi-convex if for every $\bm{U_1}$, $\flowLimVec_1$, $\bm{U_2}$, $\flowLimVec_2$, and $\xi\in [0,1]$, 
\begin{equation}|\flowTot^*(\xi\bm{U_1} + (1-\xi)\bm{U_2}, \xi\flowLimVec_1 + (1-\xi)\flowLimVec_2)|\leq \max\{|\flowTot^*(\bm{U_1},\flowLimVec_1 )|,|\flowTot^*(\bm{U_2},\flowLimVec_2 )| \}.\label{eq:quasiFlowTot}\end{equation}
Assuming \eqref{eq:quasiFlowTot} holds, we have that if $\genM_1$, $\genM_2$ both satisfy \eqref{optFlowConst}, then for any $\eta\in [0,1]$,
\begin{equation*}|\flowTot^*(\demM - \eta\genM_1 - (1 - \eta)\genM_2, \flowLimVec)|\leq \max\{|\flowTot^*(\demM - \genM_1,\flowLimVec )|,|\flowTot^*(\demM - \genM_2,\flowLimVec)| \}\leq \lambda \flowLimVec,\end{equation*} showing that the feasible set is \added{indeed }convex.

Unfortunately, the quasi-convexity of $|\flowTot^*|$ does not always hold, as we illustrate for the complete graph on four vertices in Example~\ref{exampleNonQuasi} in Appendix~\ref{appendixExample}. This example shows that the uniqueness of the optimal source production $\genM^*$ cannot be guaranteed in full generality. This is why in some of our satisfiability results (see Proposition~\ref{propClassFunc1} presented in Section~\ref{subsec:classesSatisfy}), it is necessary to limit our results to particular choices of the parameter~$\beta$.

Next, we turn to the case of the decentralized problem, which arises in many applications where production is determined by the uncoordinated actions of multiple agents rather than by a centralized planner. Such decentralized behavior is typical in systems like traffic and communication networks, where the source production represents the total volume of traffic or data packets generated independently by users at different locations, who generally act selfishly. As a result, there is no reason to assume that the corresponding source production satisfies the flow capacity constraint~\eqref{flowCap}. Similarly, individual contributions to the total source production are not restricted by a centralized limit, allowing for natural variability based on user activity and demand. Therefore, in these settings, it is natural to assume that the production and flow capacity constraints are absent or not enforced, that is, $\genMLim=\infty$ and $\lambda = \infty$.

In the literature, decentralized production is modeled in several ways. In some works, production at each source is taken as an exogenous input reflecting, for example, fixed demand or historical usage patterns \citep{Korilis97}. Other models treat production as the outcome of individual agents’ optimization problems, where each agent maximizes a \deleted{local }utility or profit subject to their own constraints \citep{kelly98}. More generally, production levels may arise from dynamic or game-theoretic mechanisms in which agents adjust their output in response to prices, congestion signals, or competition \citep{kelly98, Johari04}.

Here, we consider the first approach. Under the assumptions that $\genMLim=\infty$ and $\lambda = \infty$, the source matrix $\genM^*$ can be modeled directly as a given function of $\demM$ satisfying Constraint~\eqref{optProdBalance}, without solving \eqref{optProd}. A simple and natural choice of $\genM^*$ is a linear mapping of the form:
%So far, we have focused on optimal source matrices that arise as solutions to the resource allocation problem. However, in several applications, the source production cannot be optimized by a single entity. Instead, it arises naturally from the decisions of many individuals. This decentralized behavior is common in systems such as traffic or communication networks, where the source production corresponds to the total volume of traffic or data packets originating from a given location, generated by many users of the network. In such cases, it is more natural to assume that the source production $\genM^*$ is a given function of $\demM, \genMLim, \flowLimVec$, and $\lambda$, satisfying Constraints \eqref{optProdBalance} -- \eqref{flowCap}. {\color{red} How is this different than your first class? In principle, that is also a given function, determined a solution of an optimization problem. Ok, basically, you wish to say that there are several settings where it is natural to assume that there is no limit on production and flow capacity constraints. This is the case, for example, in systems where individuals only behave selfishly. Then, the problem reduces to the setting you desribe below. Change this wording to this story then.} For example, in scenario with no production and flow capacity constraints, i.e., $\genMLim = \infty$ and $\lambda = \infty$ , one could consider $\genM^*$ to be the following linear function of $\demM$, 
\begin{equation}\genM^* =\bm{R}\cdot\text{diag}\left(\e{_{\nv}}^T\demM\right), \quad \bm{R}\in \R_+^{\nv\times\nk}, \quad\e{_{\nv}}^T\bm{R} = \e{_{|\K|}}^T.\label{eq:linearSource}\end{equation}
Here, the entry $r_{v,k}$ of $\bm{R}$ can be interpreted as the fraction of the total sink requirement for commodity $k$ that is produced at source $v$.
The assumption on $\bm{R}$ implies that Constraint~\eqref{optProdBalance} holds because
\[\e{_{\nv}}^T\genM^* = \e{_{\nv}}^T\bm{R} \cdot \text{diag}(\e{_{\nv}}^T\demM) = \e{_{|\K|}}^T \cdot \text{diag}(\e{_{\nv}}^T\demM) = \e{_{\nv}}^T\demM.\]%Since source production is induced by individual users who do not coordinate their efforts, there is no reason to assume that the resulting flow satisfies the flow capacity constraint \eqref{flowCap}. This is why $\lambda = \infty$ is a valid assumption in this setting. Similarly, the assumption of $\genMLim = \infty$ is reasonable because, in decentralized systems like these, individual contributions to the total source production are not typically restricted by a centralized limit, allowing for natural variability based on user activity and demand.

Finally, note that the linear form in~\eqref{eq:linearSource} can \deleted{also }be obtained as a special case of the optimization problem~\eqref{optProd}. If we assume $\genMLim=\infty$, $\lambda = \infty,$ and a quadratic cost function $(\gamma = 2)$, the optimal source production satisfies
\[s^*_{v,k} = \frac{c_{v,k}}{\sum_{w\in \V}c_{w,k}} \sum_{w\in \V} t_{w,k} \quad v\in \V, k\in \K.\]
This is equivalent to \eqref{eq:linearSource} for the choice of $r_{v,k} = c_{v,k}/\sum_{w\in \V} c_{w,k}$. A corresponding optimization approach with a quadratic cost function can also be applied to obtain a linear form for $\genM^*$ when $\genMLim<\infty$. Thus, we conclude that the formulation as a resource allocation problem remains flexible and encompasses both centralized and decentralized settings, with or without production capacity constraints. Nevertheless, in scenarios where source production naturally emerges from uncoordinated individual behavior, a direct specification of $\genM^*$ can be more intuitive and better aligned with the system dynamics. In Section~\ref{subsec:classesSatisfy}, we show that the choices for $\genM^*$ presented here may satisfy Assumption~\ref{assumption} when paired with suitable choices of the flow mechanism $\flowM^*$, the production limit $\genMLim$, and the cascade cost $\cascadeCost$.

\subsection{Cascade cost functions }\label{subsec:CascadeCost}
Cascading phenomena in networks can lead to two distinct types of systemic impact. One such impact is network fragmentation, where a cascade breaks the system into disconnected components, making it impossible to satisfy some sink requirements. In practice, this can manifest as a large-scale blackout in a power transmission network ~\citep{Dobson_Carreras_Lynch_Newman_2007}, a loss of connectivity in a communication network~\citep{Ren2018}, or service unavailability in an interdependent transportation system~\citep{Daniotti_Servedio_Kager_Robben-Baldauf_Thurner_2024}. To quantify the severity of such failures, we introduce a generalized \textit{cascade size} function based on the loss of sink requirements. Specifically, we define
\begin{equation}
\cascadeCost := \sum_{v \in \V} \sum_{k \in \K} \constCostFunction{v,k} \left( \dem[v,k]^{(1)} - \dem[v,k]^{(end)} \right)^\rho,
\label{generalizedFailureCost}
\end{equation}
where \( w_{v,k} > 0 \) are vertex- and commodity-specific cost coefficients, and \( \rho > 0 \) is a fixed exponent controlling the sensitivity to local demand losses. Here, \(\demM^{(end)}\), with components $\dem[v,k]^{(end)}$, denotes the sink requirement matrix at the termination of the cascade.

The cascade cost function in~\eqref{generalizedFailureCost} generalizes the standard linear cascade size \citep{Nesti_2020,Parandehgheibi2014,Hosseinalipour2020}, which corresponds to the choice \( \rho = 1 \) and \( \constCostFunction{v,k} = 1 \) for all \(v, k\). Such a function counts the total amount of requirement lost due to network disconnections. When \( \rho > 1 \), the cost function is convex, imposing a disproportionately higher penalty on large local resource shortages. This captures situations where, following network fragmentation, resources must be procured locally. Local generation, especially in large amounts, is often inefficient or costly due to limited local capacity, higher marginal production costs, or logistical challenges. For instance, in power grids, localized generation incurs higher costs compared to centralized generation, while in supply chains, local sourcing during disruptions often leads to significant inefficiencies. Thus, large local demands that must be met without network support are more heavily penalized to reflect their critical systemic burden.

When \( 0 < \rho < 1 \), the cost function becomes concave, which assigns a relatively greater penalty to small local losses. This structure can model systems where the startup cost of initiating local resource generation is significant, while additional production beyond the initial setup incurs relatively lower marginal costs. For instance, activating backup generation facilities, rerouting supply chains, or establishing local resource pools often requires substantial fixed efforts, while scaling production once established is comparatively more efficient. Thus, concave cost functions naturally capture systems dominated by initial activation costs in response to resource delivery failures due to network fragmentation. The generalized cascade cost thus allows the modeling of both scenarios: systems where large local changes dominate systemic risk, and systems highly sensitive to small initial disruptions. %Applications where cascade-induced disconnections are central include power transmission networks~\citep{Dobson_Carreras_Lynch_Newman_2007}, communication infrastructures~\citep{Ren2018}, and interdependent transportation systems~\citep{Daniotti_Servedio_Kager_Robben-Baldauf_Thurner_2024}.

In contrast, in cascading models with congestion, the network may remain connected, but performance deteriorates due to reductions in effective edge capacities. In these settings, the source production and sink requirements in the system remain unchanged, which yields $\cascadeCost = 0$, according to~\eqref{generalizedFailureCost}, but the flow cost increases substantially. Hence, to capture the \replaced{congestion impact}{impact of congestion cascades}, it is natural to define the \textit{cascade flow cost} as
\begin{equation}
\cascadeCost := \flowCost(\flowM^{(end)}, \flowLimVec^{(end)}) - \flowCost(\flowM^{(1)}, \flowLimVec^{(1)}),
\label{deltaFlowCost}
\end{equation}
where \(\flowM^{(end)}\) and \(\flowLimVec^{(end)}\) denote the flow matrix and edge capacity vector at the termination of the cascade, respectively, and \(\flowCost\) is the flow cost function defined in~\eqref{flowCost}. The cascade flow cost function reflects the additional cost incurred by the network due to capacity degradations, even in the absence of disconnections. Such congestion effects are particularly relevant in road transportation networks~\citep{Duan2023}.%, communication networks under load shedding~\citep{yan2006efficient}, and cloud computing infrastructures.

The cascade cost functions introduced in~\eqref{generalizedFailureCost} and~\eqref{deltaFlowCost} provide a flexible framework to model a wide range of cascading effects in networked systems. Their applicability will be further illustrated in Section~\ref{Applications}. Moreover, both functions satisfy Assumption~\ref{assumption}, as established in the subsequent section.

\subsection{Satisfiability results}\label{subsec:classesSatisfy}

In this section, we show that Assumption~\ref{assumption} is satisfied for different modeling settings, given in previous sections, that correspond to different types of source production and flow behaviors. Before stating the results, we provide an intuitive overview of the cases considered.

The first setting, given in Proposition~\ref{propClassFunc1}, is suitable for modeling networks where production is centrally coordinated and flow capacity constraints actively influence flow and source allocations, as proposed in Section~\ref{subsec:SourceFunctions}. Here, the cascade cost is modeled by the generalized cascade size \eqref{generalizedFailureCost}, allowing for different exponents \(\rho\) to capture the system's sensitivity to local demand losses, while the possible choices for $\flowM^*$, $\genM^*$, and $\genMLim$ are specified in the proposition. This leads to Assumption~\ref{assumption} holding with \(\delta = \rho\). The second setting, given in Proposition~\ref{propClassFunc2}, is suitable for congestion analysis in decentralized systems. In this case, the cascade cost can be based on the increase in flow cost, following \eqref{deltaFlowCost}, or on the generalized cascade size \eqref{generalizedFailureCost}, resulting in $\delta = 1$ and $\delta = \rho$, respectively.

\begin{remark}\label{remarkNoProdCap}\textnormal{In Propositions~\ref{propClassFunc1} and \ref{propClassFunc2} below, we assume that there are no production capacity constraints ($\genMLim=\infty$), as this reflects the nature of the application examples considered in this paper. Nevertheless, we chose to incorporate production capacity constraints into the general modeling framework to maintain its full applicability to broader settings. Similar satisfiability results can be established under finite production limits by adopting our proof techniques. However, the capacity constraints must be set sufficiently large to guarantee the existence of a feasible source production matrix $\genM^*$ for all admissible sink requirement matrices $\demM$.}\end{remark}

In the first setting, we have the following proposition.
\begin{proposition}
\label{propClassFunc1}
If:
\begin{enumerate}
    \item The flow matrix \(F^*\) is the minimal-cost flow matrix, as given in \eqref{optFlow}, with \textcolor{black}{domain $\mathcal{D} = \R^{|\EE|\times|\K|}$ and} cost function \(c_f\) defined in \eqref{flowCost} with \(\beta = 2\) and constant \(a_e\) for all \(e \in \EE\);
    \item There are no production capacity constraints, i.e., \(\genMLim = \infty\);
    \item The source production matrix \(\genM^*\) is the minimal-cost source matrix as given in \eqref{optProd};
    \item The cascade cost \(\cascadeCost\) is given by the generalized cascade size (Equation~\eqref{generalizedFailureCost});
\end{enumerate}
then Assumption~\ref{assumption} \textcolor{black}{holds with \(\delta = \rho\), for any choice of parameter $\rho>0$ in Equation~\eqref{generalizedFailureCost}.}
\end{proposition}

 An example of a system following this class of models is a power transmission system, which we demonstrate in detail in Section~\ref{subsec:powerEx}. We note that in Proposition~\ref{propClassFunc1}, we assume $\beta = 2$ to ensure that the source production matrix $\genM^*$ is well-defined. As illustrated by Example~\ref{exampleNonQuasi}, the existence of $\genM^*$ is not guaranteed for general values of $\beta$. Nevertheless, if one could establish the existence of $\genM^*$ for some $\beta\neq 2$, then the result of Proposition~\ref{propClassFunc1} would extend straightforwardly to that case.

In the following proposition, we show that Assumption~\ref{assumption} holds for the second class of problems. This is relevant for decentralized systems where production emerges without coordination, and where edge failures primarily cause congestion effects rather than disconnections. 
\begin{proposition}
\label{propClassFunc2}
If:
\begin{enumerate}
    \item The flow matrix \(F^*\) is the minimal-cost flow matrix, as given in \eqref{optFlow}, with \textcolor{black}{domain $\mathcal{D} = \R_+^{|\EE|\times|\K|}$ and }cost function \(c_f\) defined in \eqref{wardrop};
    \item There are no production or flow capacity constraints, i.e., \(\genMLim = \infty\) and \(\lambda = \infty\);
    \item The source production matrix \(\genM^*\) is the minimal-cost source matrix as given in \eqref{optProd};
    \item The cascade cost \(\cascadeCost\) is given by the cascade flow cost (Equation~\eqref{deltaFlowCost}) or by the generalized cascade size (Equation~\eqref{generalizedFailureCost});
\end{enumerate}
\textcolor{black}{then Assumption~\ref{assumption} holds, with \(\delta = 1\) when the cascade cost \(\cascadeCost\) is given by the cascade flow cost as in \eqref{deltaFlowCost}, and with \(\delta = \rho\) for any chosen \(\rho > 0\) when \(\cascadeCost\) is given by the generalized cascade size as \eqref{generalizedFailureCost}.}

\end{proposition}

Examples of networks that belong to this class of systems are certain traffic or processing networks, which we discuss in detail in Sections~\ref{trafficApp} and \ref{subsec:processingEx}.

These satisfiability results are crucial because they show that our framework captures a broad range of network settings through appropriate choices of flow and production cost structures. Furthermore, they establish sufficient conditions under which our main results on the tail behavior of the cascade cost $\cascadeCost$, presented in Section~\ref{mainRes}, apply. In particular, they lend theoretical support to the hypothesis that a Pareto-tailed input induces a Pareto-tailed cascade cost, suggesting the existence of universal heavy-tailed behavior across diverse classes of flow networks. This universal behavior is further substantiated in Section~\ref{Applications}, where we examine the cascade cost in three real-life applications modeled within our framework and obtain consistent heavy-tailed behavior through the application of the aforementioned propositions and asymptotic results.

The formal proofs of the propositions are provided in Appendix~\ref{subsec:proofsPropositions}. These proofs build on a sequence of auxiliary results that establish scale invariance and SRC properties of the flow, production, and flow capacity functions at each stage of the cascade process, which are collected in Appendix~\ref{proofs}. The overall strategy proceeds inductively, characterizing how the cascade dynamics respond to perturbations of the input distribution. %The detailed proofs of Propositions~\ref{propClassFunc1} and~\ref{propClassFunc2}, along with necessary intermediate results, are presented in Section~\ref{proofs}.

\section{Application examples}\label{Applications}
In this section, we present three examples to illustrate how our framework can be used to model and analyze cascades in various applications. First, we show an example of a power transmission system where failures are associated with line overloads, and a cascade cost represents the total amount of power lost during a blackout. Second, we discuss an example of a highway network where the cascade represents the propagation of congestion, and the cascade cost is given by either the aggregate additional travel cost or the number of travelers that could not make the intended journey due to network collapse. Last, we apply our framework to model cascades in processing networks, where jobs are sent through a system of processors of a given capacity. Here, scale-free input represents job sizes, and, unlike in the other two examples, the failures occur on vertices instead of edges. We show that this scenario can still be studied with our model, under small modifications to the underlying graph. 

\subsection{Power transmission systems} \label{subsec:powerEx}
Blackouts in power transmission systems are a classic application of cascading failure models, and over the years, many variants of such models have been proposed \citep{OPA2002, Branching2004, Cascade2003, Nesti_2020, Guo2017}. Here, we show that this application can also be modeled within our general modeling framework, described in Section~\ref{modelDescription}. To do so, the parameters and functions of our framework must be specified to capture the mechanisms of power transmission networks. In this section, we identify the choices suitable for modeling blackouts, resulting in a \deleted{specific }model closely resembling \replaced{the one}{those} in~\citep{Nesti_2020}.

In cascading models for power transmission systems, the graph $\G$ represents the underlying physical network, where the vertices $v\in \V$ are associated with the buses and the edges $e\in \EE$ with the power lines. 
There is only one commodity transported through the network --- electricity, hence $\K = \{1\}$. 

In real transmission systems, buses can be purely load buses, purely generator buses, or a combination of both. In macroscopic cascade models for blackouts, this is simplified by treating each vertex as a bus that can both produce and consume electricity. The major buses are typically located near large demand centers, so for modeling purposes, we associate each vertex with a representative city.

According to our general framework in Section~\ref{modelDescription}, because there is only one commodity, the source requirement matrix $\demM$ is given by $$\demM =\text{diag}(\pareto)\cdot  \bm{Q} = (q_{1,1} \paretoI{1}, \dots q_{\nv, 1}\paretoI{\nv})^T.$$ 
Here, $\paretoI{v}$ denotes the population of the city associated with vertex $v$, and $q_{v,1}$ can be interpreted as the average electricity requirement per inhabitant of that city.%This is in line with the approach of \citep{Nesti_2020} and \citep{janicka2024scalefree}. {\color{red} Why does this needs to be mentioned? Your description is still in full generality?}

In cascading failure models for power transmission systems, it is common to model the power flow using the so-called DC power flow problem, which is a linearization of the AC power flow problem \citep{DCPowerFlowRevisited}. Specifically, for a netput matrix $\bm{U}$, the corresponding flow matrix is given by $\flowM^* = \bm{V}\bm{U}$, where $\bm{V}$ is the Power Transfer Distribution Factors matrix, uniquely defined by the underlying network. In particular, 
\(\bm{V} := \bm{S}\incM^T\bm{L}^+,\)
where $\bm{S}$ is the $\nee\times \nee$ diagonal matrix of line susceptances, $\bm{L} := \incM\bm{S}\incM^T\in \R^{\nv\times\nv}$ is the weighted graph Laplacian with edge weights given by susceptance, and $+$ denotes the matrix pseudoinverse\deleted{ operator}. 

It is known that the DC power flow $\flowM^* = \bm{V}\bm{U}$ can be written as an instance of the minimal-cost flow, which we state rigorously in the following lemma.
\begin{lemma}
    $\flowM^* = \bm{V}\bm{U}$ is the optimal solution to Problem \eqref{optFlow} with cost function \(c_f = \sum_{e\in \EE}\frac{1}{2S_{e,e}}|f_e|^2.\)\label{DCPowerFlow}
\end{lemma} 
We note that the above flow cost is an instance of \eqref{flowCost}, with $b_e = 1$ for all $e\in \EE$, $\beta = 2$, and $a_e = S_{e,e}$. The proof uses the Karush-Khun-Tucker (KKT) conditions to prove optimality of the proposed solution. The proof is provided in Appendix~\ref{app:KKTproof}. 

 In models of power transmission systems, the question of how much power has to be generated and where is typically formulated as an optimization problem, called the Optimal Power Flow (OPF) problem. Although there are many variants of the OPF problem, common formulations aim to minimize the generation cost, while complying with safety operating limits and meeting the total power demand \citep{Frank01122016}. If the OPF problem assumes a DC power flow distribution, the optimization problem is called the DC-OPF problem \citep{Baker21}. This problem can be viewed as an instance of the resource allocation problem~\eqref{optProd}, considered in our framework, where 
 \begin{enumerate}
     \item[(i)] power generation corresponds to source production, 
     \item[(ii)] the balance of total demand and generation is ensured by Constraint~\eqref{optProdBalance},
     \item[(iii)] the power generation capacity is given by Constraint~\eqref{optProdConst}, and
     \item[(iv)] the line capacity constraint is ensured by Constraint~\eqref{optFlowConst}. 
     \end{enumerate} A common approach for the cost function is to use a quadratic form \citep{Molzahn2017}, which in our model corresponds to the choice of $\gamma = 2$ and $c_{v,k} = 1$ for all $v\in \V$ in \eqref{optProd}.

In the general framework in Section~\eqref{modelDescription}, the operational line capacities $\flowLimVec^{(1)}$ are determined in the planning phase by solving the resource allocation problem~\eqref{optProd} with $\flowLimVec^{(0)} = \infty$ and $\lambda^{(0)} = 1$; a similar approach has been used in other cascading failure models, including models of power networks \citep{Nesti_2020, Motter2002}. The planning slack parameter $\tau$ allows to specify how much larger the line capacities are compared to the benchmark given by the planning flows $\flowTot^{(0)}$ (see \eqref{flowCapacity}). The value of $\varepsilon_{\min}$ can be chosen based on properties of power lines in the network. The operational source production $\genM^{(1)}$ is obtained by solving the minimal-cost source production problem \eqref{optProd} with input $\demM$, $\genMLim$, and $\flowLimVec^{(1)}$, and $\lambda^{(1)}\in (0,1)$. Here, $\lambda^{(1)}$ plays the role of a safety parameter, ensuring that at most a fraction $\lambda$ of true line capacities is utilized during a typical operation of the system.

In the emergency phase, our general framework assumes that the cascade of line failures is initiated randomly, according to some probability law $p^{(1)}(\cdot)$ that allows for correlations among the initial line failures. This law could be chosen based on statistical data or real-life features of the system (for example, lines located in wooded areas are more likely to fail, due to the proximity of trees \citep{Parent2019}). In power transmission systems, if a line is overloaded, it fully breaks down and is no longer usable. To model this, we exclude partial failures by setting $p^{(1)}(\bm{x})=0$ whenever $x_i=\mathrm{P}$ for some $i\in\EE$, and by taking $p_{e,c}(\cdot)\equiv 0$ for all $e\in\EE$. Thus, the only admissible failure type is removal (complete failure). Given this, it is reasonable to assume that all lines may fail only once, i.e., $n_e = 1$ for all $e\in \EE;$. In other words, each line is either fully functional or fully broken. At cascade stage $t\geq 2$, if the flow on a given line exceeds its capacity, i.e., $\exc{e}^{(t)}>1$, the line fails with probability $p_{e,r}(\exc{e}^{(t)})$, according to a Bernoulli distribution. Again, the function $p_{e,r}(\cdot)$ could be chosen based on statistical data or physical properties of line $e$.

The cascade cost or blackout size in power transmission systems is often measured in the amount of electricity that cannot be delivered due to disconnections in the network. This measure, in our model, is equivalent to the amount of requirement lost in the process of balance restoration. Hence\deleted{, in this setting}, the appropriate choice for the cascade cost function $\cascadeCost$ is the cascade size, defined in \eqref{generalizedFailureCost}, with $\constCostFunction{v,k} = 1$ for all $v\in \V$ and $\rho = 1$. Other cost functions could also be considered by choosing appropriate values for $\constCostFunction{v,k}$ and~$\rho$.

To conclude, we have shown how to choose functions $\flowM^*$, $\genM^*$, and $\cascadeCost$ in order to model cascading failures in power transmission systems using the modeling framework proposed in this paper. According to Proposition~\ref{propClassFunc1}, the chosen functions satisfy Assumption \ref{assumption} with $\delta = \rho$, implying that Theorem \ref{mainThm} holds. This finding is in line with the results in \citep{Nesti_2020} and \citep{janicka2024scalefree}, which establish that the total cascade size has a scale-free tail with power-law exponent $\alpha$. Our result is novel in two respects. First, we employ a probabilistic failure mechanism in the emergency phase, whereas the models in \citep{Nesti_2020, janicka2024scalefree} used a deterministic rule that removed all lines for which capacity was exceeded. Second, we generalize the cost function by introducing the parameter $\rho$, inherited from the generalized cascade cost~\eqref{generalizedFailureCost}. In earlier studies, only the case $\rho = 1$ was considered, corresponding to a direct proportionality between the cost and the total failed demand. Our formulation allows for more general measures of cascade impact, and shows explicitly how the tail exponent changes from $\alpha$ to $\alpha / \rho$ depending on the chosen cascade cost function.

\subsection{Highway traffic networks} \label{trafficApp}
A significantly different application \deleted{example }of our framework involves the modeling of congestion in highway traffic networks, which can also be interpreted as a cascade process. The example \deleted{described }in this section bears many similarities with the cascade model for highway traffic \deleted{presented }in~\citep{Janicka_Sloothaak_Vlasiou_Zwart_2025}. Minor differences occur in the construction of the source production and sink requirements and in the cascade initiation mechanism. We explain these differences as we introduce these modeling components in the following paragraphs.

The network of highways is modeled using a connected directed graph $\G$, where the vertices $v\in \V$ represent highway crossings, often associated with cities, and the edges represent highways. Note that, unlike in power systems, edges can transport flow only in a single direction, hence, two separate directed edges represent a bi-directional highway. Moreover, since the vertices are associated with cities, the vertex weight vector $\pareto$ represents the vertex of city sizes.

In highway traffic networks, commuters travel from an origin location (vertex) to their destination. The flow on each edge represents the intensity of commuters that travel through this edge at a given time. As is common in traffic literature, we assume that vehicles with a destination at a given location constitute a separate commodity. Hence, in this model, we have $\K = \V$ because each vertex can be a destination. 

The sink requirement $\dem[v,k]$ of commodity $k\in \K$ at vertex $v\in \V$ represents the number of commodity $k$ commuters that travel to vertex $v$, while the source production $\gen[v,k]$ represents the number of vehicles of commodity $k$ that originate at vertex $v$. Due to the way we have defined commodities, it is crucial that $\dem[v,k] > 0$ if and only if $v = k$. Hence, 
\[\demM = \text{diag}(\pareto) \cdot \bm{Q},\]
with $q_{v,w}>0$ if $w = v$ and $q_{v,w} = 0$ otherwise. This means that the number of commuters traveling to vertex $v$ grows linearly with its size, which is a fair assumption, as larger cities offer more opportunities for business and leisure. With this motivation in mind, the parameter $q_{v,v}$ can be interpreted as the ``attractiveness'' factor of vertex $v$. Similarly, we assume that $\gen[v,k]^* = r_{v,k}\paretoI{k}$, $r_{v,k}>0$ such that for every $k\in \K$, $\sum_{v\in \V}r_{v,k} = q_{k,k}$. In matrix notation, the source production can be written as
\[\genM^*(\pareto) =  \bm{R} \cdot \text{diag}(\pareto),\] where $\bm{R} = (r_{v,k})_{v\in\V, k\in \K}$. Note that the use of letter $\bm{R}$ is intentional as this choice of $\genM^*$ is a special case of \eqref{eq:linearSource}. The parameter $r_{v,k}$ can be interpreted as the fraction of commuters traveling from vertex $v$ to vertex $k$. In this example, the meaning behind the source production is atypical; it is not a resource that can be generated at suitably chosen locations, but it is exogenously induced by the traveling preferences of the commuters. As such, the notion of production capacity in this setting is not meaningful, which is why we set $\genMLim= \infty$, nullifying its influence on the behavior of the model. 

Alternatively, a commodity could be defined as travelers \textit{originating} at a given vertex. This has been done in \citep{Janicka_Sloothaak_Vlasiou_Zwart_2025}, where the authors also considered an alternative interpretation of sink and source matrices. There, the amount of traffic \textit{originating} from a given location is proportional to its size, which intuitively means that the roles of matrices $\demM$ and $\genM^*$ are swapped, compared with $\demM$ and $\genM^*$ defined here. Note that the model presented in this paper is capable of supporting either formulation, and the choice does not influence the resulting power-law behavior of the cascade cost. 

A common choice for the flow matrix $\flowM^*$ in traffic networks represents the so-called Wardrop \replaced{UE}{User Equilibrium} flow distribution, which we also apply in this example. This equilibrium distribution ensures that no individual commuter can improve their travel time by choosing an alternative route, given that all other commuters do not change their \deleted{travel }behavior. As stated in Section~\ref{subsec:FlowFunctions}, this flow matrix can be obtained as a solution to Problem~\eqref{optFlow} over the non-negative domain $\R^{\nee\times \nv}_+$ with cost function \eqref{wardrop}. Here, the parameter $d_e$ represents the ``free flow'' travel cost on edge $e$, $b_e$ represents the congestion effect of edge $e$, and $\gamma$ represents the overall congestion effect on travel cost. Moreover, the travel cost is a decreasing function of the edge capacity, which reflects the fact that a wider road (or with more lanes) typically allows for faster travel time (cost). 

It remains to choose the parameters $\tau$ and $\lambda$. Similarly to many traffic models, we choose $\tau>1$, which implies that the capacity of the system is at least $\tau$ times larger than the capacity required to accommodate traffic in the operational phase.  Note that the flow in highway traffic is the aggregate flow of individual commuters, who choose their routes independently. As such, there is no overlooking party able to ensure that the flow magnitudes are within the safety level determined by the parameter $\lambda$. What is more, in traffic networks, the exceedance of edge capacity, i.e., congestion, happens regularly and does not immediately render the edge unusable, but it increases the travel cost through this edge. This means that the flow capacity constraint is not an integral aspect of traffic networks and, to reflect this, we set $\lambda = \infty$, implying that Constraint~\eqref{optFlowConst} always holds. 

Before we discuss the cascading congestion effects, we first show that the choices for $\genM^*$, $\genMLim^*$, and $\flowM^*$ satisfy the assumptions of the model. Since $\genM^*$ does not depend on $\flowLimVec$ and $\genMLim$, it immediately follows that the planning and operational source requirements are equal, i.e., $\genM^{(0)}(\pareto) = \genM^{(1)}(\pareto)$. This is a desired behavior, as the source production is exogenous and therefore not influenced by the specifics of the network design determined in the planning stage. Moreover, $\genM^*$ satisfies Constraint~\eqref{optProdBalance} since
\[\e{_{\nv}}^T \genM^* =\e{_{\nv}}^T\bm{R}\cdot \text{diag}(\pareto)  = (q_{1,1}, \dots, q_{\nk,\nk}) \cdot \text{diag}(\pareto)= \e{_{\nv}}^T\text{diag}(\pareto)\cdot \bm{Q} =\e{_{\nv}}^T\demM(\pareto),\]
while Constraints $\eqref{optProdConst}$ and $\eqref{optFlowConst}$ are immediately satisfied given the choice of $\genMLim$ and $\lambda$. Hence, $\genM^{(0)}(\pareto)$ and $\genM^{(1)}(\pareto)$ both satisfy the assumptions of the model. Moreover, $\flowM^{(0)}$ and $\flowM^{(1)}$ are well-defined as the objective function given in~\eqref{wardrop} is strictly convex and graph $\G$ is connected, implying that \eqref{optFlow} has a unique solution for every input $\bm{U}$ that satisfies the production balance constraint (Constraint \eqref{optProdBalance}). Hence, we conclude that the proposed setup meets the criteria of our model. %Following many traffic congestion models, we assume that $\tau>1$, which means the capacity of the roads is a fraction $\tau$ larger than the minimum capacity necessary to ensure no congestion under the typical operation of the system. Factor $\lambda$ is a safety tuning parameter that has an intuitive interpretation whenever the source production allocation can be controlled by an overlooking party. Since this is not the case in the traffic networks, we set $\lambda = 1$, meaning that no additional safety is imposed on the system.

Next, we describe the emergency phase, modeling the congestion cascade in the network. The first failure(s) (either congestion or complete failure) may be caused by an external event, such as flooding, or by an internal event such as a road accident. Each edge failure leads to a decrease in edge capacity by a factor $\failFactor_e(\exc{e})$. If $\failFactor_e(\exc{e})>0$, then the congestion is only partial, for example, one of the lanes is blocked, therefore the road can still be used but the corresponding travel cost increases. Otherwise, the road is fully blocked and traffic through the edge needs to be completely redirected. Note that in \citep{Janicka_Sloothaak_Vlasiou_Zwart_2025}, the factor $\failFactor_e$ at the first cascade stage was modeled as a random variable, whereas here it is a deterministic function, resulting in a negligible difference between the two models. If complete edge failures are allowed, it is convenient to measure the cascade cost using the cascade size function, which tells us how many commuters could not make their trip due to congestion on the roads. Otherwise, it is more suitable to use the cascade flow cost, representing the added travel cost caused by the cascade. 

In conclusion, it follows from Proposition~\ref{propClassFunc2} that Assumption \ref{assumption} is satisfied with $\delta = 1$ or $\delta = \rho$, depending on the choice of the cascade cost function $\cascadeCost$. This implies that Theorem \ref{mainThm} holds and the tail of the cascade cost distribution is scale-free with parameter $\alpha$ or $\alpha/\rho$, respectively, which depends on the choice of $\cascadeCost$.

\subsection{Processing networks} \label{subsec:processingEx}

Our framework can also be used to describe cascading behavior in \emph{processing networks}---systems where jobs traverse a network and are transformed as they pass through processing nodes. Examples include:
\begin{itemize}
    \item \textbf{Manufacturing systems}, where raw materials are assembled into products,
    \item \textbf{Communication networks}, where data packets are processed by routers or servers,
    \item \textbf{Supply chains}, where goods move through warehouses and distribution centers toward end users.
\end{itemize}

Processing networks are diverse, and their modeling depends, amongst others, on whether production is decentralized or centralized, whether the network permits cyclic rerouting of jobs, and how flow is redistributed. In what follows, we model a simple subclass suitable for systems such as content delivery networks (CDNs): we assume a directed, acyclic network in which origin vertices, representing requests from customers, have unbounded production capacity, processing occurs at finite-capacity resources, and routing is determined by minimizing a capacity-dependent flow cost function. We later comment on how these assumptions can be modified to capture other subclasses of processing networks.

As in the traffic flow setting, we approximate the job flow with a continuous fluid, which is a common approach when dealing with a large volume of jobs. We model the network as consisting of \emph{origin vertices}, \emph{destination vertices}, and \emph{processing vertices}. Jobs enter the system at origin vertices, undergo a series of transformations at processing vertices, and exit the system at their assigned destinations. A small example of a directed, acyclic processing network is shown in Figure~\ref{fig:ProcessingNetwork}. Although cycles are uncommon in CDNs, the modeling approach presented in this section also applies to networks with cyclic topology.

\begin{figure}[ht]
    \centering
    \begin{subfigure}[ht]{0.49\textwidth}
\scalebox{0.8}{
        \begin{tikzpicture}
  % Define the style for the vertices
  \tikzset{vertex/.style = {shape=circle,draw, fill = black, minimum size=1em, inner sep = 0pt}}
  % Define the style for the edges
  \tikzset{edge/.style = {->,> = latex'}}

  % Nodes
  \node[vertex, fill = BrickRed, draw = BrickRed] (1) at (0,0) {};
  \node[vertex, fill = BrickRed, draw = BrickRed] (2) at (3,0) {};
  \node[vertex, fill = BrickRed, draw = BrickRed] (3) at (-3,0) {};
  \node[vertex] (4) at (-1.5,-2.2) {};
    \node[vertex] (5) at (1.5,-2.2) {};
  \node[vertex] (6) at (0,-4.4) {};
  \node[vertex] (7) at (3,-4.4) {};
  \node[vertex] (8) at (-3,-4.4) {};
  \node[vertex, fill = Violet, draw = Violet] (9) at (0,-6.6) {};
  % Edges
  \draw[-{Stealth[length=2mm, width=1.5mm]}] (1) -- (4);
  \draw[-{Stealth[length=2mm, width=1.5mm]}] (1) -- (5);
  \draw[-{Stealth[length=2mm, width=1.5mm]}] (2) -- (5);
  \draw[-{Stealth[length=2mm, width=1.5mm]}] (3) -- (4);
  \draw[-{Stealth[length=2mm, width=1.5mm]}] (4) -- (6);
  \draw[-{Stealth[length=2mm, width=1.5mm]}] (4) -- (8);
  \draw[-{Stealth[length=2mm, width=1.5mm]}] (5) -- (6);
  \draw[-{Stealth[length=2mm, width=1.5mm]}] (5) -- (7);
  \draw[-{Stealth[length=2mm, width=1.5mm]}] (6) -- (9);
  \draw[-{Stealth[length=2mm, width=1.5mm]}] (7) -- (9);
  \draw[-{Stealth[length=2mm, width=1.5mm]}] (8) -- (9);
  %\node at (5.5,-0.55) {\textit{ origin nodes}};
  %\node at (5.5,-3.3) {\textit{processing nodes}};
  %\node at (5.5,-6.05) {\textit{destination nodes}};
  \draw[dashed, fill = gray] (-4,-1.1) -- (4,-1.1);
    \draw[dashed, fill = gray] (-4,-5.5) -- (4,-5.5);
  %\draw[<-] (2) to [out=-90,in=-90, looseness =2.2] node[fill =white]{\footnotesize$-1 + \sqrt{2}/2$} (3); % Edge 4

 \end{tikzpicture}}
    \caption{Orginal graph.\vspace{0.75cm}}
    \label{fig:ProcessingNetworkOriginal}
            
    \end{subfigure}
    ~
    \begin{subfigure}[ht]{0.49\textwidth}
    \scalebox{0.8}{
        \begin{tikzpicture}
  % Define the style for the nodes
  \tikzset{vertex/.style = {shape=circle,draw, fill = black, minimum size=1em, inner sep = 0pt}}
  % Define the style for the edges
  \tikzset{edge/.style = {->,> = latex'}}

  % Nodes
\draw[dashed, fill = gray!20] (-1.5,-2.2) ellipse (0.5cm and 1cm);
\draw[dashed, fill = gray!20]  (1.5,-2.2) ellipse (0.5cm and 1cm);
\draw[dashed, fill = gray!20]  (-3,-4.4) ellipse (0.5cm and 1cm);
\draw[dashed, fill = gray!20]  (0,-4.4) ellipse (0.5cm and 1cm);
\draw[dashed, fill = gray!20]  (3,-4.4) ellipse (0.5cm and 1cm);
  
  \node[vertex, fill = BrickRed, draw = BrickRed] (1) at (0,0) {};
  \node[vertex, fill = BrickRed, draw = BrickRed] (2) at (3,0) {};
  \node[vertex, fill = BrickRed, draw = BrickRed] (3) at (-3,0) {};
  \node[vertex] (4a) at (-1.5,-1.8) {};
    \node[vertex] (5a) at (1.5,-1.8) {};
  \node[vertex] (4b) at (-1.5,-2.6) {};
    \node[vertex] (5b) at (1.5,-2.6) {};
  \node[vertex] (6a) at (0,-4) {};
  \node[vertex] (7a) at (3,-4) {};
  \node[vertex] (8a) at (-3,-4) {};
  \node[vertex] (6b) at (0,-4.8) {};
  \node[vertex] (7b) at (3,-4.8) {};
  \node[vertex] (8b) at (-3,-4.8) {};
  \node[vertex, fill = Violet, draw = Violet] (9) at (0,-6.6) {};
  % Edges
  \draw[-{Stealth[length=2mm, width=1.5mm]}] (1) -- (4a);
  \draw[-{Stealth[length=2mm, width=1.5mm]}] (1) -- (5a);
  \draw[-{Stealth[length=2mm, width=1.5mm]}] (2) -- (5a);
  \draw[-{Stealth[length=2mm, width=1.5mm]}] (3) -- (4a);
  \draw[-{Stealth[length=2mm, width=1.5mm]}] (4b) -- (6a);
  \draw[-{Stealth[length=2mm, width=1.5mm]}] (4b) -- (8a);
  \draw[-{Stealth[length=2mm, width=1.5mm]}] (5b) -- (6a);
  \draw[-{Stealth[length=2mm, width=1.5mm]}] (5b) -- (7a);
  \draw[-{Stealth[length=2mm, width=1.5mm]}] (6b) -- (9);
  \draw[-{Stealth[length=2mm, width=1.5mm]}] (7b) -- (9);
  \draw[-{Stealth[length=2mm, width=1.5mm]}] (8b) -- (9);
  \draw[-{Stealth[length=2mm, width=1.5mm]}] (4a) -- (4b);
  \draw[-{Stealth[length=2mm, width=1.5mm]}] (5a) -- (5b);
  \draw[-{Stealth[length=2mm, width=1.5mm]}] (6a) -- (6b);
  \draw[-{Stealth[length=2mm, width=1.5mm]}] (7a) -- (7b);
  \draw[-{Stealth[length=2mm, width=1.5mm]}] (8a) -- (8b);
  %\node at (5.5,-0.55) {\textit{ origin nodes}};
  %\node at (5.5,-3.3) {\textit{processing nodes}};
  %\node at (5.5,-6.05) {\textit{destination nodes}};
  \draw[dashed, fill = gray] (-4,-1.1) -- (4,-1.1);
    \draw[dashed, fill = gray] (-4,-5.5) -- (4,-5.5);

  %\draw[<-] (2) to [out=-90,in=-90, looseness =2.2] node[fill =white]{\footnotesize$-1 + \sqrt{2}/2$} (3); % Edge 4

\end{tikzpicture}}
    \caption{Modified graph, which arises by replacing each processing vertex with a pair of vertices connected by a directed edge.}
    \label{fig:ProcessingNetworkMod}
    \end{subfigure}
    \caption{Example of a small processing network. Here, red, black, and blue vertices correspond to the origin, processing, and destination vertices, respectively.}
    \label{fig:ProcessingNetwork}
\end{figure}
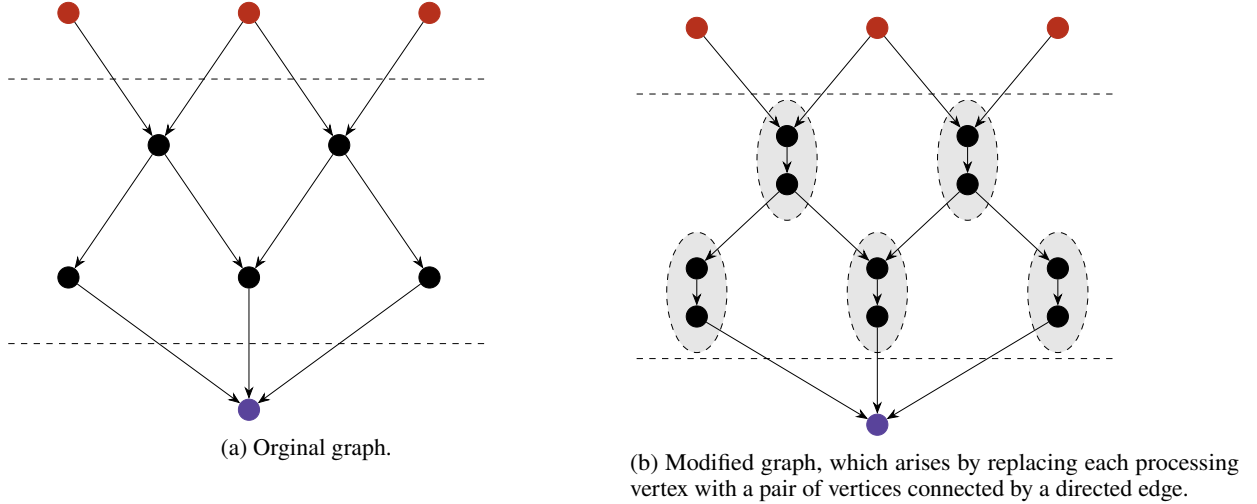

In classical examples of processing networks, including CDNs, the processing capacity is associated with the \emph{vertices} \citep{Ren2018, tang2016complex, Fu2023}. In contrast, our modeling framework in Section~\ref{modelDescription} assumes that processing occurs on the \emph{edges}. However, we can circumvent this difference by modeling each processing vertex as two auxiliary vertices connected by a single \emph{processing edge}, as shown in Figure~\ref{fig:ProcessingNetworkMod}. The capacity of a processing vertex is thus captured by the capacity of the corresponding edge.

To \deleted{clearly }distinguish different network components, we define
\(
\V = \V_o \cup \V_d \cup \V_p\) and \( \EE = \EE_p \cup \EE_r,
\)
where:
\begin{itemize}
    \item \( \V_o, \V_d, \V_p \) denote the sets of origin, destination, and processing vertices, respectively,
    \item \( \EE_p \) denotes the set of \emph{processing edges}, each taking the form \( (v_p, v_p) \) for \( v_p \in \V_p \),
    \item \( \EE_r \) denotes the set of all other (routing) edges.
\end{itemize}
While in general processing networks a vertex may simultaneously act as an origin, a processing node, and/or a destination, such dual roles can be removed without loss of generality by introducing auxiliary origin and destination vertices so that each vertex in the model plays a singular role.

The \emph{arrival of jobs at origin vertices} corresponds to source production in our model, while the \emph{departure at destination vertices} corresponds to sink requirements. Each commodity represents a different job type.
\subsubsection*{Handling job arrivals}
In many real-world processing networks, job sizes are heavy-tailed \citep{crovella1998heavy, Leland93, resnick1997heavy}---for example, file sizes in communication networks or CDNs. In our example, we model their impact via source production rates at origin nodes: larger jobs require proportionally more processing per unit time, so heavy-tailed job sizes yield heavy-tailed source production rates in the aggregated flow.  Intuitively, one could model this by letting the source production \( \genM^{(0)} \) be proportional to the heavy-tailed vertex sizes $\pareto$, with the sink requirements \( \demM^{(0)} \) defined as a function of \( \genM^{(0)} \). However, this is the opposite of the structure assumed in our model, where \( \genM^{(0)} \) is derived from  \( \demM^{(0)} \). To resolve this mismatch, we reinterpret source and sink roles using sign conventions: \emph{negative production corresponds to demand}, and \emph{negative demand corresponds to production}. We thus define:
\[
\demM^{(0)} = -\mathrm{diag}(\pareto) \bm{Q}, \quad 
\genM^{(0)} = -\bm{R} \cdot \mathrm{diag}(\e{_{\nv}}^T \demM^{(0)}),
\]
as given in Equation~\eqref{eq:linearSource}.
Moreover, we set:
\(
q_{v,k} = 0\) for all \(v \in \V_d \cup \V_p,\) such that \(\sum_{k \in \K} q_{v,k} = 1\) for all \(v \in \V_o.\) This allows us to interpret \( \paretoI{v} \) as the total arrival rate of jobs at origin vertex \( v \), and \( q_{v,k} \) as the fraction of the jobs arriving at $v$ that are of type $k$ .  We further assume that the matrix \( \bm{R} \) satisfies \( r_{v,k} > 0 \) if and only if \( v \in \V_d \), where $r_{v,k}$ indicates the fraction of type \( k \) jobs destined for \( v \).

Equation~\eqref{eq:linearSource} is a special case of~\eqref{optProd}, under the assumption of no generation and flow constraints, i.e., $\genMLim = \infty$ and $\lambda = \infty$. This choice corresponds to decentralized generation and is justified in settings such as CDNs, where users can generate an unbounded number of requests and do so without awareness of the network’s congestion state. In other processing networks, for example manufacturing, production capacity is typically finite; while we do not treat this case explicitly in our satisfiability results in Section~\ref{sec:satisfiability}, Remark \ref{remarkNoProdCap} explains that our results can be extended to accommodate certain production capacity function $\genMLim$. For systems with centralized production, generation can instead be modeled using Equation~\eqref{optProd}.
\subsubsection*{Flows and cost structure}
As in the traffic model, we assume \emph{non-negative flow}, which means that jobs can only move in the assigned direction of each edge. This is common in many processing networks because even bidirectional links often have asymmetric capacities and control, making it natural to model each direction as a separate arc with non-negative flow. The flow is chosen to \emph{minimize the total processing cost}, defined using the cost function in Equation~\eqref{wardrop}, reflecting the centralized routing logic typical of CDN-like systems, where flows are allocated with awareness of processing capacities across the network. We assume that the flow on non-processing edges, i.e., edges in the set $\EE_r$ do not incur any cost, which is why we set:
\[
b_e = d_e = 0 \quad \text{for all } e \in \EE_r, \qquad d_e > 0, \ b_e \geq 0 \quad \text{for all } e \in \EE_p.
\]
This choice is not restrictive---positive costs on edges $e\in \EE_r$ could be introduced to capture, for example, transport delays or communication latencies in other types of processing networks. Here, \( d_e \) represents the fundamental cost of processing one unit of flow, and \( b_e \) controls the sensitivity of the cost to load; higher \( b_e \) values imply a higher dependence of the processing cost on the processing load.
\subsubsection*{Cascade process}
To model cascade failures, we assume \emph{only processing edges can fail}, i.e., \( p_{e,c}(x) = p_{e,r}(x) = 0 \) for all \( e \in \EE_r \). This assumption could be relaxed---for example, in transport or communication networks, routing edges may \deleted{also }fail due to disruptions or outages, \deleted{which can be }modeled by assigning nonzero failure probabilities to edges in $\EE_r$. The initial failure is triggered at random, and consequent failures occur when the flow on a processing edge exceeds its capacity \( \flowLim{e} \). The failure may be \emph{partial or complete}, and the associated cascade cost is measured using either Equation~\eqref{generalizedFailureCost}, which, if $\rho = 1$, counts the number of jobs that cannot be processed due to network disconnections or \eqref{deltaFlowCost}, measuring the flow cost increase due to congestions.

This setup satisfies \deleted{the conditions of }Proposition~\ref{propClassFunc2}, which implies that Theorem~\ref{mainThm} also holds in the proposed context of the processing networks. \replaced{S}{More s}pecifically, this result shows that the cascade cost has a Pareto tail with parameter $\alpha/\rho$ ($\rho = 1$ if $\cascadeCost$ is the cascade flow cost), where $\alpha$ governs the tail of the distribution of job arrivals to the system. 

We point out that this application example illustrates how our framework can be applied to model nodal cascade failures through simple modifications of the underlying graph. Note that here, for the sake of exposition, we presented a simple setting of a processing network, where each processing unit was capable of processing every job type. This could be generalized to the case of specialized servers by adapting the flow cost function. In particular, for every processing unit that is incapable of processing job type $k$, the flow function should assign an infinite cost to any non-zero flow of commodity $k$. Although such a generalization no longer satisfies Proposition~\ref{propClassFunc2} directly, we are confident that using a similar approach as in the proof of Proposition~\ref{propClassFunc2}, one could show that Assumption~\ref{assumption} remains true.

\section{Discussion}
%\textcolor{red}{Our main contributions}\\ 
%This study develops a multi-commodity cascade model for flow networks. By parametrizing core components such as resource allocation, flow distribution, and failure effects, the model accommodates a wide variety of network behaviors. This allows it to represent operational differences across systems like power grids, highway traffic, and processing networks---an integration that, to the best of our knowledge, has not been achieved in prior models.
This paper studies a universal mechanism for the emergence of Pareto-tailed cascade costs in a multi-commodity cascade model for flow networks. By parametrizing core model components, our framework captures a wide range of network behaviors and represents operational differences across systems such as power grids, highway traffic, and processing networks---an integration that, to the best of our knowledge, has not been achieved in prior models. Under mild assumptions, in Theorem~\ref{mainThm} we prove that cascade costs in this framework follow a scale-free tail. To make this result applicable in diverse settings, we identify broad classes of parametrized functions for resource allocation, flow distribution, and cascade cost measurement that satisfy these assumptions (Propositions \ref{propClassFunc1} and \ref{propClassFunc2}). While not exhaustive, these classes cover many practical systems, and the proof strategy used to establish them is versatile, allowing verification of other functional choices and further broadening the scope of our results.

%\noindent\textcolor{red}{Sufficient conditions results enhance applicability}\\
%To make the asymptotic results readily applicable to various settings, we derive sufficient conditions under which this scale-free behavior emerges (Propositions~\ref{propClassFunc1} and \ref{propClassFunc2}). These conditions characterize classes of parametrized functions for resource allocation, flow distribution, and cost measurement under which our main theorem holds. Although not exhaustive, these sufficient conditions already cover a wide range of practical systems. We believe that the proof strategy used to establish these sufficient conditions can be extended to verify many other functional choices, thereby broadening the practical relevance of our framework.

%\noindent\textcolor{red}{Generalization of previous asymptotic results}\\
%Our framework balances modeling flexibility with analytical tractability, enabling rigorous asymptotic analysis. We generalize existing techniques used to analyze heavy-tailed cascade costs in power and traffic networks~\citep{Nesti_2020,Janicka_Sloothaak_Vlasiou_Zwart_2025} and derive the tail behavior of the cascade cost under mild assumptions (Theorem~\ref{mainThm}). 
Our theorem generalizes the asymptotic results of \citep{Nesti_2020,Janicka_Sloothaak_Vlasiou_Zwart_2025}, extending them to diverse applications and cost functions. \replaced{W}{Specifically, w}e prove that the probability of the cascade cost $\cascadeCost$ exceeding a large threshold $y$ decays as $l_\cascadeCost \cdot y^{-\alpha/\delta}$ and $y\rightarrow \infty$, where $\alpha$ and $\delta$ govern the tail of the vertex weight distribution and the scaling of the cost function, respectively. In contrast to previous studies \citep{Nesti_2020, Janicka_Sloothaak_Vlasiou_Zwart_2025}, which assumed linear cost functions ($\delta = 1$), our generalization shows how different cost structures change the tail behavior and allows for richer modeling of cascade costs across versatile networks. 

%\noindent\textcolor{red}{Scale-free behavior across flow networks}\\
Theorem~\ref{mainThm} also supports the hypothesis \deleted{first }proposed by \cite{Nesti_2020} that the scale-free nature of cascade costs may be inherited from scale-free input distributions, such as those reflecting city sizes. Our findings take this theory one step further by demonstrating that \replaced{this}{the same} emergence mechanism is not application-specific but arises universally across a wide range of flow networks. \replaced{O}{The o}perational differences between flow networks influence the constant $l_\cascadeCost$, but not the scaling exponent $\alpha/\delta$, emphasizing the robustness of this phenomenon. 

%\noindent\textcolor{red}{Explanation beyond existing heavy-tailed mechanisms}

Existing universal models provide important insights into heavy-tailed cascade behavior in networks, but none offer a comprehensive explanation across all contexts \citep{Sornette2006}. %Self-Organized Criticality models, which require slow demand growth and a fixed threshold \citep{Watkins_Pruessner_Chapman_Crosby_Jensen_2015}, fit certain natural or underinvested systems but less so engineered ones, like modern power grids, where capacity is managed in response to demand growth. 
Self-Organized Criticality requires slow demand driving, separation of time scales, and fixed thresholds---conditions more typical of natural than engineered systems like power grids---and consequently empirical evidence for SOC in real-world systems remains limited \citep{Watkins_Pruessner_Chapman_Crosby_Jensen_2015}. Percolation models, relying on local redistribution, suit spatial contagion processes such as epidemics or wildfires but not flow networks with global redistribution of flows \citep{Zhang18}. In other approaches, such as load–capacity or preferential-attachment models, the scale-free phenomenon is inherited from the scale-free degree distribution of the network (e.g., airline networks or the internet) \citep{Lai2004, Barabasi2000, Yao_Zhang_2023}; however, many flow systems, including traffic, power, and certain processing networks, often lack such a property. Our model fills the gap, offering a universal explanation for Pareto-tailed cascades in flow networks, showing that this behavior can also arise from an exogenous Pareto-tailed demand driver, regardless of the underlying topology.

%\noindent\textcolor{red}{Realism of the planning problem}

In most real-world networks, capacities result from the history of incremental planning and upgrades rather than from a single centralized optimization. The planning phase in our model is therefore an idealization, but one that aligns with common practice in cascade modeling, and is necessary to let capacities respond to probabilistic demand. Fixing capacities independently of demand would miss this essential dependence. By solving the planning problem with stochastic inputs, we approximate the capacity patterns that historical processes might produce under similar conditions, while preserving analytical tractability.

%\noindent\textcolor{red}{Limitations due to bilevel optimization}\\
Our choices for the source generation matrix $\genM^*$ and flow matrix $\flowM^*$ in Section~\ref{sec:satisfiability} give rise to a bilevel optimization problem. In general, the behavior of such problems is not well understood, except in cases where the problem exhibits a special structure \citep{Dempe_Kalashnikov_Pérez-Valdés_Kalashnykova_2015}. For this reason, in Proposition~\ref{propClassFunc1}, we restrict attention to the case $\beta=2$, which results in a quadratic inner optimization problem (minimal cost flow) and ensures the uniqueness of the optimal solution to the outer problem (resource allocation). In Example~\ref{exampleNonQuasi}, we demonstrate that for other values of $\beta$, the quasi-convexity of the flow function, a sufficient condition for uniqueness in the resource allocation problem, is no longer preserved. Thus, extending our proposition to general $\beta>1$ requires a deeper understanding of the bilevel structure and, in particular, the conditions that ensure uniqueness. Such a generalization would broaden the admissible class of flow mechanisms, enhancing the applicability of our framework.

%\noindent\textcolor{red}{Relaxation of scale-invariance}\\
%Two structural properties play a key role in our derivation of the sufficient conditions: continuity and scale invariance with respect to the vertex weight vector $\pareto$. Continuity is a natural and desirable property, ensuring that small perturbations in the input do not cause drastic changes in cascade behavior.
A structural property that plays a key role in our derivation of sufficient conditions is scale invariance with respect to the vertex weight vector $\pareto$. This property, however, may not be realistic in all settings, especially in systems where operational rules change with load intensity, represented here by $\pareto$. We argue that scale invariance is not strictly necessary; it may suffice to understand the asymptotic behavior of the cascade cost as $\max\{\pareto\}\rightarrow\infty$. Relaxing this assumption would require a more nuanced analysis to prove Proposition~\ref{sbj} and Theorem~\ref{mainThm}, but could significantly expand the scope of the model and admissible functional choices.

%\noindent\textcolor{red}{Extension to regularly varying distributions}\\
Furthermore, although we assume Pareto-tailed input weights $\pareto$, we expect that the results extend naturally to a broader class of regularly varying distributions, albeit at the cost of additional technical complexities in the proofs. Generalizing the input tail behavior in this way would further enhance the applicability of our framework to systems that fall outside the current assumptions.

%\noindent\textcolor{red}{Extension to time-dependent sink requirements} \\
Finally, our model assumes static sink requirements: demand remains fixed unless the loss of network connectivity renders it unattainable. In reality, however, some systems respond dynamically to failures—particularly over longer time scales—or exhibit inherently dynamic demand behavior. For example, in traffic networks, travelers may delay or cancel trips during heavy congestion, thereby altering sink demand, while in insurance networks, new claims arrive at each timestep, causing demand to evolve over time \citep{Blanchet_Shi_2012}. Extending our model to incorporate such dynamics could yield more realistic insights into phenomena such as cascade duration and time to recovery, and would also encompass other cascading systems, including insurance networks. Investigating whether the scale-free hypothesis holds in these dynamic settings is a promising direction for future research.

\bibliographystyle{unsrt}
\bibliography{sample}

@article{Nesti_2020,
	doi = {10.1103/physrevlett.125.058301},
  
	url = {https://doi.org/10.1103%2Fphysrevlett.125.058301},
  
	year = 2020,
	month = {July},
  
	publisher = {American Physical Society ({APS})},
  
	volume = {125},
  
	number = {5},
  
	author = {Tommaso Nesti and Fiona Sloothaak and Bert Zwart},
  
	title = {Emergence of Scale-Free Blackout Sizes in Power Grids},
  
	journal = {Physical Review Letters}
}

@article{Daqing_2014, title={Spatial Correlation Analysis of Cascading Failures: {C}ongestions and {B}lackouts}, volume={4}, url={http://dx.doi.org/10.1038/srep05381}, DOI={10.1038/srep05381}, number={1}, journal={Scientific Reports}, publisher={Springer Science and Business Media LLC}, author={Daqing, Li and Yinan, Jiang and Rui, Kang and Havlin, Shlomo}, year={2014}, month={June}, language={en} }

@article{SoC1987,
  title = {Self-Organized Criticality: An Explanation of the 1/f Noise},
  author = {Bak, Per and Tang, Chao and Wiesenfeld, Kurt},
  journal = {Physical Review Letters},
  volume = {59},
  issue = {4},
  pages = {381--384},
  numpages = {0},
  year = {1987},
  month = {July},
  publisher = {American Physical Society},
  doi = {10.1103/PhysRevLett.59.381},
  url = {https://link.aps.org/doi/10.1103/PhysRevLett.59.381}
}

@article{OPA2002,
author = {Carreras, Benjamin and Lynch, Vickie and Dobson, Ian and Newman, David E.},
year = {2003},
month = {January},
pages = {985--994},
title = {Critical Points and Transitions in an Electric Power Transmission Model for Cascading Failure Blackouts},
volume = {12},
journal = {Chaos},
doi = {10.1063/1.1505810}
}

@INPROCEEDINGS{Cascade2003,

  author={Dobson, Ian and Carreras, Benjamin A. and Newman, David E.},

  booktitle={Proceedings of the 36th Annual Hawaii International Conference on System Sciences}, 

  title={A Probabilistic Loading-Dependent Model of Cascading Failure and Possible Implications for Blackouts}, 

  year={2003},

  volume={},

  number={},

  doi={10.1109/HICSS.2003.1173909}}

@INPROCEEDINGS{Branching2004,

  author={Dobson, Ian and Carreras, Benjamin A. and Newman, David E.},

  booktitle={Proceedings of the 37th Annual Hawaii International Conference on System Sciences}, 

  title={A Branching Process Approximation to Cascading Load-Dependent System Failure}, 

  year={2004},

  volume={},

  number={},

  doi={10.1109/HICSS.2004.1265185}}

@inproceedings{manzano2014epidemic,
  title={Epidemic and Cascading Survivability of Complex Networks},
  author={Marc Manzano and Eusebi Calle and Jordi Ripoll and Anna Manolova Fagertun and Victor Torres-Padrosa and Sakshi Pahwa and Caterina Scoglio},
  year={2014},
  publisher={Institute of Electrical and Electronics Engineers (IEEE)},
booktitle={2014 6th International Workshop on Reliable Networks Design and Modeling (RNDM)}
}

@inproceedings{Manzano2013,
author = {Manzano, Marc and Calle, Eusebi and Ripoll, Jordi and Fagertun, Anna and Torres-Padrosa, Víctor},
year = {2013},
month = {January},
pages = {95--102},
title = {Epidemic Survivability: Characterizing Networks Under Epidemic-like Failure Propagation Scenarios},
booktitle = {2013 9th International Conference on the Design of Reliable Communication Networks, DRCN 2013}
}

@article{VC2011,
title = {Viral Conductance: Quantifying the Robustness of Networks with respect to Spread of Epidemics},
journal = {Journal of Computational Science},
volume = {2},
number = {3},
pages = {286--298},
year = {2011},
issn = {1877-7503},
doi = {https://doi.org/10.1016/j.jocs.2011.03.001},
url = {https://www.sciencedirect.com/science/article/pii/S1877750311000299},
author = {Mina Youssef and Robert Kooij and Caterina Scoglio},

}

@article{Barabasi2000, title={Error and Attack Tolerance of Complex Networks}, volume={406}, url={http://dx.doi.org/10.1038/35019019}, DOI={10.1038/35019019}, number={6794}, journal={Nature}, publisher={Springer Science and Business Media LLC}, author={Albert, Réka and Jeong, Hawoong and Barabási, Albert-László}, year={2000}, month={July}, pages={378--382}, language={en} }

@article{Pastor2001,
	doi = {10.1103/physreve.63.066117},
  
	url = {https://doi.org/10.1103%2Fphysreve.63.066117},
  
	year = 2001,
	month = {May},
  
	publisher = {American Physical Society ({APS})},
  
	volume = {63},
  
	number = {6},
  
	author = {Romualdo Pastor-Satorras and Alessandro Vespignani},
  
	title = {Epidemic Dynamics and Endemic States in Complex Networks},
  
	journal = {Physical Review E}
}

@article{Motter2002,
  title = {Cascade-based Attacks on Complex Networks},
  author = {Motter, Adilson E. and Lai, Ying-Cheng},
  journal = {Physical Review E},
  volume = {66},
  issue = {6},
  numpages = {4},
  year = {2002},
  month = {December},
  publisher = {American Physical Society},
  doi = {10.1103/PhysRevE.66.065102},
  url = {https://link.aps.org/doi/10.1103/PhysRevE.66.065102}
}

@article{Crucitti2004,
	doi = {10.1103/physreve.69.045104},
  
	url = {https://doi.org/10.1103%2Fphysreve.69.045104},
  
	year = 2004,
	month = {April},
  
	publisher = {American Physical Society ({APS})},
  
	volume = {69},
  
	number = {4},
  
	author = {Paolo Crucitti and Vito Latora and Massimo Marchiori},
  
	title = {Model for Cascading Failures in Complex Networks},
  
	journal = {Physical Review E}
}

@Inbook{Lai2004,
author="Lai, Ying-Cheng
and Motter, Adilson E.
and Nishikawa, Takashi",
title="Attacks and Cascades in Complex Networks",
bookTitle="Complex Networks",
year="2004",
publisher="Springer Berlin Heidelberg",
address="Berlin, Heidelberg",
pages={299--310},

isbn="978-3-540-44485-5",
doi="10.1007/978-3-540-44485-5_14",
url="https://doi.org/10.1007/978-3-540-44485-5_14"
}

@article{Guo2017,
title = {A Critical Review of Cascading Failure Analysis and Modeling of Power System},
journal = {Renewable and Sustainable Energy Reviews},
volume = {80},
pages = {9--22},
year = {2017},
issn = {1364-0321},
doi = {https://doi.org/10.1016/j.rser.2017.05.206},
url = {https://www.sciencedirect.com/science/article/pii/S1364032117308432},
author = {Hengdao Guo and Ciyan Zheng and Herbert Ho-Ching Iu and Tyrone Fernando}
}

@book{nair_wierman_zwart_2022, place={Cambridge}, series={Cambridge Series in Statistical and Probabilistic Mathematics}, title={The Fundamentals of Heavy Tails: Properties, Emergence, and Estimation}, DOI={10.1017/9781009053730}, publisher={Cambridge University Press}, author={Nair, Jayakrishnan and Wierman, Adam and Zwart, Bert}, year={2022}, collection={Cambridge Series in Statistical and Probabilistic Mathematics}}

@article{Huang_Vodenska_Havlin_Stanley_2013, title={Cascading Failures in Bi-partite Graphs: Model for Systemic Risk Propagation}, volume={3}, url={http://dx.doi.org/10.1038/srep01219}, DOI={10.1038/srep01219}, number={1}, journal={Scientific Reports}, publisher={Springer Science and Business Media LLC}, author={Huang, Xuqing and Vodenska, Irena and Havlin, Shlomo and Stanley, H. Eugene}, year={2013}, month={February}, language={en} }

@article{Mahdi2017,
    author = {Jalili, Mahdi and Perc, Matjaž},
    title = "{Information Cascades in Complex Networks}",
    journal = {Journal of Complex Networks},
    volume = {5},
    number = {5},
    pages = {665--693},
    year = {2017},
    month = {July},
    issn = {2051-1310},
    doi = {10.1093/comnet/cnx019},
    url = {https://doi.org/10.1093/comnet/cnx019},
    eprint = {https://academic.oup.com/comnet/article-pdf/5/5/665/20506399/cnx019.pdf},
}

@ARTICLE{DCPowerFlowRevisited,

  author={Stott, Brian and Jardim, Jorge and Alsac, Ongun},

  journal={IEEE Transactions on Power Systems}, 

  title={{DC} Power Flow Revisited}, 

  year={2009},

  volume={24},

  number={3},

  pages={1290--1300},

  doi={10.1109/TPWRS.2009.2021235}}

@book{boyd2004convex,
  title={Convex Optimization},
  author={Boyd, Stephen P. and Vandenberghe, Lieven},
  year={2004},
  publisher={Cambridge University Press}
}

@article{TONDEL2003489,
title = {An Algorithm for Multi-parametric Quadratic Programming and Explicit {MPC} solutions},
journal = {Automatica},
volume = {39},
number = {3},
pages = {489--497},
year = {2003},
issn = {0005-1098},
doi = {https://doi.org/10.1016/S0005-1098(02)00250-9},
url = {https://www.sciencedirect.com/science/article/pii/S0005109802002509},
author = {Petter Tøndel and Tor Arne Johansen and Alberto Bemporad}
}

@inproceedings{chen2001analysis,
  title={Analysis of Electric Power System Disturbance Data},
  author={Chen, Jie and Thorp, James S. and Parashar, Manu},
booktitle={Proceedings of the 34th Annual Hawaii International Conference on System Sciences}, 
  year={2001},
  volume={},
  number={},
  pages={738-744},
  doi={10.1109/HICSS.2001.926277}}

@article{Berry1961, author = {Berry, Brian J. L.}, volume={9}, url={http://dx.doi.org/10.1086/449923},  title = {City Size Distributions and Economic Development}, DOI={10.1086/449923}, number={4, Part 1}, journal={Economic Development and Cultural Change}, publisher={University of Chicago Press}, year={1961}, month={July}, pages={573--588}, language={en} }

@article{Watkins_Pruessner_Chapman_Crosby_Jensen_2015, title={25 Years of Self-Organized Criticality: Concepts and Controversies}, volume={198}, url={http://dx.doi.org/10.1007/s11214-015-0155-x}, DOI={10.1007/s11214-015-0155-x}, number={1–4}, journal={Space Science Reviews}, publisher={Springer Science and Business Media LLC}, author={Watkins, Nicholas W. and Pruessner, Gunnar and Chapman, Sandra C. and Crosby, Norma B. and Jensen, Henrik J.}, year={2015}, month={May}, pages={3-–44}, language={en} }

@article{janicka2024scalefree,
      title={Scale-free Cascading Failures: Generalized Approach for all Simple, Connected Graphs}, 
journal = {arXiv [math.PR]},
      author={Agnieszka Janicka and Fiona Sloothaak and Maria Vlasiou},
      year={2024},
      eprint={2403.11727},
      archivePrefix={arXiv },
      primaryClass={math.PR},
doi = { 	
https://doi.org/10.48550/arXiv.2403.11727
}
}

@book{Whittle_2007, place={Cambridge}, series={Cambridge Series in Statistical and Probabilistic Mathematics}, title={Networks: Optimisation and Evolution}, publisher={Cambridge University Press}, author={Whittle, Peter}, year={2007}, collection={Cambridge Series in Statistical and Probabilistic Mathematics}}

@article{Zhang_Zeng_Li_Huang_Stanley_Havlin_2019, title={Scale-free Resilience of Real Traffic Jams}, volume={116}, url={http://dx.doi.org/10.1073/pnas.1814982116}, DOI={10.1073/pnas.1814982116}, number={18}, journal={Proceedings of the National Academy of Sciences}, publisher={Proceedings of the National Academy of Sciences}, author={Zhang, Limiao and Zeng, Guanwen and Li, Daqing and Huang, Hai-Jun and Stanley, H. Eugene and Havlin, Shlomo}, year={2019}, month={April}, pages={8673--8678}, language={en} }

@article{Fleurquin_Ramasco_Eguiluz_2013, title={Systemic Delay Propagation in the {US} Airport Network}, volume={3}, url={http://dx.doi.org/10.1038/srep01159}, DOI={10.1038/srep01159}, number={1}, journal={Scientific Reports}, publisher={Springer Science and Business Media LLC}, author={Fleurquin, Pablo and Ramasco, José J. and Eguiluz, Victor M.}, year={2013}, month={January}, language={en} }

@book{wood2013power,
  title={Power Generation, Operation, and Control},
  author={Wood, Allen J. and Wollenberg, Bruce F. and Shebl{\'e}, Gerald B.},
  year={2013},
  publisher={John Wiley \& Sons}
}

@Article{hess14,
AUTHOR = {Hergarten, Stefan  and Winkler, Gerfried and Birk, Steffen},
TITLE = {Transferring the Concept of Minimum Energy Dissipation from River Networks to Subsurface Flow Patterns},
JOURNAL = {Hydrology and Earth System Sciences},
VOLUME = {18},
YEAR = {2014},
NUMBER = {10},
PAGES = {4277--4288},
URL = {https://hess.copernicus.org/articles/18/4277/2014/},
DOI = {10.5194/hess-18-4277-2014}
}

@book{Doyle_Snell_1984, series={Carus Mathematical Monographs}, title={Random Walks and Electric Networks}, publisher={Mathematical Association of America}, author={Doyle, Peter G. and Snell, J. Laurie}, year={1984}, collection={Carus Mathematical Monographs}}

@article{Janicka_Sloothaak_Vlasiou_Zwart_2025, title={Emergence of Scale-free Traffic Jams in Highway Networks: A Probabilistic Approach}, url={http://arxiv.org/abs/2502.13944}, DOI={10.48550/ARXIV.2502.13944}, journal={arXiv [physics.soc-ph]}, author={Janicka, Agnieszka and Sloothaak, Fiona and Vlasiou, Maria and Zwart, Bert}, year={2025}}

@inbook{crovella1998heavy,
  title={Heavy-tailed Probability Distributions in the {World Wide Web}},
  author={Crovella, Mark E. and Taqqu, Murad S. and Bestavros, Azer},
isbn = {0817639519},
publisher = {Birkhauser Boston Inc.},
address = {USA},
booktitle = {A Practical Guide to Heavy Tails: Statistical Techniques and Applications},
pages = {3–25},
numpages = {23},
year = {1998}
}

@book{landau1960mechanics,
  title={Mechanics},
  author={Landau, Lev Davidovich and Lifshits, Evgeni{\u\i} Mikha{\u\i}lovich},
  volume={1},
  year={1960},
publisher = {Butterworth-Heinemann},
edition = {Third},
address = {Oxford},
}

@article{Murray1926, title={The Physiological Principle of Minimum Work}, volume={12}, url={http://dx.doi.org/10.1073/pnas.12.3.207}, DOI={10.1073/pnas.12.3.207}, number={3}, journal={Proceedings of the National Academy of Sciences}, publisher={Proceedings of the National Academy of Sciences}, author={Murray, Cecil D.}, year={1926}, month={March}, pages={207--214}, language={en} }

@inbook{Singh2022, place={Cambridge}, title={Theory of Minimum Energy Dissipation Rate}, booktitle={Handbook of Hydraulic Geometry: Theories and Advances}, publisher={Cambridge University Press}, author={Singh, Vijay P.}, year={2022}, pages={450--469}}

@article{LI20211,
title = {Percolation on Complex Networks: Theory and Application},
journal = {Physics Reports},
volume = {907},
pages = {1--68},
year = {2021},
issn = {0370-1573},
doi = {https://doi.org/10.1016/j.physrep.2020.12.003},
url = {https://www.sciencedirect.com/science/article/pii/S0370157320304269},
author = {Ming Li and Run-Ran Liu and Linyuan Lü and Mao-Bin Hu and Shuqi Xu and Yi-Cheng Zhang}
}

@article{Watts2002, title={A Simple Model of Global Cascades on Random Networks}, volume={99}, url={http://dx.doi.org/10.1073/pnas.082090499}, DOI={10.1073/pnas.082090499}, number={9}, journal={Proceedings of the National Academy of Sciences}, publisher={Proceedings of the National Academy of Sciences}, author={Watts, Duncan J.}, year={2002}, month={April}, pages={5766--5771}, language={en} }

@article{Dhar1999, title={The {Abelian} Sandpile and Related Models}, volume={263}, url={http://dx.doi.org/10.1016/S0378-4371(98)00493-2}, DOI={10.1016/s0378-4371(98)00493-2}, number={1–4}, journal={Physica A: Statistical Mechanics and its Applications}, publisher={Elsevier BV}, author={Dhar, Deepak}, year={1999}, month={February}, pages={4--25}, language={en} }

@ARTICLE{Ren2018,
  author={Ren, Wendi and Wu, Jiajing and Zhang, Xi and Lai, Rong and Chen, Liang},
  journal={IEEE Transactions on Circuits and Systems II: Express Briefs}, 
  title={A Stochastic Model of Cascading Failure Dynamics in Communication Networks}, 
  year={2018},
  volume={65},
  number={5},
  pages={632--636},
  doi={10.1109/TCSII.2018.2822049}}

@article{tang2016complex,
  title={Complex Interdependent Supply Chain Networks: Cascading Failure and Robustness},
  author={Tang, Liang and Jing, Ke and He, Jie and Stanley, H. Eugene},
  journal={Physica A: Statistical Mechanics and its Applications},
  volume={443},
  pages={58--69},
  year={2016},
  publisher={Elsevier}
}

@article{Fu2023,
author = {Fu, Xiuwen and Pace, Pasquale and Aloi, Gianluca and Guerrieri, Antonio and Li, Wenfeng and Fortino, Giancarlo},
title = {Tolerance Analysis of Cyber-Manufacturing Systems to Cascading Failures},
year = {2023},
issue_date = {November 2023},
publisher = {Association for Computing Machinery},
address = {New York, NY, USA},
volume = {23},
number = {4},
issn = {1533-5399},
url = {https://doi.org/10.1145/3579847},
doi = {10.1145/3579847},
journal = {ACM Transactions on Internet Technology},
month = {November},
articleno = {50},
numpages = {23}
}

@ARTICLE{Montoya2019,
  author={Montoya, Oscar Danilo and Gil-González, Walter and Garces, Alejandro},
  journal={IEEE Transactions on Circuits and Systems II: Express Briefs}, 
  title={Optimal Power Flow on {DC} Microgrids: A Quadratic Convex Approximation}, 
  year={2019},
  volume={66},
  number={6},
  pages={1018--1022},
  doi={10.1109/TCSII.2018.2871432}}

@ARTICLE{Momoh1999,

  author={Momoh, James A. and Adapa, Ram and El-Hawary, Mohamed E.},

  journal={IEEE Transactions on Power Systems}, 

  title={A Review of Selected Optimal Power Flow Literature to 1993. {I.} {N}onlinear and Quadratic Programming Approaches}, 

  year={1999},

  volume={14},

  number={1},

  pages={96--104},

  doi={10.1109/59.744492}}

@book{ahuja1993network,
  title={Network Flows: Theory, Algorithms, and Applications},
  author={Ahuja, Ravindra K. and Magnanti, Thomas L. and Orlin, James B.},
  volume={1},
  year={1993},
  publisher={Prentice-Hall Englewood Cliffs, NJ}
}

@article{Labadie2004,
author = {John W. Labadie },
title = {Optimal Operation of Multireservoir Systems: State-of-the-Art Review},
journal = {Journal of Water Resources Planning and Management},
volume = {130},
number = {2},
pages = {93--111},
year = {2004},
doi = {10.1061/(ASCE)0733-9496(2004)130:2(93)},

URL = {https://ascelibrary.org/doi/abs/10.1061/%28ASCE%290733-9496%282004%29130%3A2%2893%29},
eprint = {https://ascelibrary.org/doi/pdf/10.1061/%28ASCE%290733-9496%282004%29130%3A2%2893%29}
}

@article{Wardrop1952,
author = {Wardrop, John G.},
title = {ROAD PAPER. {S}OME THEORETICAL ASPECTS OF ROAD TRAFFIC RESEARCH.},
journal = {Proceedings of the Institution of Civil Engineers},
volume = {1},
number = {3},
pages = {325--362},
year = {1952},
doi = {10.1680/ipeds.1952.11259},

URL = {

        https://doi.org/10.1680/ipeds.1952.11259



},
eprint = {

        https://doi.org/10.1680/ipeds.1952.11259}}

@inbook{Correa2011,
author = {Correa, José R. and Stier-Moses, Nicolás E.},
publisher = {John Wiley \& Sons, Ltd},
isbn = {9780470400531},
title = {Wiley Encyclopedia of Operations Research and Management Science},
chapter = {Wardrop Equilibria},
pages = {1--12},
doi = {https://doi.org/10.1002/9780470400531.eorms0962},
url = {https://onlinelibrary.wiley.com/doi/abs/10.1002/9780470400531.eorms0962},
eprint = {https://onlinelibrary.wiley.com/doi/pdf/10.1002/9780470400531.eorms0962},
year = {2011}
}

@book{us1964traffic,
  title={Traffic Assignment Manual for Application with a Large, High Speed Computer},
  author={{U.S. Bureau of Public Roads}},
  year={1964},
  publisher={US Department of Commerce},
  note = {Office of Planning. Urban Planning Division}
}

@article{1,
   author = {Wu, Jianjun and Sun, Huijun and Gao, Ziyou},
   doi = {10.1016/j.physa.2007.08.034},
   issn = {03784371},
   issue = {1},
   journal = {Physica A: Statistical Mechanics and its Applications},
   month = {December},
   pages = {407--413},
   title = {Cascading Failures on Weighted Urban Traffic Equilibrium Networks},
   volume = {386},
   year = {2007},
}

@article{katoh2024resource,
  title={Resource Allocation Problems},
  author={Katoh, Naoki and Shioura, Akiyoshi and Ibaraki, Toshihide},
  journal={Handbook of Combinatorial Optimization},
  pages={1--93},
  year={2024},
  publisher={Springer}
}

@article{PATRIKSSON20081,
title = {A Survey on the Continuous Nonlinear Resource Allocation Problem},
journal = {European Journal of Operational Research},
volume = {185},
number = {1},
pages = {1--46},
year = {2008},
issn = {0377-2217},
doi = {https://doi.org/10.1016/j.ejor.2006.12.006},
url = {https://www.sciencedirect.com/science/article/pii/S0377221706012215},
author = {Michael Patriksson}
}

@article{MONTOYA201918,
title = {Sequential Quadratic Programming Models for Solving the {OPF} Problem in {DC} Grids},
journal = {Electric Power Systems Research},
volume = {169},
pages = {18--23},
year = {2019},
issn = {0378-7796},
doi = {https://doi.org/10.1016/j.epsr.2018.12.008},
url = {https://www.sciencedirect.com/science/article/pii/S037877961830405X},
author = {Oscar Danilo Montoya and Walter Gil-González and Alejandro Garces}
}

@book{Dempe_Kalashnikov_Pérez-Valdés_Kalashnykova_2015, title={Bilevel Programming Problems}, url={http://dx.doi.org/10.1007/978-3-662-45827-3}, DOI={10.1007/978-3-662-45827-3}, journal={Energy Systems}, publisher={Springer Berlin Heidelberg}, author={Dempe, Stephan and Kalashnikov, Vyacheslav and Pérez-Valdés, Gerardo A. and Kalashnykova, Nataliya}, year={2015}, language={en} }

@article{Dobson_Carreras_Lynch_Newman_2007, title={Complex Systems Analysis of Series of Blackouts: Cascading Failure, Critical Points, and Self-Organization}, volume={17}, url={http://dx.doi.org/10.1063/1.2737822}, DOI={10.1063/1.2737822},  number={2}, journal={Chaos: An Interdisciplinary Journal of Nonlinear Science}, publisher={AIP Publishing}, author={Dobson, Ian and Carreras, Benjamin A. and Lynch, Vickie E. and Newman, David E.}, year={2007}, month={June}, language={en} }

@article{Daniotti_Servedio_Kager_Robben-Baldauf_Thurner_2024, title={Systemic Risk Approach to Mitigate Delay Cascading in Railway Networks}, volume={1}, url={http://dx.doi.org/10.1038/s44333-024-00012-6}, DOI={10.1038/s44333-024-00012-6}, number={1}, journal={NPJ Sustainable Mobility and Transport}, publisher={Springer Science and Business Media LLC}, author={Daniotti, Simone and Servedio, Vito D. P. and Kager, Johannes and Robben-Baldauf, Aad and Thurner, Stefan}, year={2024}, month={December}, language={en} }

@article{Liu_Yang_2021, title={ASYMPTOTICS FOR SYSTEMIC RISK WITH DEPENDENT HEAVY-TAILED LOSSES}, volume={51}, DOI={10.1017/asb.2021.11}, number={2}, journal={ASTIN Bulletin}, author={Liu, Jiajun and Yang, Yang}, year={2021}, pages={571--605}}

@book{ibragimov2015heavy,
  title={Heavy-Tailed Distributions and Robustness in Economics and Finance},
  author={Ibragimov, Marat and Ibragimov, Rustam and Walden, Johan},
  volume={214},
  year={2015},
  publisher={Springer}
}

@article{Duan2023,
   author = {Jinxiao Duan and Daqing Li and Hai Jun Huang},
   doi = {10.1016/j.trc.2023.104017},
   issn = {0968090X},
   journal = {Transportation Research Part C: Emerging Technologies},
   month = {February},
   publisher = {Elsevier Ltd},
   title = {Reliability of the Traffic Network Against Cascading Failures with Individuals Acting Independently or Collectively},
   volume = {147},
   year = {2023},
}

@article{Frank01122016,
author = {Stephen Frank and Steffen Rebennack},
title = {An Introduction to Optimal Power Flow: Theory, Formulation, and Examples},
journal = {IIE Transactions},
volume = {48},
number = {12},
pages = {1172--1197},
year = {2016},
publisher = {Taylor \& Francis},
doi = {10.1080/0740817X.2016.1189626},


URL = {

        https://doi.org/10.1080/0740817X.2016.1189626
},
eprint = {

        https://doi.org/10.1080/0740817X.2016.1189626

}

}

@inproceedings{Baker21,
author = {Baker, Kyri},
title = {Solutions of {DC} {OPF} are Never {AC} Feasible},
year = {2021},
isbn = {9781450383332},
publisher = {Association for Computing Machinery},
address = {New York, NY, USA},
url = {https://doi.org/10.1145/3447555.3464875},
doi = {10.1145/3447555.3464875},
booktitle = {Proceedings of the Twelfth ACM International Conference on Future Energy Systems},
pages = {264--268},
numpages = {5},
location = {Virtual Event, Italy},
series = {e-Energy '21}
}

@ARTICLE{Molzahn2017,
  author={Molzahn, Daniel K. and Dörfler, Florian and Sandberg, Henrik and Low, Steven H. and Chakrabarti, Sambuddha and Baldick, Ross and Lavaei, Javad},
  journal={IEEE Transactions on Smart Grid}, 
  title={A Survey of Distributed Optimization and Control Algorithms for Electric Power Systems}, 
  year={2017},
  volume={8},
  number={6},
  pages={2941--2962},
  doi={10.1109/TSG.2017.2720471}}

@article{Leland93,
author = {Leland, Will E. and Taqqu, Murad S. and Willinger, Walter and Wilson, Daniel V.},
title = {On the Self-Similar Nature of {E}thernet Traffic},
year = {1993},
issue_date = {Oct. 1993},
publisher = {Association for Computing Machinery},
address = {New York, NY, USA},
volume = {23},
number = {4},
issn = {0146-4833},
url = {https://doi.org/10.1145/167954.166255},
doi = {10.1145/167954.166255},
journal = {ACM SIGCOMM Computer Communication Review},
month = {October},
pages = {183--193},
numpages = {11}
}

@article{resnick1997heavy,
  title={Heavy Tail Modeling and Teletraffic Data: Special Invited Paper},
  author={Resnick, Sidney I.},
  journal={The Annals of Statistics},
  volume={25},
  number={5},
  pages={1805--1869},
  year={1997},
  publisher={Institute of Mathematical Statistics}
}

@article{Blanchet_Shi_2012, title={Stochastic Risk Networks: Modeling, Analysis and Efficient {M}onte {C}arlo}, url={http://dx.doi.org/10.2139/ssrn.2012987}, DOI={10.2139/ssrn.2012987}, journal={SSRN Electronic Journal}, publisher={Elsevier BV}, author={Blanchet, Jose and Shi, Yixi}, year={2012}, language={en} }

@article{Thurner_Farmer_Geanakoplos_2010, title={Leverage Causes Fat Tails and Clustered Volatility}, url={http://dx.doi.org/10.2139/ssrn.1534648}, DOI={10.2139/ssrn.1534648}, journal={SSRN Electronic Journal}, publisher={Elsevier BV}, author={Thurner, Stefan and Farmer, J. Doyne and Geanakoplos, John}, year={2010}, language={en} }

@article{Moritz_Morais_Summerell_Carlson_Doyle_2005, title={Wildfires, Complexity, and Highly Optimized Tolerance}, volume={102}, url={http://dx.doi.org/10.1073/pnas.0508985102}, DOI={10.1073/pnas.0508985102},  number={50}, journal={Proceedings of the National Academy of Sciences}, publisher={Proceedings of the National Academy of Sciences}, author={Moritz, Max A. and Morais, Marco E. and Summerell, Lora A. and Carlson, J. M. and Doyle, John}, year={2005}, month={December}, pages={17912–17917}, language={en} }

@article{Turcotte_Malamud_Guzzetti_Reichenbach_2002, title={Self-Organization, the Cascade Model, and Natural Hazards}, volume={99}, url={http://dx.doi.org/10.1073/pnas.012582199}, DOI={10.1073/pnas.012582199}, journal={Proceedings of the National Academy of Sciences}, publisher={Proceedings of the National Academy of Sciences}, author={Turcotte, Donald L. and Malamud, Bruce D. and Guzzetti, Fausto and Reichenbach, Paola}, year={2002}, month={February}, pages={2530–2537}, language={en} }

@article{CNCFreview,
    author = {Valdez, Lucas D. and Shekhtman, Louis and La Rocca, Cristian E. and Zhang, Xin and Buldyrev, Sergey V. and Trunfio, Paul A. and Braunstein, Lidia A. and Havlin, Shlomo},
    title = "{Cascading Failures in Complex Networks}",
    journal = {Journal of Complex Networks},
    volume = {8},
    number = {2},
    year = {2020},
    month = {May},
    doi = {10.1093/comnet/cnaa013},
    url = {https://doi.org/10.1093/comnet/cnaa013},
    eprint = {https://academic.oup.com/comnet/article-pdf/8/2/cnaa013/33582729/cnaa013.pdf},
}

@ARTICLE{IoTReview2021,

  author={Xing, Liudong},

  journal={IEEE Internet of Things Journal}, 

  title={Cascading Failures in Internet of Things: Review and Perspectives on Reliability and Resilience}, 

  year={2021},

  volume={8},

  number={1},

  pages={44-64},

  doi={10.1109/JIOT.2020.3018687}}

@article{Elliott_Golub_Jackson_2012, title={Financial Networks and Contagion}, url={http://dx.doi.org/10.2139/ssrn.2175056}, DOI={10.2139/ssrn.2175056}, journal={SSRN Electronic Journal}, publisher={Elsevier BV}, author={Elliott, Matthew and Golub, Benjamin and Jackson, Matthew O.}, year={2012}, language={en} }

@article{Branching2002,
  title = {Critical and Near-critical Branching Processes},
  author = {Adami, Christoph and Chu, Johan},
  journal = {Physical Review E},
  volume = {66},
  issue = {1},
  pages = {011907},
  numpages = {8},
  year = {2002},
  month = {July},
  publisher = {American Physical Society},
  doi = {10.1103/PhysRevE.66.011907},
  url = {https://link.aps.org/doi/10.1103/PhysRevE.66.011907}
}

@article{Gang_Jia_Herniter_2022, title={Sand and Fire: {A}pplying the Sandpile Model of Self-Organised Criticality to Wildfire Mitigation}, volume={31}, url={http://dx.doi.org/10.1071/wf22017}, DOI={10.1071/wf22017},  number={9}, journal={International Journal of Wildland Fire}, publisher={CSIRO Publishing}, author={Gang, Joshua E. and Jia, Wanqi and Herniter, Ira A.}, year={2022}, month={August}, pages={847–856}, language={en} }

@article{PhysRevLett.69.1629,
  title = {Self-Organized Critical Forest-fire Model},
  author = {Drossel, Barbara and Schwabl, Franz},
  journal = {Physical Review Letters},
  volume = {69},
  issue = {11},
  pages = {1629--1632},
  numpages = {0},
  year = {1992},
  month = {September},
  publisher = {American Physical Society},
  doi = {10.1103/PhysRevLett.69.1629},
  url = {https://link.aps.org/doi/10.1103/PhysRevLett.69.1629}
}

@article{Beggs_Plenz_2003, title={Neuronal Avalanches in Neocortical Circuits}, volume={23}, url={http://dx.doi.org/10.1523/JNEUROSCI.23-35-11167.2003}, DOI={10.1523/jneurosci.23-35-11167.2003}, number={35}, journal={The Journal of Neuroscience}, publisher={Society for Neuroscience}, author={Beggs, John M. and Plenz, Dietmar}, year={2003}, month={December}, pages={11167–11177}, language={en} }

@article{PhysRevLett.68.1244,
  title = {Self-Organized Criticality in a Continuous, Nonconservative Cellular Automaton Modeling Earthquakes},
  author = {Olami, Zeev and Feder, Hans Jacob S. and Christensen, Kim},
  journal = {Physical Review Letters},
  volume = {68},
  issue = {8},
  pages = {1244--1247},
  numpages = {0},
  year = {1992},
  month = {February},
  publisher = {American Physical Society},
  doi = {10.1103/PhysRevLett.68.1244},
  url = {https://link.aps.org/doi/10.1103/PhysRevLett.68.1244}
}

@article{Bak_Paczuski_Shubik_1997, title={Price Variations in a Stock Market with Many Agents}, volume={246}, url={http://dx.doi.org/10.1016/s0378-4371(97)00401-9}, DOI={10.1016/s0378-4371(97)00401-9}, number={3–4}, journal={Physica A: Statistical Mechanics and its Applications}, publisher={Elsevier BV}, author={Bak, Per and Paczuski, Maya and Shubik, Martin}, year={1997}, month={December}, pages={430–453}, language={en} }

@inbook{Saharidis11,
author = {Saharidis, Georgios},
year = {2011},
month = {April},
pages = {3--26},
title = {Supply Chain Optimization: Centralized vs Decentralized Planning and Scheduling},
isbn = {978-953-307-184-8},
doi = {10.5772/15860},
publisher={IntechOpen London, UK}
}

@INPROCEEDINGS{Parandehgheibi2014,

  author={Parandehgheibi, Marzieh and Modiano, Eytan and Hay, David},

  booktitle={2014 IEEE International Conference on Smart Grid Communications}, 

  title={Mitigating Cascading Failures in Interdependent Power Grids and Communication Networks}, 

  year={2014},

  volume={},

  number={},

  pages={242-247},

  doi={10.1109/SmartGridComm.2014.7007653}}

@ARTICLE{Hosseinalipour2020,
  author={Hosseinalipour, Seyyedali and Mao, Jiayu and Eun, Do Young and Dai, Huaiyu},
  journal={IEEE Transactions on Network Science and Engineering}, 
  title={Prevention and Mitigation of Catastrophic Failures in Demand-Supply Interdependent Networks}, 
  year={2020},
  volume={7},
  number={3},
  pages={1710-1723},
  doi={10.1109/TNSE.2019.2951084}}

@ARTICLE{Korilis97,

  author={Korilis, Yannis A. and Lazar, Aurel A. and Orda, Ariel},

  journal={IEEE/ACM Transactions on Networking}, 

  title={Achieving Network Optima using {S}tackelberg Routing Strategies}, 

  year={1997},

  volume={5},

  number={1},

  pages={161-173},

  doi={10.1109/90.554730}}

@article{kelly98,
  title={Rate Control for Communication Networks: Shadow Prices, Proportional Fairness and Stability},
  author={Kelly, Frank P. and Maulloo, Aman K. and Tan, David Kim Hong},
  journal={Journal of the Operational Research society},
  volume={49},
  number={3},
  pages={237--252},
  year={1998},
  publisher={Taylor \& Francis}
}

@article{Johari04,
author = {Johari, Ramesh and Tsitsiklis, John N.},
title = {Efficiency Loss in a Network Resource Allocation Game},
year = {2004},
issue_date = {August 2004},
publisher = {INFORMS},
address = {Linthicum, MD, USA},
volume = {29},
number = {3},
issn = {0364-765X},
url = {https://doi.org/10.1287/moor.1040.0091},
doi = {10.1287/moor.1040.0091},
journal = {Mathematics of Operations Research},
pages = {407–435},
numpages = {29}
}

@Inbook{Sornette2006,
title="Mechanisms for Power Laws",
bookTitle="Critical Phenomena in Natural Sciences: Chaos, Fractals, Selforganization and Disorder: Concepts and Tools",
year="2006",
publisher="Springer Berlin Heidelberg",
address="Berlin, Heidelberg",
pages="345--394",
isbn="978-3-540-33182-7",
doi="10.1007/3-540-33182-4_14",
url="https://doi.org/10.1007/3-540-33182-4_14",
author = "Didier Sornette"
}

@article{Zhang18,
  title = {Cascading Failures in Interdependent Systems under a Flow Redistribution Model},
  author = {Zhang, Yingrui and Arenas, Alex and Ya\ifmmode \breve{g}\else \u{g}\fi{}an, Osman},
  journal = {Physical Review E},
  volume = {97},
  issue = {2},
  pages = {022307},
  numpages = {13},
  year = {2018},
  month = {February},
  publisher = {American Physical Society},
  doi = {10.1103/PhysRevE.97.022307},
  url = {https://link.aps.org/doi/10.1103/PhysRevE.97.022307}
}

@article{Yao_Zhang_2023, title={Critical and Steady-state Characteristics of Delay Propagation in an Airport Network}, volume={18}, ISSN={1932-6203}, url={http://dx.doi.org/10.1371/journal.pone.0288200}, DOI={10.1371/journal.pone.0288200}, number={7}, journal={PlOS One}, author={Yao, Hong-Guang and Zhang, Hang}, year={2023}, language={en} }

@article{Parent2019,
title = {An Analysis of Enhanced Tree Trimming Effectiveness on Reducing Power Outages},
journal = {Journal of Environmental Management},
volume = {241},
pages = {397-406},
year = {2019},
issn = {0301-4797},
doi = {https://doi.org/10.1016/j.jenvman.2019.04.027},
url = {https://www.sciencedirect.com/science/article/pii/S0301479719304840},
author = {Jason R. Parent and Thomas H. Meyer and John C. Volin and Robert T. Fahey and Chandi Witharana}
}

\newpage

\appendix
    \section{Proofs of results in Section~\ref{mainRes}}
\subsection{Proof of Proposition \ref{sbj}} \label{EC:propProof}
\begin{proof}
From Assumption \ref{assumption}, we know that $\cascadeCost\leq \boundZ\cdot \sum_{v\in \V} \paretoI{v}^\delta$ for some constant $\boundZ>0$. Hence,
\[\cascadeCost\leq \boundZ\cdot \sum_{v\in \V} \paretoI{v}^\delta \leq \nv\cdot \boundZ\cdot \paretoI{\max}^\delta.\]
Let $\paretoI{\max - 1}$ denote the second largest vertex weight. Using the above bound on $\cascadeCost$, we obtain
\begin{align*}
    &\PR{\cascadeCost >y, \sum_{i = 1}^{\nv} \paretoI{i}^\delta > (1+\varepsilon) \paretoI{\max}^\delta }\leq \PR{\nv\cdot \boundZ\cdot \paretoI{\max}^\delta>y, \sum_{i = 1}^{\nv} \paretoI{i}^\delta - \paretoI{\max}^\delta > \varepsilon \paretoI{\max}^\delta }\\
    &\qquad\leq \PR{\paretoI{\max}^\delta>y/(\nv\cdot \boundZ), \sum_{i = 1}^{\nv} \paretoI{i}^\delta - \paretoI{\max}^\delta> \varepsilon y/(\nv\cdot \boundZ)}\\
    &\qquad\overset{(*)}{\leq}  \PR{\paretoI{\max}^\delta>y/(\nv\cdot \boundZ), (\nv-1)\paretoI{\max-1}^\delta >\varepsilon y/(\nv\cdot \boundZ)}\\&\qquad= \PR{\paretoI{\max}>\left(y/(\nv\cdot \boundZ)\right)^{1/\delta}, \paretoI{\max-1} >\left(\varepsilon y/(\nv\cdot \boundZ\cdot (\nv-1))\right)^{1/\delta}},
\end{align*}
\newcommand{\constSBJ}{\textcolor{black}{\kappa}}
where in $(*)$ we bound all but the largest vertex weight by $\paretoI{\max - 1}$. Furthermore, for any constants $\constSBJ_{\max}, \constSBJ_{\max - 1}>0$, 
\begin{align*}&\PR{\paretoI{\max}>\constSBJ_{\max}\cdot y^{1/\delta}, \paretoI{\max-1}> \constSBJ_{\max - 1} \cdot y^{1/\delta}} \\&\qquad\overset{(**)}{\leq} \nv(\nv - 1)\PR{\paretoI{1}> \constSBJ_{\max}\cdot y^{1/\delta}, \paretoI{2}>\constSBJ_{\max - 1}\cdot  y^{1/\delta}, \paretoI{1}\geq \paretoI{2}\geq \paretoI{i},~ i = 3,\dots, \nv}  \\
&\qquad \leq \nv(\nv - 1)\PR{\paretoI{1}> \constSBJ_{\max}\cdot y^{1/\delta}, \paretoI{2}>\constSBJ_{\max - 1} \cdot y^{1/\delta}} \\
&\qquad \overset{X_1, X_2 \text{ indep.}}{=}\nv(\nv - 1)\PR{\paretoI{1}> \constSBJ_{\max}\cdot y^{1/\delta}}\PR{ \paretoI{2}>\constSBJ_{\max - 1} \cdot y^{1/\delta}} = \OO{y^{-2\alpha/\delta}} \quad \text{as }y\rightarrow \infty,\end{align*}
where in $(**)$ we use the fact that there are $\nv(\nv - 1)$ possibilities for the two random variables to be the largest, each occurring with equal probability.
Thus, we conclude that 
\[ \PR{\cascadeCost>y, \sum_{i = 1}^{\nv} \paretoI{i}^\delta > (1+\varepsilon) \paretoI{\max}^\delta} =  \OO{y^{-2\alpha/\delta}} \quad \text{as }y\rightarrow \infty. \qedhere\]
\end{proof}

\subsection{Proof of Theorem~\ref{mainThm}}\label{EC:thmProof}

\begin{proof}
%Consider the setting, as described in Theorem \ref{mainThm}. First (\hyperlink{stepA}{Step A}), we observe that the total cascade cost depends on the cascade sequence $\cascadeSeqDet\in \cascadeSet$ that occurred. Hence, we write
%\[\PR{\cascadeCost>y} = \sum_{\cascadeSeqDet\in \cascadeSet} \PR{\cascadeCost>y, \cascadeSeqRand = \cascadeSeqDet}.\]
\hyperlink{stepA}{Step A}:  First, we partition the event $\PR{\cascadeCost>y}$ into disjoint cases, each defined by the particular realization of the cascade sequence $\cascadeSeqRand = \cascadeSeqDet$ for some $\cascadeSeqDet\in \cascadeSet$ and by a unique index $i\in \V$ such that $\paretoI{i}$ is the largest component of $\pareto$. Hence, we obtain
\[\PR{\cascadeCost>y} = \sum_{i\in \V}\sum_{\cascadeSeqDet\in \cascadeSet} \PR{\cascadeCost>y, \cascadeSeqRand = \cascadeSeqDet, \paretoI{j}< \paretoI{i}; \, \forall j\in \V\setminus\{i\}},\]
where we note that this is an equality since the probability that two independent Pareto-distributed random variables are equal occurs with probability zero.
%Next, for $i\in \V$, we additionally distinguish between $\nv$ disjoint events $\{\paretoI{j}\leq \paretoI{i};\, \forall ~j\in\V/\{i\}\}$. Note that $\paretoI{i}$'s are i.i.d.; hence, each of the $\nv$ events occurs equally likely. 

\hyperlink{stepB}{Step~B}: Fix $\varepsilon\in (0,\frac{1}{2})$, and observe that
\begin{align*}
    \PR{\cascadeCost>y, \cascadeSeqRand = \cascadeSeqDet}&= \sum_{i=1}^{\nv} \PR{\cascadeCost>y, \cascadeSeqRand=\cascadeSeqDet, \sum_{j\neq i} \paretoI{j}^\delta \leq \varepsilon \paretoI{i}^\delta , \paretoI{j}\leq \paretoI{i};\, \forall ~j\in\V/\{i\}} \\&\qquad+ \underbrace{\sum_{i=1}^{\nv} \PR{\cascadeCost>y, \cascadeSeqRand=\cascadeSeqDet, \sum_{j\neq i} \paretoI{j}^\delta > \varepsilon \paretoI{i}^\delta , \paretoI{j}\leq \paretoI{i} ;\, \forall ~j\in\V/\{i\}}}_{\mathcal{O}(y^{-2\alpha/\delta}) \text{ by Prop.~\ref{sbj}}}\\
    &= \sum_{i=1}^{\nv} \underbrace{\PR{\cascadeCost>y, \cascadeSeqRand=\cascadeSeqDet, \sum_{j\neq i} \paretoI{j}^\delta \leq \varepsilon \paretoI{i}^\delta}}_{(\star)} + \mathcal{O}\left(y^{-2\alpha/\delta}\right),
\end{align*}
where the latter equality follows because $\varepsilon<\frac{1}{2}<1$.

\hyperlink{stepC}{Step~C}: Next, we consider the term $(\star)$ for some $i\in \V$. We condition on the value of $\paretoI{i}$, the event $\left\{\sum_{j\neq i}\paretoI{j}^\delta\leq \varepsilon \paretoI{i}^\delta \right\}$, and the occurrence of cascade $\cascadeSeqDet\in \cascadeSet$. The first conditioning requires a lower bound on the value of $\paretoI{i}^\delta $, which we derive next. By Assumption \ref{assumption}, $\cascadeCost$ is bounded from above by $\boundZ\cdot\sum_{i\in \V}\paretoI{i}^\delta$. Hence, if $\paretoI{i}^\delta \leq \frac{\varepsilon}{\boundZ}y$ and $\sum_{j\neq i} \paretoI{j}^\delta \leq \varepsilon \paretoI{i}^\delta $, then $$\sum_{i\in\V} \paretoI{i}^\delta \leq (1+\varepsilon) \paretoI{i}^\delta \leq (1+\varepsilon)\frac{\varepsilon}{\boundZ}y,$$ and therefore  \[\cascadeCost\leq \boundZ\cdot\left(\frac{\varepsilon}{\boundZ} +   \frac{\varepsilon^2}{\boundZ}\right)y \overset{\varepsilon\in (0,\frac{1}{2})}{\leq} \boundZ\left(\frac{1}{2\boundZ} + \frac{1}{4\boundZ}\right) y < y.\] Thus,  
\begin{equation}\PR{\cascadeCost>y,\, \cascadeSeqRand=\cascadeSeqDet, \,\sum_{j\neq i} \paretoI{j}^\delta \leq \varepsilon \paretoI{i}^\delta,\, \paretoI{i}^{\delta}<\frac{\varepsilon}{\boundZ}y } = 0.\label{eq:prob0}\end{equation}
In other words, the probability in consideration is only non-zero if $\paretoI{i}\geq\left(\frac{\varepsilon}{\boundZ} y\right)^{1/\delta}$. Proceeding with the conditioning on the value of $\paretoI{i}^\delta$ yields
\begin{align*}
    &\PR{\cascadeCost>y, \cascadeSeqRand=\cascadeSeqDet, \sum_{j\neq i} \paretoI{j}^\delta\leq \varepsilon \paretoI{i}^\delta}\\
    &\qquad = \int_{0}^\infty \PR{\cascadeCost>y, \cascadeSeqRand=\cascadeSeqDet, \sum_{j\neq i} \paretoI{j}^\delta\leq \varepsilon z^\delta y ~|~ \paretoI{i}^\delta = z^\delta y} \cdot \PR{\frac{\paretoI{i}}{y^{1/\delta}} \in \textnormal{d} z}\\
    &\qquad\overset{\eqref{eq:prob0}}{=}\int_{\left(\frac{\varepsilon}{\boundZ}\right)^{1/\delta}}^\infty \PR{\cascadeCost>y, \cascadeSeqRand=\cascadeSeqDet, \sum_{j\neq i} \paretoI{j}^\delta\leq \varepsilon z^\delta y ~|~ \paretoI{i}^\delta = z^\delta y} \cdot \PR{\frac{\paretoI{i}}{y^{1/\delta}} \in \textnormal{d} z}.
\end{align*}
Next, we condition on the remaining entries of $\pareto$. Let $\mathcal{X}_i(z) = \{\bm{x}\in \R_+^{\nv} : \sum_{j \neq i} x_j^\delta \leq \varepsilon x_i^\delta, x_i^\delta = z^\delta y\}$. Conditioning on all possible values of $\bm{x}\in\mathcal{X}_i(z)$, we obtain
\begin{multline*}
    \PR{\cascadeCost>y, \cascadeSeqRand=\cascadeSeqDet, \sum_{j\neq i} \paretoI{j}^\delta\leq \varepsilon \paretoI{i}^\delta}\\
    =\int_{\left(\frac{\varepsilon}{\boundZ}\right)^{1/\delta}}^\infty\int_{\bm{x}\in \mathcal{X}_i(z)} \PR{\cascadeCost>y, \cascadeSeqRand=\cascadeSeqDet  ~|~ \paretoI{i}^\delta = z^\delta y, \pareto = \bm{x}}\\\cdot \PR{\pareto\in \textnormal{d} \bm{x}~|~ \paretoI{i}^\delta = z^\delta y} \PR{\frac{\paretoI{i}}{y^{1/\delta}} \in \textnormal{d} z}.
\end{multline*}
Then, conditioning on $\cascadeSeqRand=\cascadeSeqDet$, we obtain
\begin{multline*}
    \PR{\cascadeCost>y, \cascadeSeqRand=\cascadeSeqDet, \sum_{j\neq i} \paretoI{j}^\delta\leq \varepsilon \paretoI{i}^\delta}\\ =\int_{\left(\frac{\varepsilon}{\boundZ}\right)^{1/\delta}}^\infty \int_{\bm{x}\in \mathcal{X}_i(z)}\PR{\cascadeCost>y  ~|~ \paretoI{i}^\delta = z^\delta y, \pareto = \bm{x}, \cascadeSeqRand=\cascadeSeqDet} \PR{\cascadeSeqRand=\cascadeSeqDet~|~ \paretoI{i}^\delta = z^\delta y, \pareto = \bm{x}}\\\cdot
    \PR{\pareto\in \textnormal{d}\bm{x}~|~ \paretoI{i}^\delta = z^\delta y} \PR{\frac{\paretoI{i}}{y^{1/\delta}} \in \textnormal{d} z}.
\end{multline*}

\hyperlink{stepD}{Step~D}: By Assumption~\ref{assumption}, for a given cascade $\cascadeSeqRand = \cascadeSeqDet$ and a given vertex weight vector $\pareto = \bm{x}$, the cascade cost $\detZ(\bm{x}, \cascadeSeqDet)$ is $\delta$-scale-invariant in $\bm{x}$. Moreover, the probability of observing cascade $\cascadeSeqRand=\cascadeSeqDet$ given $\pareto = \omega \bm{x}$ is the same for all $\omega>0$. We apply these two properties to the expression above and obtain
\begin{multline}
    \PR{\cascadeCost>y, \cascadeSeqRand=\cascadeSeqDet, \sum_{j\neq i} \paretoI{j}^\delta\leq \varepsilon \paretoI{i}^\delta}\\ =\int_{\left(\frac{\varepsilon}{\boundZ}\right)^{1/\delta}}^\infty \int_{\bm{x\in \mathcal{X}_i(z)}} \PR{\cascadeCost>z^{-\delta}  ~|~ \paretoI{i} = 1, \pareto = \bm{x}/(z^{\delta} y), \cascadeSeqRand=\cascadeSeqDet} \PR{\cascadeSeqRand=\cascadeSeqDet~|~ \paretoI{i} = 1, \pareto = \bm{x}/(z^{\delta} y)}\\\cdot\PR{\pareto \in \textnormal{d}\bm{x}~|~ \paretoI{i}^\delta = z^\delta y} \PR{\frac{\paretoI{i}}{y^{1/\delta}} \in \textnormal{d} z}.
\label{eq:condZ}\end{multline}

%\begin{align*}
%    \mathbb{P}(S>x, [s], &\sum_{j\neq i} \paretoI{j}\leq \varepsilon \paretoI{i})= \int_{\frac{x_{\min}}{x}}^\infty \mathbb{P}(S>x, [s], \sum_{j\neq i} \paretoI{j}\leq \varepsilon \paretoI{i}~|~ \paretoI{i} = yx) f_{\paretoI{i}}(yx) \,x\, dy\\
%    &= \int_{\frac{x_{\min}}{x}}^\infty \mathbb{P}(S>x, [s]~|~ \paretoI{i} = yx, \sum_{j\neq i} \paretoI{j}\leq \varepsilon yx) \mathbb{P}(\sum_{j\neq i} \paretoI{j}<\varepsilon yx~|~\paretoI{i} = yx)f_{\paretoI{i}}(yx) \,x\, dy\\
%   &= \int_{\frac{x_{\min}}{x}}^\infty \mathbb{P}(S>x~|~ \paretoI{i} = yx, \sum_{j\neq i} \paretoI{j}\leq \varepsilon yx, [s]) \\&\hspace{2.2cm}\cdot \mathbb{P}([s]~|~\paretoI{i} = yx, \sum_{j\neq i} \paretoI{j}\leq \varepsilon yx)\mathbb{P}(\sum_{j\neq i} \paretoI{j}<\varepsilon yx~|~\paretoI{i} = yx)f_{\paretoI{i}}(yx) \,x\, dy.
%\end{align*}

Next, we evaluate the term $\PR{\cascadeCost>z^{-\delta} ~|~ \paretoI{i}= 1, \pareto = \bm{x}/(z^{\delta} y), \cascadeSeqRand=\cascadeSeqDet}$. Let $\tilde{\mathcal{X}}_i(\varepsilon)$ denote the set $\mathcal{X}_i(z)$ rescaled by $z^{\delta} y$, i.e., $$\tilde{\mathcal{X}}_i(\varepsilon):= \{\tilde{\bm{x}} : z^\delta y\tilde{\bm{x}}\in \mathcal{X}_i(z)\} = \{\tilde{\bm{x}}\in \R_+^{\nv}: \sum_{j\neq i} \tilde{x}_j^\delta\leq \varepsilon, \tilde{x}_i = 1\}.$$ 
After rescaling, the family $\{\tilde{\mathcal X}_i(\varepsilon)\}_{\varepsilon>0}$ carries no explicit dependence on $z$.
We note, however, the explicit dependence on $\varepsilon$, since in what follows we consider the limit $\varepsilon\downarrow 0$.

Recall that $\detZ(\tilde{\bm{x}}, \cascadeSeqDet)$ denotes the \textit{deterministic} cascade cost given $\cascadeSeqDet$ and $\pareto = \tilde{\bm{x}}$. By Assumption~\ref{assumption}, $\detZ(\tilde{\bm{x}},\cascadeSeqDet)$ is SRC with respect to $\tilde{\bm{x}}$, which means that for $\tilde{\bm{x}}(\varepsilon)\in \tilde{\mathcal{X}}_i(\varepsilon)$,
\[\lim_{\varepsilon \downarrow 0} \detZ( \tilde{\bm{x}}(\varepsilon),\cascadeSeqDet) = \detZ(\e{i,\nv},\cascadeSeqDet).\]
This implies that, for every $\tilde{\bm{x}}(\varepsilon)\in \tilde{\mathcal{X}}_i(\varepsilon)$ and function $b(\tilde{\bm{x}}(\varepsilon), \cascadeSeqDet):= |\detZ(\tilde{\bm{x}}(\varepsilon), \cascadeSeqDet) - \detZ(\e{i,\nv}, \cascadeSeqDet)|$, we have the following bounds on $\detZ(\tilde{\bm{x}}(\varepsilon), \cascadeSeqDet)$: 
\[\detZ(\e{{i,{\nv}}}, \cascadeSeqDet) - b( \tilde{\bm{x}}(\varepsilon),\cascadeSeqDet)\leq \detZ(\tilde{\bm{x}}(\varepsilon),\cascadeSeqDet)\leq \detZ(\e{{i,{\nv}}},\cascadeSeqDet) + b( \tilde{\bm{x}}(\varepsilon), \cascadeSeqDet).\]

Next, we show that $b(\cdot, \cdot)$ is bounded. First, note that $\detZ(\cdot,\cdot)\geq 0$ by definition. Furthermore, we assumed that $\cascadeCost\leq \boundZ\sum_{j\in \V} \paretoI{j}^\delta$, which means that this inequality must hold for any realization of $\pareto$ and $\cascadeSeqRand$, and in particular, for $\pareto = \tilde{\bm{x}}(\varepsilon)$ and $\cascadeSeqRand = \cascadeSeqDet$. Since $\tilde{\bm{x}}(\varepsilon)\in \tilde{\mathcal{X}}_i(\varepsilon)$, $$\sum_{j\in \V} (\tilde{x}(\varepsilon))_j^\delta \leq  1+\varepsilon,$$ and thus $\detZ(\tilde{\bm{x}}(\varepsilon),\cascadeSeqDet)\leq \boundZ (1+\varepsilon)$. Thus, we conclude that $$b(\tilde{\bm{x}}(\varepsilon),\cascadeSeqDet) = |\detZ(\tilde{\bm{x}}(\varepsilon),\cascadeSeqDet) - \detZ(\e{i,\nv}, \cascadeSeqDet)| \leq (1+\varepsilon) \boundZ$$ for all $\varepsilon>0$.

Next, let $b^*(i, \varepsilon, \cascadeSeqDet) := \sup_{\tilde{\bm{x}}\in \tilde{\mathcal{X}}_i(\varepsilon)} \{b(\tilde{\bm{x}},\cascadeSeqDet)\}$, which exists because $b(\cdot,\cdot)$ is bounded for $\tilde{\bm{x}}\in\tilde{\mathcal{X}}_i(\varepsilon)$ and $\tilde{\mathcal{X}}_i(\varepsilon)$ is non-empty. Hence, 
\begin{align}\begin{split}\label{eq:ineq1}&\PR{\detZ(\e{{i,{\nv}}},\cascadeSeqDet) - b^*(i, \varepsilon, \cascadeSeqDet)>z^{-\delta} ~|~ \cascadeSeqRand=\cascadeSeqDet, \paretoI{i} = 1, \pareto = \bm{x}/(z^{\delta}y)}\\
&\qquad\leq \PR{\cascadeCost>z^{-\delta } ~|~ \cascadeSeqRand=\cascadeSeqDet, \paretoI{i} = 1, \pareto = \bm{x}/(z^{\delta}y)}\\&\qquad\leq \PR{\detZ(\e{{i,{\nv}}},\cascadeSeqDet) + b^*(i, \varepsilon, \cascadeSeqDet)>z^{-\delta } ~|~\cascadeSeqRand=\cascadeSeqDet, \paretoI{i} = 1, \pareto = \bm{x}/(z^{\delta}y)}.\end{split}\end{align}

As the last part of \hyperlink{stepD}{Step~D}, we evaluate the term \(\PR{\cascadeSeqRand=\cascadeSeqDet ~|~ \paretoI{i} = 1, \pareto = \bm{x}/(z^{\delta}y)}\), using a similar approach to that above. Due to the assumption that the probability of occurrence of cascade $\cascadeSeqRand=\cascadeSeqDet$ is SRC, we again have that  \[\lim_{\varepsilon \downarrow 0}\PR{\cascadeSeqRand=\cascadeSeqDet~|~\paretoI{i} = 1, \pareto = \bm{x}/(z^{\delta}y)} = \PR{\cascadeSeqRand=\cascadeSeqDet ~|~\pareto = \e{{i,{\nv}}}}.\] Therefore, there exists some function $\xi( i,\varepsilon, \cascadeSeqDet)\rightarrow 0$ as $\varepsilon \downarrow 0$ such that 
\begin{align}&\label{eq:ineq2}\PR{\cascadeSeqRand=\cascadeSeqDet ~|~\pareto = \e{{i,{\nv}}}} - \xi( i,\varepsilon, \cascadeSeqDet) \nonumber\\&\qquad \leq\PR{\cascadeSeqRand=\cascadeSeqDet~|~\paretoI{i} = 1, \pareto = \bm{x}/(z^{\delta}y)} \\&\qquad\leq \PR{\cascadeSeqRand=\cascadeSeqDet ~|~\pareto = \e{{i,{\nv}}}} + \xi( i,\varepsilon, \cascadeSeqDet).\nonumber\end{align}

\hyperlink{stepE}{Step~E}: Next, we apply the SRC-based bounds \eqref{eq:ineq1} and \eqref{eq:ineq2} to \eqref{eq:condZ}. We obtain the following upper bound for the large cascade cost 
\begin{multline*}
    \PR{\cascadeCost>y, \cascadeSeqRand=\cascadeSeqDet, \sum_{j\neq i} \paretoI{j}^\delta \leq \varepsilon \paretoI{i}^\delta }\\\leq \int_{\left(\frac{\varepsilon}{\boundZ}\right)^{1/\delta}}^\infty\PR{\detZ(\e{{i,{\nv}}}, \cascadeSeqDet) + b^*(i, \varepsilon,\cascadeSeqDet)>z^{-\delta }~|~\cascadeSeqRand=\cascadeSeqDet, \paretoI{i}= 1, \pareto = \bm{x}/(z^\delta y)}\\\cdot\left(\PR{\cascadeSeqRand=\cascadeSeqDet ~|~\pareto = \e{{i,{\nv}}}}+\xi( i,\varepsilon, \cascadeSeqDet)\right) \cdot \int_{\bm{x}\in \mathcal{X}_i(z)}\PR{\pareto \in \textnormal{d}\bm{x} ~|~ \paretoI{i}^\delta  = z^\delta y}\cdot\PR{\frac{\paretoI{i}}{y^{1/\delta}} \in \textnormal{d} z},
\end{multline*}
and the associated lower bound
\begin{multline*}
    \PR{\cascadeCost>y, \cascadeSeqRand=\cascadeSeqDet, \sum_{j\neq i} \paretoI{j}^\delta \leq \varepsilon \paretoI{i}^\delta }\\\geq \int_{\left(\frac{\varepsilon}{\boundZ}\right)^{1/\delta}}^\infty\PR{\detZ(\e{{i,{\nv}}}, \cascadeSeqDet) - b^*(i,\varepsilon, \cascadeSeqDet)>z^{-\delta }~|~\cascadeSeqRand=\cascadeSeqDet, \paretoI{i}= 1, \pareto = \bm{x}/(z^\delta y)} \\\cdot\left(\PR{\cascadeSeqRand=\cascadeSeqDet ~|~\pareto = \e{{i,{\nv}}}}-\xi( i,\varepsilon, \cascadeSeqDet)\right) \cdot \int_{\bm{x}\in \mathcal{X}_i(z)}\PR{\pareto \in \textnormal{d}\bm{x}~|~ \paretoI{i}^\delta  = z^\delta y}\cdot \PR{\frac{\paretoI{i}}{y^{1/\delta}} \in \textnormal{d} z}.
\end{multline*}
We observe that $\detZ(\e{{i,{\nv}}}, \cascadeSeqDet) \pm b^*( i,\varepsilon, \cascadeSeqDet)$ is a constant in $z$, which implies that \[\PR{\detZ(\e{{i,{\nv}}}, \cascadeSeqDet) \pm b^*( i,\varepsilon, \cascadeSeqDet)>z^{-\delta }~|~\cascadeSeqRand=\cascadeSeqDet, \paretoI{i} = 1, \pareto = \bm{x}/(z^\delta y)} = \vmathbb{1}\mkern-4mu\left\{\detZ(\e{{i,{\nv}}}, \cascadeSeqDet) \pm b^*( i,\varepsilon, \cascadeSeqDet)>z^{-\delta }\right\}.\]
We can therefore replace the original lower integration limit $(\varepsilon/\boundZ)^{1/\delta}$ by
\[
u^*(i,\varepsilon, \cascadeSeqDet) = \max\left\{\big(\detZ(\e{i,\nv},\cascadeSeqDet) + b^*\big)^{-1/\delta}, (\varepsilon/\boundZ)^{1/\delta}\right\}\]
and \[l^*(i,\varepsilon,\cascadeSeqDet) = \max\left\{\big(\detZ(\e{i,\nv},\cascadeSeqDet) - b^*\big)^{-1/\delta}, (\varepsilon/\boundZ)^{1/\delta}\right\},
\]
for the upper and lower bounds, respectively, and drop the indicators.

For the upper bound, we additionally observe that\[\int_{\bm{x}\in \mathcal{X}_i(z)}\PR{\pareto \in \textnormal{d}\bm{x}~|~ \paretoI{i}^\delta  = z^\delta y} = \PR{\sum_{j\neq i} \paretoI{j}^\delta <\varepsilon z^\delta y~|~\paretoI{i}^\delta  = z^\delta y}\leq 1.\] Hence, altogether we obtain that
\begin{align*}&\PR{\cascadeCost>y, \cascadeSeqRand=\cascadeSeqDet, \sum_{j\neq i} \paretoI{j}^\delta \leq \varepsilon \paretoI{i}^\delta }\\&\qquad\leq  \PR{\cascadeSeqRand=\cascadeSeqDet~|~\pareto = \e{{i,{\nv}}})+\xi( i,\varepsilon, \cascadeSeqDet)}\cdot \int_{u^*( i,\varepsilon, \cascadeSeqDet)}^{\infty} \PR{\frac{\paretoI{i}}{y^{1/\delta}} \in \textnormal{d} z}\\
&\qquad= \left(\PR{\cascadeSeqRand=\cascadeSeqDet~|~\pareto = \e{{i,{\nv}}}}+\xi( i,\varepsilon, \cascadeSeqDet)\right)\cdot \PR{\paretoI{i}^\delta > u^*( i,\varepsilon, \cascadeSeqDet)^\delta y}\\
&\qquad= \left(\PR{\cascadeSeqRand=\cascadeSeqDet~|~\pareto = \e{{i,{\nv}}}}+\xi( i,\varepsilon, \cascadeSeqDet)\right)\cdot y^{-\alpha/\delta} \cdot K\cdot u^*( i,\varepsilon, \cascadeSeqDet)^{-\alpha}.
\end{align*}
For the lower bound, we use the fact that $\paretoI{i}> l^*(i,\varepsilon,\cascadeSeqDet)$ and obtain the following expression:
\begin{align*}
    &\PR{\cascadeCost>y, \cascadeSeqRand=\cascadeSeqDet, \sum_{j\neq i} \paretoI{j}^\delta \leq \varepsilon \paretoI{i}^\delta }\\&\qquad\geq   \left(\PR{\cascadeSeqRand=\cascadeSeqDet~|~\pareto = \e{{i,{\nv}}}}-\xi( i,\varepsilon, \cascadeSeqDet)\right)\int_{l^*( i,\varepsilon, \cascadeSeqDet)}^{\infty} \PR{\sum_{j\neq i} \paretoI{j}^\delta <\varepsilon z^\delta y~|~\paretoI{i}^\delta  = z^\delta y}\PR{\frac{\paretoI{i}}{y^{1/\delta}} \in \textnormal{d} z}\\
    &\qquad=  \left(\PR{\cascadeSeqRand=\cascadeSeqDet~|~\pareto = \e{i,{\nv}}}-\xi( i,\varepsilon, \cascadeSeqDet)\right)\PR{\sum_{j\neq i} \paretoI{j}^\delta <\varepsilon \paretoI{i}^\delta , \paretoI{i}^\delta > l^*( i,\varepsilon, \cascadeSeqDet)^\delta y}\\
    &\qquad\overset{\paretoI{i}, \paretoI{j} \text{indep.}}{\geq}
    \left(\PR{\cascadeSeqRand=\cascadeSeqDet~|~\pareto = \e{i,{\nv}}}-\xi( i,\varepsilon, \cascadeSeqDet)\right)\PR{\sum_{j\neq i} \paretoI{j}^\delta <\varepsilon l^*( i,\varepsilon, \cascadeSeqDet)^\delta y} \PR{\paretoI{i}^\delta > l^*( i,\varepsilon, \cascadeSeqDet)^\delta y}\\
    &\qquad=
    \left(\PR{\cascadeSeqRand=\cascadeSeqDet~|~\pareto = \e{i,{\nv}}}-\xi( i,\varepsilon, \cascadeSeqDet)\right)\PR{\sum_{j\neq i} \paretoI{j}^\delta <\varepsilon l^*( i,\varepsilon, \cascadeSeqDet)^\delta y}\cdot y^{-\alpha/\delta}\cdot K \cdot l^*(\cascadeSeqDet, i,\varepsilon)^{-\alpha}.
\end{align*}

\hyperlink{stepF}{Step~F}: So far, we have partitioned the probability $\PR{\cascadeCost>y}$ into a sum of probabilities of mutually exclusive events $\{\cascadeCost>y, \cascadeSeqRand = \cascadeSeqDet, \sum_{j\neq i} \paretoI{j}^\delta \leq \varepsilon \paretoI{i}^\delta\}$ and derived upper and lower bounds for the probabilities associated with these events. In this final step, we take the limits $y\rightarrow \infty$ and $\varepsilon\downarrow 0$ to show that the bounds are asymptotically equal, which proves the theorem. First, observe that as $y\rightarrow \infty,$ $\PR{\sum_{j\neq i} \paretoI{j}^\delta <\varepsilon l^*( i,\varepsilon, \cascadeSeqDet)^{\delta}y}\rightarrow 1$. Thus,  we obtain that
\[\lim_{y\rightarrow \infty}y^{\alpha/\delta} \PR{\cascadeCost>y, \cascadeSeqRand=\cascadeSeqDet, \sum_{j\neq i} \paretoI{j}^\delta \leq \varepsilon \paretoI{i}^\delta }\geq K\left(\PR{\cascadeSeqRand=\cascadeSeqDet~|~\pareto = \e{i,{\nv}}}-\xi( i,\varepsilon, \cascadeSeqDet)\right) l^*( i,\varepsilon, \cascadeSeqDet)^{-\alpha}\]
and 
\[\lim_{y\rightarrow \infty}y^{\alpha/\delta} \PR{\cascadeCost>y, \cascadeSeqRand=\cascadeSeqDet, \sum_{j\neq i} \paretoI{j}^\delta \leq \varepsilon \paretoI{i}^\delta }\leq K\left(\PR{\cascadeSeqRand=\cascadeSeqDet~|~\pareto = \e{i,{\nv}}}+\xi( i,\varepsilon, \cascadeSeqDet)\right) u^*( i,\varepsilon, \cascadeSeqDet)^{-\alpha}.\]
%Both upper and lower bounds hold for all $\varepsilon>0$ sufficiently small.
Therefore, taking the limit $\varepsilon\downarrow 0$, we obtain
\[\lim_{\varepsilon\downarrow 0}\lim_{y\rightarrow\infty} y^{\alpha/\delta} \PR{\cascadeCost>y,\cascadeSeqRand=\cascadeSeqDet, \sum_{j\neq i} \paretoI{j}^\delta \leq \varepsilon \paretoI{i}^\delta } = K\cdot \PR{\cascadeSeqRand=\cascadeSeqDet~|~\pareto = \e{i,{\nv}}}\cdot \detZ(\e{{i,{\nv}}}, \cascadeSeqDet)^{\alpha/\delta}.\]
Thus, altogether we find
\begin{align*}
    \lim_{y\rightarrow \infty}y^{\alpha/\delta}\PR{\cascadeCost>y}&= \lim_{y\rightarrow \infty}y^{\alpha/\delta}\sum_{i = 1}^{\nv}\sum_{\cascadeSeqDet\in \cascadeSet}\PR{\cascadeCost>y,\cascadeSeqRand=\cascadeSeqDet, \sum_{j\neq i} \paretoI{j}^\delta \leq \varepsilon \paretoI{i}^\delta } + \lim_{y\rightarrow \infty} \mathcal{O}(y^{-2\alpha/\delta})\\
    %&= \sum_{i = 1}^{\nv}\sum_{\cascadeSeqDet\in \cascadeSet}\lim_{\varepsilon\downarrow 0}\lim_{y\rightarrow \infty}y^{\alpha/\delta}\PR{\cascadeCost>y,\cascadeSeqRand=\cascadeSeqDet, \sum_{j\neq i} \paretoI{j}^\delta \leq \varepsilon \paretoI{i}^\delta }\\
    &= \sum_{i = 1}^{\nv}\sum_{\cascadeSeqDet\in \cascadeSet} K\cdot \PR{\cascadeSeqRand=\cascadeSeqDet ~|~ \pareto = \e{i,{\nv}}} \cdot \detZ(\e{i,{\nv}}, \cascadeSeqDet)^{\alpha/\delta}=l_\cascadeCost.
\end{align*}

We stress that that heavy-tailed behavior with tail parameter $\alpha/\delta$, i.e.,
$
\PR{\cascadeCost>y} \sim l_\cascadeCost\cdot y^{-\alpha/\delta},
$
occurs if there exist $\cascadeSeqDet\in\cascadeSet$ such that, for some $i\in\V$,
$
\PR{\cascadeSeqRand=\cascadeSeqDet \mid \pareto=\mathbf e_{i,\nv}}\cdot 
\detZ\big(\mathbf e_{i,\nv},\cascadeSeqDet\big)^{\alpha/\delta} > 0,
$
which ensures that $l_\cascadeCost>0$. Otherwise,
$
\PR{\cascadeCost>y} = \mathcal{O}\!\big(y^{-2\alpha/\delta}\big).
$
\end{proof}

\section{Scaling and SRC properties across phases} \label{proofs}
This section contains the auxiliary results required to establish Propositions~\ref{propClassFunc1} and~\ref{propClassFunc2}. These results characterize the behavior of the minimal-cost flow, source production, and flow capacity functions with respect to scale transformations and perturbations of the vertex weight input. Specifically, for each of these functions, we show scale-invariance and the Sequential Right Continuity (SRC), which is done in steps for different phases of the process. First, Section~\ref{subsec:flowproperties} focuses on the properties of the minimal-cost flow functions. Then, Sections~\ref{subsec:resultsPlanning}, \ref{subsec:Operationalproperties}, and \ref{subsec:emergencyProperties} show the scale invariance and SRC properties at the planning, operational, and emergency stage, respectively. Next, Section~\ref{subsec:probabilityCascadeProperties} shows that the probability of a particular cascade sequence occurring is also SRC w.r.t.\ the vertex weight $\pareto$ and it does not depend on the scale of $\pareto$. Finally, the Propositions~\ref{propClassFunc1} and \ref{propClassFunc2} are proven in Section~\ref{subsec:proofsPropositions}, using the behavioral properties derived so far. 

We emphasize that our results hold under the assumption of no production capacity constraints, i.e., $\genMLim = \infty$. However, as noted in Remark~\ref{remarkNoProdCap}, by following the proof strategy outlined in this section, the results of Propositions~\ref{propClassFunc1} and~\ref{propClassFunc2} can be extended to accommodate suitable choices of the production capacity matrix~$\genMLim$.
%In this section, we prove Propositions~\ref{propClassFunc1} and \ref{propClassFunc2}. However, as the proofs rely on several auxiliary results, we state and derive them first. The auxiliary results show the scale invariance and the SRC properties for the flow, flow capacity, source production, and sink requirement functions at every step of the cascade. 

\subsection{Properties of the minimal-cost flow function} \label{subsec:flowproperties}
The flow matrix $\flowM^*$, defined as the solution to a minimal-cost flow problem, serves as a fundamental building block in both planning and operational stages. This section establishes that $\flowM^*$ is scale-invariant (Lemma~\ref{flowInv}) and continuous with respect to its input arguments (Lemma~\ref{contFlow}). The following lemma shows the scale invariance.
\begin{lemma}\label{flowInv}
    The minimal-cost flow matrix $\flowM^*$ with costs given by \eqref{flowCost} or $\eqref{wardrop}$ is a scale-invariant function of the netput matrix and the flow capacity vector $(\bm{U}, \flowLimVec)$. In particular, 
    \[\flowM^*(\omega \bm{U}, \omega \flowLimVec) = \omega \flowM^*(\bm{U}, \flowLimVec).\]
\end{lemma}
We prove this result by showing that $\omega\flowM^*(\bm{U}, \flowLimVec)$ is a feasible solution to the minimal-cost flow problem with input $(\omega \bm{U}, \omega \flowLimVec)$ and the cost of this flow is equal to the cost of $\flowM^*(\omega\bm{U}, \omega\flowLimVec)$. 
\begin{proof}
%The scale-invariance of the \textcolor{red}{type 2} minimal cost flow function is proven in \textcolor{red}{Refer to the traffic paper (Lemma 5.2 atm)}, hence we only provide the proof for the \textcolor{red}{type 1} case.
Recall that the minimal cost flow matrix $\flowM^*$ is given as the solution 
\[\min_{\flowM\in \mathcal{D}}c_f(\flowM, \bm{U})~\text{s.t.:}~ \incM \flowM = \bm{U},\]
where $\mathcal{D} = \R^{\nee\times \nk}$ for cost function \eqref{flowCost} or $\mathcal{D} = \R_+^{\nee\times \nk}$ for cost function \eqref{wardrop}. Consider the scale parameter $\omega>0$ and, for better exposition, let $\flowM^*(\omega):= \flowM^*(\omega\bm{U}, \omega\flowLimVec)$ denote the minimal-cost flow matrix corresponding to input $(\omega \bm{U}, \omega \flowLimVec)$. First, we show that $\omega \flowM^*(1) = \omega \flowM^*(\bm{U}, \flowLimVec)$ is a feasible solution to the minimal-cost flow problem with input $(\omega \bm{U}, \omega \flowLimVec)$. Clearly, $\omega \flowM^*(1)\in \mathcal{D}$ because $\flowM^*(1)\in \mathcal{D}$. Moreover, 
\begin{equation}\incM\cdot \left(\omega\flowM^*(1)\right) = \omega \incM\flowM^*(1) = \omega \bm{U}\label{arg1},\end{equation}
because $\flowM^*(1)$ is the optimal, and therefore feasible, solution for input $(\bm{U}, \flowLimVec)$. Hence, we conclude that $\omega \flowM^*(1)$ is a feasible solution to the minimal-cost flow problem with input $(\omega \bm{U}, \omega \flowLimVec)$. 

Similarly, we can show that $\frac{1}{\omega} \flowM^*(\omega)$ is a feasible solution to the minimal-cost flow problem with input $(\bm{U}, \flowLimVec)$. Specifically, $\frac{1}{\omega}\flowM^*(\omega)\in \mathcal{D}$ and  $\incM \cdot \left(\frac{1}{\omega}\flowM^*(\omega)\right) = \frac{1}{\omega} \incM \flowM^*(\omega) = \bm{U}$, because $\flowM^*(\omega)$ is the optimal solution to the problem with the input $(\omega \bm{U}, \omega \flowLimVec)$. We use this property in the next step.

It remains to show that $\omega \flowM^*(1)$ is not only feasible, but also optimal for input $(\omega \bm{U}, \omega\flowLimVec)$. Due to the optimality of $\flowM^*(\omega)$ and $\flowM^*(1)$ for $(\omega \bm{U}, \omega \flowLimVec)$ and $(\bm{U}, \flowLimVec)$, respectively, and the corresponding feasibility of $\omega \flowM^*(1)$  and $\frac{1}{\omega}\flowM^*(\omega)$, the following two inequalities hold:
\begin{align}\label{ineq:costsFunc1}
 c_f(\flowM^*(\omega),\omega\flowLimVec)\leq  c_f(\omega\flowM^*(1),\omega \flowLimVec), \end{align}
 \begin{align} \label{ineq:costsFunc2}  c_f(\flowM^*(1),\flowLimVec)\leq  c_f(\flowM^*(\omega)/\omega, \flowLimVec).
\end{align}
Now, we consider the two different cost functions separately. First, suppose that $c_f$ is given by \eqref{flowCost}. Note that this particular cost function is independent of its second argument $\flowLimVec$, but we write it for completeness. Multiplying \eqref{ineq:costsFunc2} by $\omega^\beta$, we obtain
\begin{align*}c_f(\omega \flowM^*(1), \omega \flowLimVec) &= \frac{\omega ^\beta}{\beta} \sum_{e\in \EE} b_e a_e (|\flow[e]^*(1)|/a_e)^\beta = \frac{\omega^\beta}{\beta}c_f(\flowM^*(1), \flowLimVec)\\&
\overset{\eqref{ineq:costsFunc2}}\leq  \frac{\omega^\beta}{\beta}c_f(\flowM^*(\omega)/\omega, \flowLimVec) = \frac{\omega ^\beta}{\beta}  \sum_{e\in \EE} b_e a_e (|\flow[e]^*(\omega)|/(\omega a_e))^\beta = c_f(\flowM^*(\omega), \omega \flowLimVec).\end{align*} 
Hence, applying \eqref{ineq:costsFunc1}, we conclude that \begin{align}
 c_f(\omega\flowM^*(1),\omega\flowLimVec) =  c_f(\flowM^*(\omega), \omega \flowLimVec). \label{eq: ineqFlowCost}
\end{align}
Similarly, if $c_f$ is given by \eqref{wardrop}, \eqref{ineq:costsFunc2} multiplied with $\omega$ yields
\begin{align*}
    c_f(\omega \flowM^*(1), \omega \flowLimVec) &=  \omega \sum_{e\in \EE} d_e\flow[e]^*(1) + \frac{b_e}{\beta} \flowLim{e}(\flow[e]^*(1)/\flowLim{e})^\beta = \omega c_f(\flowM^*(1), \flowLimVec)\\&\overset{\eqref{ineq:costsFunc2}}{\leq} \omega c_f(\flowM^*(\omega)/\omega, \flowLimVec) = \omega \sum_{e\in \EE} d_e\flow[e]^*(\omega)/\omega +  \frac{b_e}{\beta} \flowLim{e}(\flow[e]^*(\omega)/(\omega \flowLim{e}))^\beta = c_f(\flowM^*(\omega), \omega \flowLimVec).
\end{align*}
Hence, again, applying \eqref{ineq:costsFunc1} yields \eqref{eq: ineqFlowCost}, which implies that the costs of $\omega \flowM^*(1)$ and $\flowM^*(\omega)$ for the input $\omega\bm{U}$ are equal. Finally, by the uniqueness of the solution of the minimum cost flow problem with cost functions \eqref{flowCost} and \eqref{wardrop} discussed in Section~\ref{subsec:FlowFunctions}, we conclude that $\flowM^*(\omega) = \omega\flowM^*(1)$.
\end{proof}

Next, we establish the continuity of the flow function $\flowM^*$, which maps a given netput matrix $\bm{U}$ and flow capacity vector $\flowLimVec$ to the corresponding minimal-cost flow. Continuity plays a crucial role in our analysis, particularly when studying the behavior of the optimal source production $\genM^*$, which depends on $\flowM^*$ through the flow capacity constraint~\eqref{optFlowConst}. The following lemma shows that the minimal-cost flow matrix $\flowM^*$ is indeed a continuous function of $\bm{U}$ and $\flowLimVec$ for the choices of cost functions introduced in Section~\ref{subsec:FlowFunctions}. 
%\begin{lemma} \label{contFlowSignPattern}
%Let the minimal-cost flow matrix \( \flowM^* \) be defined as in \eqref{optFlow}, with cost function
%%\[
%\sum_{e \in \EE} \left( d_e \flow[e] + b_e \flowLim{e} \left( \flow[e]/\flowLim{e} \right)^\beta \right), \quad \text{where } d_e, b_e \geq 0, \ \beta > 1.
%\]
%Suppose \(  (\bm{U}^l, \flowLimVec^l)_{l \in \mathbb{N}} \) is a sequence of net requirement matrices and flow capacity vectors such that:
%\begin{itemize}
%    \item \( \lim_{l \to \infty} (\bm{U}^l, \flowLimVec^l) = (\bm{U}^\infty, \flowLimVec^\infty) \),
%    \item $|\bm{U}^l|\geq |\bm{U}^\infty|$ for all $l\in \N$,
%    \item There exists a sign pattern \( \bm{\sigma} \in \{-1, +1\}^{\nv \times |\K|} \) such that for all \( l \in \mathbb{N} \), and for all \( v \in \V \), \( k \in \K \), it holds that
%    \[
 %   \sigma_{v,k} \cdot \bm{U}^l_{v,k} \geq 0,
 %   \]
 %   i.e., each entry of \( \bm{U}^l \) remains either non-negative or non-positive throughout the sequence.
%\end{itemize}
%Then the flow mapping is continuous along the sequence, i.e.,
%\[
%\lim_{l \to \infty} \flowM^*(\bm{U}^l, \flowLimVec^l) = \flowM^*(\bm{U}^\infty, \flowLimVec^\infty).
%\]
%\end{lemma}
\begin{lemma} \label{contFlow}
Let $\flowM^*$ denote the minimal-cost flow matrix as defined in~\eqref{optFlow}, with the cost function given by either~\eqref{flowCost} or~\eqref{wardrop}. Then $\flowM^*$ is a continuous function of the netput matrix $\bm{U}$ and the flow capacity vector $\flowLimVec$. In particular, for any sequence of netput matrices and flow capacity vectors $(\bm{U}^l, \flowLimVec^l) \to (\bm{U}^\infty, \flowLimVec^\infty)$, it holds that
\[
\lim_{l \to \infty} \flowM^*(\bm{U}^l, \flowLimVec^l) = \flowM^*(\bm{U}^\infty, \flowLimVec^\infty).
\]
\end{lemma}
The proof establishes continuity by showing that all optimal flows converge to the desired limit. Compactness guarantees the existence of subsequential limits, and we construct a sequence of feasible flows that converges to the desired limit. A cost comparison between the feasible and optimal subsequences then shows that every subsequential limit---and hence the limit of the complete sequence---is the desired limit.
\begin{proof}
Let $(\bm{U}^l, \flowLimVec^l)_{l \in \mathbb{N}}$ be a sequence of netput matrices and flow capacity vectors converging to $(\bm{U}^\infty, \flowLimVec^\infty)$. For each $l\in \N$, define $\flowM^l := \flowM^*(\bm{U}^l, \flowLimVec^l)$ and $\flowM^\infty := \flowM^*(\bm{U}^\infty, \flowLimVec^\infty)$. These flows are well-defined under our assumptions, since i) the graph is strongly connected, ii) the netput matrices by construction satisfy the total requirement–production balance for each commodity, and iii) the cost function is strictly convex. The first two conditions ensure the existence of a feasible flow, such as a shortest-path flow, while the last guarantees the uniqueness of the optimal solution.

We begin by showing that the sequence $(\flowM^l)$ lies in a compact subset of $\mathbb{R}^{\nee \times \nk}$. Define the total demand for each $\bm{U}^l$ as
\[
\xi^l := \sum_{v \in \V} \sum_{k \in \K} \bm{U}^l_{v,k} \cdot \vmathbb{1}\{ \bm{U}^l_{v,k} \geq 0 \}.
\]
Since $(\bm{U}^l)$ converges, it is bounded and, as a result, $(\xi^l)$ is also bounded. Let $\xi := \sup_l \xi^l < \infty$. Because the cost function is non-negative and increasing in flow, optimal flows avoid unnecessary cycles and thus never exceed the total demand. Hence, for every $e \in \EE$ and $k \in \K$,
\[
|\flow[e,k]^l| \leq \xi^l \leq \xi.
\]
It follows that $\flowM^l \in [-\xi, \xi]^{\nee \times \nk}$ for each $l$. This set is compact in $\mathbb{R}^{\nee \times \nk}$, so $(\flowM^l)_{l\in \N}$ admits a convergent subsequence. 

Our next goal is to show that the limit of an arbitrary convergent subsequence $(\flowM^{l_m})_{m\in\N}$ is equal to $\flowM^{\infty}$, i.e., $\lim_{m\rightarrow \infty } \flowM^{l_m} := \widehat{\flowM}^{\infty} = \flowM^{\infty}$.  To do so, we first construct a \textit{sequence} of feasible flows converging to $\flowM^\infty$. Let $\flowM^{(cont)}(\bm{U})$ denote the solution to the flow problem with quadratic cost function $\sum_{e\in \EE} f_e^2$. This problem is a strictly convex quadratic program with linear constraints, and its solution is a continuous function of $\bm{U}$; see~\citep{TONDEL2003489}. Define
\[
\widetilde{\flowM}^l := \flowM^\infty + \flowM^{(cont)}(\bm{U}^l - \bm{U}^\infty).
\]
 Note that, since $\flowM^\infty, \flowM^{(cont)}(\bm{U}^l - \bm{U}^\infty )\in \mathcal{D}$, we have $\widetilde{\flowM}^l\in \mathcal{D}$. Moreover, by definition of $\flowM^\infty$, we have $\incM\flowM^\infty = \bm{U}^\infty$, and thus
\[\incM\widetilde{\flowM}^l =\incM\flowM^\infty + \incM\flowM^{(cont)}(\bm{U}^l - \bm{U}^\infty)= \bm{U}^{\infty} + \bm{U}^l - \bm{U}^{\infty} = \bm{U}^l.\]
This means that $\widetilde{\flowM}^l$ is a feasible flow matrix for netput $\bm{U}^l$ and, as a result, it is a feasible solution for all $l\in \N$. In addition, by continuity of $\flowM^{(cont)}$ it follows that $\widetilde{\flowM}^l \to \flowM^\infty$ as $l \to \infty$.

With the feasible sequence defined, we now compare its flow costs to those of the optimal subsequence and use continuity of the cost function to establish equality of the limits. Since $\flowM^{l_m}$ is optimal under cost $c_f$ and $\widetilde{\flowM}^{l_m}$ is feasible, we have
\[
c_f(\flowM^{l_m}, \flowLimVec^{l_m}) \leq c_f(\widetilde{\flowM}^{l_m}, \flowLimVec^{l_m}) \quad \text{for all } m\in \N.
\]
Taking the limit as $m \to \infty$ and using continuity of $c_f$, we obtain
\[
c_f(\widehat{\flowM}^\infty, \flowLimVec^\infty) \leq c_f(\flowM^\infty, \flowLimVec^\infty).
\]
Note that, for every $m\in\N$, $\flowM^{l_m}$ satisfies $\incM\flowM^{l_m} = \bm{U}^{l_m}$ and $ \flowM^{l_m}\in \mathcal{D}$. Therefore, taking limits on both sides and using the continuity of linear transformations, we obtain $\incM\widehat{\flowM}^\infty = \bm{U}^{\infty}$ with $\widehat{\flowM}^{\infty}\in \mathcal{D}$. In other words, $\widehat{\flowM}^\infty$ is a feasible flow for the netput matrix $\bm{U}^\infty$.
Since $\flowM^\infty$ is the unique minimizer for $(\bm{U}^\infty, \flowLimVec^\infty)$ and $\widehat{\flowM}^\infty$ is feasible, it follows that $\widehat{\flowM}^\infty = \flowM^\infty$.

As every convergent subsequence of $(\flowM^l)_{l\in\N}$ converges to $\flowM^\infty$, and the sequence lies in a compact set, the full sequence also converges to $\flowM^\infty$, i.e., 
\[
\lim_{l \to \infty} \flowM^l = \flowM^\infty. \qedhere
\]
\end{proof}
\subsection{Scale invariance and SRC in the planning phase} \label{subsec:resultsPlanning}
Having established the behavioral properties of the minimal-cost flow function, we now turn to analyzing the behavior of the source production matrix \eqref{eq:defPlanningProd} in the planning stage. We first show the scale-invariance property of property of source production, under the assumption that $\genM^*$ is given as the minimal cost solution to the resource allocation problem, as stated in \eqref{optProd} in Section~\ref{subsec:SourceFunctions}.
\begin{lemma} \label{lemmaProductionPlanningInv}
    Suppose that $\genM^*$ is given by the minimal-cost source production \eqref{optProd} and suppose that there are no production capacity constraints, i.e., $\genMLim = \infty$. Then, $\genM^{(0)}$ is a scale-invariant function of the vertex weight vector $\pareto$.
\end{lemma}
We prove this result by transforming the minimization problem \eqref{optProd} with input $\omega \pareto$ into a problem that has the same solution as \eqref{optProd} with input $\pareto$. This is done using a change-of-variables approach. We emphasize that all functions in the planning and operational phases are random variables, fully determined by the vertex weight $
\pareto$. With this in mind, in the proofs that follow, for any random variable $A$, we use the shorthand notation 
$A(\pareto)$ to denote random variable $A$ conditioned on 
$\pareto$.

\begin{proof}
Fix $\omega > 0$ and consider the vertex weight vectors $\pareto$ and $\omega \pareto$. From Equation~\eqref{eq:defT0} it follows that the corresponding sink requirement matrices in the planning stage are
\[
\demM^{(0)}(\pareto) = \mathrm{diag}(\pareto)\cdot \bm{Q}, \quad \text{and} \quad \demM^{(0)}(\omega \pareto) =  \text{diag}(\omega\pareto) \cdot \bm{Q} = \omega \demM^{(0)}(\pareto).
\]
According to Equation~\eqref{eq:defPlanningProd}, the planning source production matrix is obtained by solving the minimal-cost source production problem with input $\demM^{(0)}(\omega \pareto)$ given no flow or capacity constraints. For brevity, we denote this by $$\genM^{(0)}(\omega) := \genM^*(\demM^{(0)}(\omega\pareto),\infty, \infty, 1 ).$$ Then, using~\eqref{optProd}, we can write
\begin{equation} \label{eq:prod1}
\genM^{(0)}(1) = \argmin_{\genM \in \mathbb{R}_+^{\nv \times \nk}} \sum_{v \in \V} \sum_{k \in \K} \frac{1}{\gamma} c_{v,k} \gen[v,k]^\gamma \quad \text{s.t.:} \quad \e{_{\nv}}^T \genM = \e{_{\nv}}^T \demM^{(0)}(\pareto),
\end{equation}
and
\begin{equation} \label{eq:prod2}
\genM^{(0)}(\omega) = \argmin_{\widetilde{\genM} \in \mathbb{R}_+^{\nv \times \nk}} \sum_{v \in \V} \sum_{k \in \K} \frac{1}{\gamma} c_{v,k} \tilde{s}_{v,k}^\gamma \quad \text{s.t.:} \quad \e{_{\nv}}^T \widetilde{\genM} = \omega \cdot \e{_{\nv}}^T \demM^{(0)}(\pareto).
\end{equation}
Using the change of variables $\tilde{\genM} := \omega\genM$, it follows that \eqref{eq:prod1} and \eqref{eq:prod2} are identical minimization problems up to a scalar, and therefore
%To relate \eqref{eq:prod2} to \eqref{eq:prod1}, we apply a change of variables. Let $\tilde{\genM} := \genM / \omega$. Substituting into \eqref{eq:prod2} yields the equivalent problem
%\begin{equation} \label{eq:prod3}
%\argmin_{\omega\tilde{\genM} \in \mathbb{R}_+^{\nv \times \nk}} \sum_{v \in \V} \sum_{k \in \K} \frac{1}{\gamma} c_{v,k} \cdot \omega^\gamma \cdot \tilde{s}_{v,k}^\gamma \quad \text{s.t.:} \quad \e{_{\nv}}^T \tilde{\genM} = \e{_{\nv}}^T \demM^{(0)}(\pareto).
%\end{equation}
%Since the objective in \eqref{eq:prod3} differs from that in \eqref{eq:prod1} only by a constant multiplicative factor $\omega^\gamma$, the two problems have the same minimizer up to scaling. Therefore, we conclude that
\[
\genM^{(0)}(\omega) = \omega \cdot \genM^{(0)}(1),
\]
which establishes the scale invariance of the planning-stage source production function.
\end{proof}
The following lemma shows that the planning source production matrix is SRC with respect to the vertex weight vector~$\pareto$ on the positive orthant. In particular, along any componentwise nondecreasing sequence of inputs converging to a non-zero point in the positive orthant, the image under $\genM^{(0)}$ converges to the value at that point. 

\begin{lemma}
\label{lemmaProductionPlanning}
Suppose that $\genM^*$ is given by the minimal-cost source production solution to \eqref{optProd}, and assume there are no production capacity constraints, i.e., $\genMLim = \infty$. Then the planning-stage source production matrix $\genM^{(0)}$ is an SRC function of the vertex weight vector $\pareto$ on $\mathbb{R}^{\nv}_+ \setminus \{ \bm{0} \}$. Specifically, for any sequence $(\pareto^l)_{l \in \mathbb{N}}$ with $\pareto^l \to \pareto^\infty \neq \bm{0}$ and $\pareto^l \geq \pareto^\infty$ componentwise for all~$l$, it holds that
\[
\lim_{l \to \infty} \genM^{(0)}(\pareto^l) = \genM^{(0)}(\pareto^\infty).
\]
\end{lemma}
The proof follows a strategy similar to that used in the proof of Lemma~\ref{contFlow}, but relies on a different construction for the alternative feasible sequence of production matrices.
\begin{proof}
Let $(\pareto^l)_{l \in \mathbb{N}}$ be a sequence satisfying the conditions of the lemma. From the definition in~\eqref{eq:defT0}, the planning sink requirement matrix corresponding to a vertex weight vector $\pareto$ is given by
\[
\demM^{(0)}(\pareto) = \operatorname{diag}(\pareto) \cdot \bm{Q},
\]
and the associated source production matrix is defined as
\[
\genM^{(0)}(\pareto) := \genM^*(\demM^{(0)}(\pareto), \infty, \infty, 1),
\]
which exists and is unique for every $\pareto \geq \bm{0}$ due to strict convexity of the objective function in~\eqref{optProd}.

We first establish that the sequence $\big(\genM^{(0)}(\pareto^l)\big)_{l \in \mathbb{N}}$ lies in a compact subset of $\mathbb{R}^{\nv \times \nk}$. For each $l\in \N$, the balance constraint~\eqref{optProdBalance} implies that for every $v \in \V$ and $k \in \K$, the production of commodity $k$ at vertex $v\in \V$ is bounded by the total requirement for this commodity, i.e., 
\[
\gen[v,k]^{(0)}(\pareto^l) \leq \left( \e{_{\nv}}^T \demM^{(0)}(\pareto^l) \right)_k = \left( (\pareto^l)^T \bm{Q} \right)_k \leq \sup_{l \in \mathbb{N}} \left( (\pareto^l)^T \bm{Q} \right)_k =: m_k.
\]
Since $(\pareto^l)$ converges, the supremum $m_k$ is finite for each $k$, so the sequence lies in $[0, \max_k m_k]^{\nv \times \nk}$. As this is a closed and bounded subset of a finite-dimensional Euclidean space, it is compact. Hence, $\big(\genM^{(0)}(\pareto^l)\big)_{l \in \mathbb{N}}$ has at least one convergent subsequence.

Next, we construct another sequence of source production matrices $\widetilde{\genM}^l$ that converges to $\genM^{(0)}(\pareto^{\infty})$ and satisfies Constraint~\eqref{optProdBalance} for every $l\in \N$. For each $l$ and $k\in \K$, let the scaling factor $\sigma_k^l$ be the ratio of the total sink requirements of commodity $k$ for vertex weight vectors $\pareto^l$ and $\pareto^{\infty}$, i.e., 
\[
\sigma^l_k := \frac{ \left( \e{_{\nv}}^T \demM^{(0)}(\pareto^l) \right)_k }{ \left( \e{_{\nv}}^T \demM^{(0)}(\pareto^\infty) \right)_k }, \quad \bm{\sigma}^l := (\sigma^l_k)_{k \in \K},
\]
and set
\[
\widetilde{\genM}^l := \genM^{(0)}(\pareto^\infty) \cdot \operatorname{diag}(\bm{\sigma}^l).
\]
Then for all $l$, we have
\[
\e{_{\nv}}^T \widetilde{\genM}^l = \e{_{\nv}}^T \genM^{(0)}(\pareto^\infty) \cdot \operatorname{diag}(\bm{\sigma}^l) \overset{\eqref{optProdBalance}}{=} \e{_{\nv}}^T \demM^{(0)}(\pareto^\infty) \cdot \operatorname{diag}(\bm{\sigma}^l)=  \e{_{\nv}}^T \demM^{(0)}(\pareto^l),
\]
so each $\widetilde{\genM}^l$ satisfies the balance constraint~\eqref{optProdBalance}. Moreover, since the sink requirement and source production matrices, $\demM^{(0)}(\cdot)$ and $\genM^{(0)}(\cdot)$, are non-negative by construction, it follows that $\genM^{(0)}(\pareto^\infty) \geq 0$ and $\bm{\sigma}^l \geq 0$. Consequently, $\widetilde{\genM}^l \geq 0$, so each $\widetilde{\genM}^l$ is feasible for problem~\eqref{optProd} with input $\demM^{(0)}(\pareto^l)$. Furthermore, by linearity of $\demM^{(0)}(\cdot)$, the components of $\bm{\sigma}^l$ are continuous in $\pareto$, and together with the convergence $\pareto^l \to \pareto^\infty$, this implies that
\begin{equation}
\lim_{l \to \infty} \widetilde{\genM}^l = \genM^{(0)}(\pareto^\infty).
\label{eq:limitSourceProd}
\end{equation}

Next, we show that each convergent subsequence of $\left(\genM^{(0)}(\pareto^l)\right)_{l\in \N}$ has a limit equal to $\genM^{(0)}(\pareto^{\infty})$, by comparing the production costs of the feasible and optimal subsequences. Let $(\genM^{l_m})_{m \in \mathbb{N}}$ be any convergent subsequence of $\genM^{(0)}(\pareto^l)$ with limit $\widehat{\genM} := \lim_{m \rightarrow \infty} \genM^{l_m}$. Since $\genM^{l_m}$ is optimal for $\demM^{(0)}(\pareto^{l_m})$ and $\widetilde{\genM}^{l_m}$ is feasible, we have that for all $m$,
\[
c_s(\genM^{l_m}) \leq c_s(\widetilde{\genM}^{l_m}),
\]
where $c_s$ denotes the production cost function, as in \eqref{optProd}. By the continuity of $c_s$ in $\genM$, and using $\eqref{eq:limitSourceProd}$, it follows that
\[
c_s(\widehat{\genM}) \leq c_s(\genM^{(0)}(\pareto^\infty)).
\]
Moreover, since $\widehat{\genM}$ is feasible for the limit input and the cost function is strictly convex, the minimizer is unique, and thus
\[
\widehat{\genM} = \genM^{(0)}(\pareto^\infty).
\]

Finally, every convergent subsequence of $\genM^{(0)}(\pareto^l)$ converges to the same limit, and the sequence lies in a compact set. Therefore, the entire sequence converges:
\[
\lim_{l \to \infty} \genM^{(0)}(\pareto^l) = \genM^{(0)}(\pareto^\infty). \qedhere
\]
\end{proof}
\subsection{Scale invariance and SRC in the operational phase} \label{subsec:Operationalproperties}
We have thus far established the scale-invariance and SRC properties of the relevant functions in the planning phase. The results presented in this section extend these properties to the operational phase. First, the following lemma shows that the operational edge capacity vector \( \flowLimVec^{(1)} \) inherits scale-invariance and SRC from the structure of the minimum-cost flow and planning source production functions.

\begin{lemma} \label{planningFlowInv}
Suppose that \( \flowM^* \) is given by the minimal-cost flow matrix associated with either cost function~\eqref{flowCost} or~\eqref{wardrop}, and that \( \genM^* \) is given by the minimal-cost source production \eqref{optProd}. Furthermore, suppose there are no production capacity constraints, i.e., \( \genMLim = \infty \). Then,
\begin{enumerate}\item The operational edge capacity vector \( \flowLimVec^{(1)} \) is a scale-invariant function of the vertex weight vector \( \pareto \), i.e., 
\[
\flowLimVec^{(1)}(\omega \pareto) = \omega \flowLimVec^{(1)}(\pareto).
\]
\item The operational edge capacity vector \( \flowLimVec^{(1)} \) is SRC with respect to the vertex weight $\pareto$ on \( \mathbb{R}^{\nv}_+ \setminus \{\mathbf{0}\} \). In other words, for every sequence $\pareto^l\rightarrow \pareto^\infty \neq \bm{0}$, with $\pareto^l\geq \pareto^\infty$ for every $l\in \N$, 
\[\lim_{l\rightarrow \infty } \flowLimVec^{(1)}(\pareto^l) = \flowLimVec^{(1)}{\pareto^\infty}.\]
\end{enumerate}
\end{lemma}

\begin{proof}
Recall that in the planning stage we assume that there are no flow capacity constraints, i.e., $\flowLimVec^{(1)} = \infty$. The operational edge flow capacity is defined in Equation~\eqref{flowCapacity} as
\[
\flowLim{e}^{(1)} = \max\{ \tau \flow[e]^{(0)}, \flowLim{\min} \}, \quad \text{where } \flowLim{\min} = \varepsilon_{\min} \sum_{v \in \V} \paretoI{v}.
\]
We begin by establishing that the planning flow matrix \( \flowM^{(0)} = \flowM^*(\bm{U}^{(0)}, \infty) \) is scale-invariant. By Lemma~\ref{lemmaProductionPlanningInv}, we have \( \genM^{(0)}(\omega \pareto) = \omega \genM^{(0)}(\pareto) \), and consequently,
\[
\bm{U}^{(0)}(\omega \pareto) = \demM^{(0)}(\omega \pareto) - \genM^{(0)}(\omega \pareto) = \omega \demM^{(0)}(\pareto) - \omega \genM^{(0)}(\pareto) = \omega \bm{U}^{(0)}(\pareto).
\]
Moreover, by Lemma~\ref{flowInv}, we obtain
    \[
    \flowM^{(0)}(\omega \pareto) = \flowM^*(\bm{U}^{(0)}(\omega \pareto), \infty) = \omega \flowM^*(\bm{U}^{(0)}(\pareto), \infty) = \omega \flowM^{(0)}(\pareto).
    \]
%Next, we consider the two cases for the flow cost function:
%\begin{itemize}
%    \item If the cost function is given by~\eqref{flowCost}, then by definition it does not depend on the flow capacity vector at any phase. As a result, Lemma~\ref{flowInv} implies that
%    \[
%    \flowM^{(0)}(\omega \pareto) = \flowM^*(\bm{U}^{(0)}(\omega \pareto), \infty) = \omega \flowM^*(\bm{U}^{(0)}(\pareto), \infty) = \omega \flowM^{(0)}(\pareto).
%    \]  
%    \item If the cost function is given by~\eqref{wardrop}, then, since \( \beta > 1 \) and \( \flowLimVec^{(0)} = \infty \), this cost function reduces to  
%    \[
%    \lim_{\bm{x} \to \infty} \sum_{e \in \EE} \left( d_e \flow[e] + \frac{1}{\beta} b_e x_e \left( \frac{\flow[e]}{x_e} \right)^\beta \right) = \sum_{e \in \EE} d_e \flow[e].
%    \]
%     This limiting cost is a special case of~\eqref{flowCost}, and by Lemma~\ref{flowInv}, we again conclude that \( \flowM^{(0)}(\omega \pareto) = \omega \flowM^{(0)}(\pareto) \).
%\end{itemize}
Then, the total flow vector \( \flow^{(0)}(\cdot) = \flowM^{(0)}(\cdot) \e{_{|\K|}} \) satisfies
\[
\flow^{(0)}(\omega \pareto) = \flowM^{(0)}(\omega\pareto)\e{_{\nk}}=\omega \flowM^{(0)}(\pareto)\e{_{\nk}} =\omega \flow^{(0)}(\pareto).
\]
Consequently, for each edge \( e \in \EE \),
\[
\flowLim{e}^{(1)}(\omega \pareto) = \max \left\{ \tau \flow[e]^{(0)}(\omega \pareto), \omega \varepsilon_{\min} \sum_{v \in \V} \paretoI{v} \right\} = \omega \flowLim{e}^{(1)}(\pareto),
\]
proving the scale invariance.

To establish the SRC property, consider some sequence \( (\pareto^l)_{l \in \N} \in \mathbb{R}^{\nv}_+ \) as introduced in the statement of this lemma, i.e., satisfying \( \pareto^l \to \pareto^\infty \neq \mathbf{0} \) and \( \pareto^l \geq \pareto^\infty \) componentwise for every $l\in \N$. Since Lemma~\ref{lemmaProductionPlanning} applies to sequences of this form, we have that both $\genM^{(0)}$ and $\demM^{(0)}$ are SRC in $\pareto$, and therefore
\begin{equation} \label{eq:Ucont}
\lim_{l \to \infty} \bm{U}^{(0)}(\pareto^l) = \lim_{l \to \infty} \left(\demM^{(0)}(\pareto^l) - \genM^{(0)}(\pareto^l)\right) = \demM^{(0)}(\pareto^\infty) - \genM^{(0)}(\pareto^\infty)= \bm{U}^{(0)}(\pareto^\infty).
\end{equation}
Furthermore, Lemma~\ref{contFlow} implies that \( \flowM^* \) is continuous in both \( \bm{U} \) and \( \flowLimVec \). Hence, for each edge \( e \in \EE \),
\begin{align*}
\lim_{l \to \infty} \flowLim{e}^{(1)}(\pareto^l)
&= \lim_{l \to \infty} \max \left\{ \tau \flow[e]^{(0)}(\pareto^l), \varepsilon_{\min} \sum_{v \in \V} \paretoI{v}^l \right\} \tag*{\hfill \textit{(Def.\ \eqref{flowCapacity})}} \\
&= \max \left\{ \tau \lim_{l \to \infty} \flow[e]^{(0)}(\pareto^l), \varepsilon_{\min} \lim_{l \to \infty} \sum_{v \in \V} \paretoI{v}^l \right\} \tag*{\hfill \textit{(Cont.\ of max operator)}} \\
&= \max \left\{ \tau \left[ \flowM^*(\bm{U}^{(0)}(\pareto^\infty), \infty) \e{_{|\K|}} \right]_e, \varepsilon_{\min} \sum_{v \in \V} \paretoI{v}^\infty \right\} \tag*{\hfill \textit{(Def.\ \eqref{flow}, Lemma~\ref{contFlow}, \eqref{eq:Ucont})}} \\
&= \flowLim{e}^{(1)}(\pareto^\infty). \tag*{\hfill \textit{(Def.\ \eqref{flowCapacity})}}
\end{align*}

Since this holds for every edge \( e \in \EE \), we conclude that \( \flowLimVec^{(1)}(\pareto^l) \to \flowLimVec^{(1)}(\pareto^\infty) \) as $l\rightarrow \infty$, i.e., \( \flowLimVec^{(1)} \) is SRC with respect to \( \pareto \).
\end{proof}

We next show that the operational source production matrix $\genM^{(1)}$ inherits scale invariance from its inputs. In particular, we prove that if the source production matrix and the flow vector are computed as solutions to their respective minimal-cost problems, then the operational source production matrix scales linearly with the vertex weight vector $\pareto$.

\begin{lemma}
    \label{lemmaProductionOpInv}
    Suppose that \( \genM^* \) and \( \flowM^* \) are given by the minimal-cost source production and flow functions, respectively. Moreover, suppose that there are no production capacity constraints, i.e., \( \genMLim = \infty \). Then, the operational source production matrix \( \genM^{(1)} \) is a scale-invariant function of the vertex weight vector \( \pareto \).
\end{lemma}
The proof of this lemma uses similar arguments as the proof of Lemma \ref{lemmaProductionPlanningInv}. However, in addition, it is necessary to apply Lemma \ref{flowInv} and Lemma~\ref{planningFlowInv} to show that Constraint~\eqref{optFlowConst} scales with $\omega$ in the same manner as Constraint~\eqref{optProdBalance}. 
\begin{proof}
Fix \( \omega > 0 \), and consider two vertex weight vectors \( \pareto \) and \( \omega \pareto \). As in the planning stage, the corresponding sink requirement matrices are given by
\[
\demM^{(1)}(\pareto) = \mathrm{diag}(\pareto) \cdot \bm{Q}, \quad \demM^{(1)}(\omega \pareto) = \mathrm{diag}(\omega \pareto) \cdot \bm{Q} = \omega \cdot \demM^{(1)}(\pareto).
\]
We wish to show that
\[
\genM^{(1)}(\omega \pareto) = \omega \cdot \genM^{(1)}(\pareto),
\]
where \( \genM^{(1)}(\pareto) := \genM^*(\demM^{(1)}(\pareto), \infty, \flowLimVec^{(1)}(\pareto), \lambda^{(1)}) \) is defined as the unique minimizer of the convex program \eqref{optProd}, with constraints \eqref{optProdBalance} and \eqref{optFlowConst} active and \( \genMLim = \infty \).

Let us first examine how the constraints scale with \( \omega \). By Lemma~\ref{planningFlowInv}, the operational flow capacity vector scales linearly with \( \pareto \); that is,
\[
\flowLimVec^{(1)}(\omega \pareto) = \omega \cdot \flowLimVec^{(1)}(\pareto).
\]
Hence, the flow capacity constraint \eqref{optFlowConst} for input $\omega\pareto$ becomes
\begin{align*}
|\flowM^*(\omega\demM^{(1)}(\pareto) - \genM, \omega\flowLimVec^{(1)}(\pareto)) \e{_{|\K|}}| &=|\flowM^*(\demM^{(1)}(\omega \pareto) - \genM, \flowLimVec^{(1)}(\omega \pareto)) \e{_{|\K|}}| \\ &\leq \lambda^{(1)}  \flowLimVec^{(1)}(\omega \pareto) = \lambda^{(1)} \omega\flowLimVec^{(1)}(\pareto).
\end{align*}
By Lemma~\ref{flowInv}, the minimum-cost flow function is scale-invariant with respect to its arguments. Therefore, for any tuple \( (\demM, \genM, \flowLimVec) \), we have:
\[
\flowM^*(\omega(\demM - \genM), \omega \flowLimVec) = \omega \cdot \flowM^*(\demM - \genM, \flowLimVec).
\]
%Substituting \( \demM = \demM^{(1)}(\pareto) \) and \( \genM = \genM^{(1)}(\pareto) \), and applying the above identity, we find that the scaled variables \( \omega \cdot \genM^{(1)}(\pareto) \) and \( \omega \cdot \flowLimVec^{(1)}(\pareto) \) satisfy both the production balance and the capacity constraint in the scaled problem.

Now, consider the two optimization problems:
\begin{align}
\label{eq:optUnscaled}
\genM^{(1)}(\pareto) 
= \argmin_{\genM \in \mathbb{R}_+^{\nv \times \nk}} 
& \sum_{v \in \V} \sum_{k \in \K} \frac{1}{\gamma} c_{v,k} \gen[v,k]^\gamma \notag \\
\text{s.t.:} \quad 
& \e{_{\nv}}^T \genM = \e{_{\nv}}^T \demM^{(1)}(\pareto), \notag \\
& |\flowM^*(\demM^{(1)}(\pareto) - \genM, \flowLimVec^{(1)}(\pareto)) \e{_{|\K|}}| 
\leq \lambda^{(1)} \cdot \flowLimVec^{(1)}(\pareto),
\end{align}

\begin{align}
\label{eq:optScaled}
\genM^{(1)}(\omega \pareto) 
= \argmin_{\genM \in \mathbb{R}_+^{\nv \times \nk}} 
& \sum_{v \in \V} \sum_{k \in \K} \frac{1}{\gamma} c_{v,k} \gen[v,k]^\gamma \notag \\
\text{s.t.:} \quad 
& \e{_{\nv}}^T \genM = \omega \cdot \e{_{\nv}}^T \demM^{(1)}(\pareto), \notag \\
& |\flowM^*(\omega(\demM^{(1)}(\pareto) - \genM/\omega), \omega \flowLimVec^{(1)}(\pareto)) \e{_{|\K|}}| 
\leq \lambda^{(1)} \cdot \omega \cdot \flowLimVec^{(1)}(\pareto).
\end{align}

Now apply a change of variables in \eqref{eq:optScaled}: define \( \tilde{\genM} = \genM / \omega \). The objective becomes
\[
\sum_{v \in \V} \sum_{k \in \K} \frac{1}{\gamma} c_{v,k} (\omega \tilde{s}_{v,k})^\gamma = \omega^\gamma \sum_{v \in \V} \sum_{k \in \K} \frac{1}{\gamma} c_{v,k} \tilde{s}_{v,k}^\gamma.
\]
The constraints also simplify: the first becomes \( \e{_{\nv}}^T \tilde{\genM} = \e{_{\nv}}^T \demM^{(1)}(\pareto) \), and the second reduces, using the scale-invariance of \( \flowM^* \) and the fact that $\omega>0$, to
\[
|\flowM^*(\demM^{(1)}(\pareto) - \tilde{\genM}, \flowLimVec^{(1)}(\pareto)) \e{_{|\K|}}| \leq \lambda^{(1)} \cdot \flowLimVec^{(1)}(\pareto).
\]

We observe that, after the change of variables, the two optimization problems differ only by a constant multiplicative factor \( \omega^\gamma \) in the objective. This implies that $\genM^{(1)}(\pareto)$ is the optimizer of \eqref{eq:optUnscaled} only if \(\omega \genM^{(1)}(\pareto) \) is the optimizer of \eqref{eq:optScaled}. As a result, 
\[
\genM^{(1)}(\omega \pareto) = \omega \cdot \genM^{(1)}(\pareto),
\]
establishing the scale invariance of \( \genM^{(1)} \).
\end{proof}

The following result shows that the operational production matrix is SRC w.r.t.\ $\pareto$. The proof strategy is similar to the proof of Lemma \ref{lemmaProductionPlanning}.

\begin{lemma}
    \label{lemmaProductionOp}
    Suppose that $\genM^*$ and $\flowM^*$ are given by the minimal-cost source production and flow, respectively. Moreover, assume that there are no production capacity constraints, i.e., $\genMLim = \infty$ and either:
    \begin{enumerate}[label=(\roman*)]
        \item[\hypertarget{item1}{1.}] the flow cost function is quadratic, i.e., given by \eqref{flowCost} with $\beta = 2$, or
        \item[\hypertarget{item2}{2.}] there are no flow capacity constraints, i.e., $\lambda^{(1)} = \infty$, and the flow cost function is given by \eqref{wardrop}.
    \end{enumerate}
     Then, $\genM^{(1)}$ is an SRC function of the vertex weight vector $\pareto$ on $\mathbb{R}_+^{\nv} \setminus \{\bm{0}\}$. In particular, for any convergent sequence $(\pareto^l)_{l \in \mathbb{N}}$ with limit $\pareto^\infty \neq \bm{0}$ such that $\pareto^l \geq \pareto^\infty$ for all $l$, we have
    \[
        \lim_{l \to \infty} \genM^{(1)}(\pareto^l) = \genM^{(1)}(\pareto^\infty).
    \]
\end{lemma}

\begin{proof}
First we observe that under Case~\hyperlink{item2}{2}, the result follows immediately from Lemma~\ref{lemmaProductionPlanning} because, under the assumption of no flow constraints, $\genM^{(1)} = \genM^{(0)}$. Hence, it remains to prove the lemma under Case~\hyperlink{item1}{1}. 

Recall that $\genM^{(1)}(\pareto) = \genM^*\big(\demM^{(1)}(\pareto), \infty, \flowLimVec^{(1)}(\pareto), \lambda^{(1)}\big)$. Under the assumed quadratic flow cost ($\beta = 1$), the optimal flow is linear in $\pareto$, and thus the optimization problem \eqref{optProd} is convex with a strictly convex objective and linear constraints. Consequently, the solution $\genM^{(1)}(\pareto)$ exists and is unique for every $\pareto \in \mathbb{R}_+^{\nv} \setminus \{\bm{0}\}$ \citep{boyd2004convex}.

Fix a sequence $(\pareto^l)_{l \in \mathbb{N}}$ converging to $\pareto^\infty$ with $\pareto^l \geq \pareto^\infty$ for all $l$. Analogue to the proof in Lemma~\ref{lemmaProductionPlanning}, the boundedness of the aggregate demand implies that $\big(\genM^{(1)}(\pareto^l)\big)_{l \in \mathbb{N}}$ lies in a compact subset of $\mathbb{R}_+^{\nv \times \nk}$.

We next construct a sequence of feasible source production matrices that converges to $\genM^{(1)}(\pareto^\infty)$. To do so, for each edge $e \in \EE$, consider the ratio between the operational flow capacities for $\pareto^l$ and $\pareto^\infty$. 
Let $\nu_l$ be the smallest of these ratios, capped above by $1$, i.e.,  
\[
    \nu_l :=  \min \left\{ \min_{e \in \EE} \frac{\flowLim{e}^{(1)}(\pareto^l)}{\flowLim{e}^{(1)}(\pareto^\infty)},\, 1 \right\},
\]
and define
\[
    \widehat{\genM}^l := \demM^{(1)}(\pareto^l) - \nu_l \cdot \left(\demM^{(1)}(\pareto^\infty) - \genM^{(1)}(\pareto^\infty)\right).
\]
First, we verify feasibility of $\widehat{\genM}^l$ for the optimization problem \eqref{optProd} with input $(\demM^{(1)}(\pareto^l), \infty, \flowLimVec^{(1)}(\pareto^l), \lambda^{(1)})$.

\smallskip
\noindent\emph{Production balance}: As $\genM^{(1)}(\pareto^l)$ for the optimization problem with input $\pareto^l$, by taking limit $l\rightarrow \infty$, it follows that $\genM^{(1)}(\pareto^\infty)$ is feasible for the corresponding problem with input $\pareto^\infty$. Using this, we compute:
\[
    \e{_{\nv}}^T \widehat{\genM}^l = \e{_{\nv}}^T \demM^{(1)}(\pareto^l) - \nu_l \cdot \underbrace{\e{_{\nv}}^T \left(\demM^{(1)}(\pareto^\infty) - \genM^{(1)}(\pareto^\infty)\right)}_{=0} = \e{_{\nv}}^T \demM^{(1)}(\pareto^l),
\]
so Constraint \eqref{optProdBalance} is satisfied.

\smallskip
\noindent\emph{Flow capacity}: First, note that the assumed flow cost function does not depend on $\flowLimVec$; nevertheless, we retain both arguments of $\flowM^*$ and $\flow[e]^*$ in the notation for completeness. Using the scale invariance of $\flowM^*$ given by Lemma~\ref{flowInv}, we obtain:
\[
    \flowM^*\left(\demM^{(1)}(\pareto^l) - \widehat{\genM}^l, \flowLimVec(\pareto^l) \right) 
    = \nu_l \cdot \flowM^*\left(\demM^{(1)}(\pareto^\infty) - \genM^{(1)}(\pareto^\infty), \flowLimVec(\pareto^l)\right).
\]
Thus, for every edge $e \in \EE$, it holds that
\[
    \left| \flow[e]^*\left(\demM^{(1)}(\pareto^l) - \widehat{\genM}^l,\flowLimVec(\pareto^l)\right) \right| 
    = \nu_l \cdot \left| \flow[e]^{(1)}(\pareto^\infty) \right| 
    \leq \lambda^{(1)} \cdot \flowLim{e}^{(1)}(\pareto^l),
\]
where the inequality follows from $\nu_l \leq 1$ by definition and from the feasibility of $\genM^{(1)}(\pareto^\infty)$, which ensures that the flow-capacity constraint \eqref{optFlowConst} holds. As a result, Constraint~\eqref{optFlowConst} is also satisfied when the source production is given by $\widehat{\genM}^l$.

\smallskip
\noindent\emph{Non-negativity}: Using that $\demM^{(1)}(\pareto) = \mathrm{diag}(\pareto) \bm{Q}$ , we find:
\[
    \widehat{\genM}^l 
    = \left(\mathrm{diag}(\pareto^l) - \nu_l \cdot \mathrm{diag}(\pareto^\infty)\right) \bm{Q} + \nu_l \cdot \genM^{(1)}(\pareto^\infty).
\]
Since $\pareto^l \geq \pareto^\infty$ and $\nu_l \leq 1$, both summands are non-negative, ensuring that $\widehat{\genM}^l \geq 0$.

\smallskip
Altogether, we obtain that $\widehat{\genM}^l$ is feasible for the problem defining $\genM^{(1)}(\pareto^l)$. Moreover, by Lemma~\ref{planningFlowInv}, $\nu_l \to 1$ as $l \to \infty$, and so
\[
    \lim_{l \to \infty} \widehat{\genM}^l = \genM^{(1)}(\pareto^\infty).
\]

Finally, let $\left(\genM^{l_m}\right)_{m \in \mathbb{N}}$ be an arbitrary convergent subsequence of $\genM^{(1)}(\pareto^l)$. Since $\genM^{l_m}$ is optimal for the input $\pareto^{l_m}$ and $\widehat{\genM}^{l_m}$ is feasible for the same problem, we obtain
\[\sum_{v\in \V}\sum_{k\in \K} \frac{1}{\gamma} c_{v,k}\left(s_{v,k}^{l_m}\right)^\gamma \leq \sum_{v\in \V}\sum_{k\in \K} \frac{1}{\gamma} c_{v,k}\left(\hat{s}_{v,k}^{l_m}\right)^\gamma.\]
Taking the limit $m\rightarrow \infty$ from both sides, we obtain
\[c_s\left(\lim_{m\rightarrow \infty }\genM^{l_m}\right)\leq c_s\left(\genM^{(1)}(\pareto^\infty)\right),\]
where $c_s$ denotes the objective function. By strict convexity of the objective and uniqueness of the optimizer, we conclude that $\lim_{m \to \infty} \genM^{l_m} = \genM^{(1)}(\pareto^\infty)$. Since every convergent subsequence has this same limit, and the sequence lies in a compact set, the full sequence converges, i.e., 
\[
    \lim_{l \to \infty} \genM^{(1)}(\pareto^l) = \genM^{(1)}(\pareto^\infty),
\]
as claimed.
\end{proof}

\subsection{Scale invariance and SRC in the emergency phase} \label{subsec:emergencyProperties}
The results in this section show that the scale invariance and SRC properties of flows, flow capacities, source production, and sink requirement as functions of $\pareto$ are preserved by the cascade process. The following two lemmas show these properties for a given cascade sequence $\cascadeSeqRand$, which is fixed for all scalings $\omega$. These results are proven using induction arguments, where the induction basis follows from Lemmas \ref{flowInv}--\ref{lemmaProductionOp}.

We emphasize that all functions in the emergency phase are random variables, each fully determined by the vertex weight $
\pareto$ and the cascade sequence 
$\cascadeSeqRand$. In the proofs that follow, for any random variable $A$, we use the shorthand notation 
$A(\pareto,\cascadeSeqDet)$ to denote random variable $A$ conditioned on 
$\pareto$ and the event 
$\cascadeSeqRand=\cascadeSeqDet$.
\begin{lemma} \label{scaleInvCascade}
Suppose that $\genM^*$ and $\flowM^*$ are given by the minimal-cost source production and flow, respectively, and that there are no production capacity constraints, i.e., $\genMLim = \infty$. Let $\cascadeSeqDet \in \cascadeSet$ be a fixed cascade sequence. Then, for every stage $t$ of the cascade, the flow matrix $\flowM^{(t)}$, flow capacity vector $\flowLimVec^{(t)}$, source production matrix $\genM^{(t)}$, and sink requirement matrix $\demM^{(t)}$ are scale-invariant functions of the vertex weight vector $\pareto$, conditional on the cascade $\cascadeSeqRand=\cascadeSeqDet$.
\end{lemma}

\begin{proof}
We prove the result by induction on the cascade stage $t$.

\noindent\textbf{Base case $(t = 1)$:}
\begin{enumerate}
    \item By Lemma~\ref{planningFlowInv}, the flow capacity vector $\flowLimVec^{(1)}$ is scale-invariant  with respect to $\pareto$.
    
    \item The sink requirement matrix is given by $\demM^{(1)}(\pareto) = \mathrm{diag}(\pareto) \cdot \bm{Q}$. By linearity, this function is scale-invariant with respect to $\pareto$.

    \item By Lemma~\ref{lemmaProductionOpInv}, the source production matrix $\genM^{(1)}$ is scale-invariant  with respect to $\pareto$.

    \item Since both $\demM^{(1)}$ and $\genM^{(1)}$ are scale-invariant, the netput $\bm{U}^{(1)} = \demM^{(1)} - \genM^{(1)}$ is scale-invariant. Then, by Lemma~\ref{flowInv}, the flow matrix $\flowM^{(1)}(\pareto) = \flowM^*(\bm{U}^{(1)}(\pareto), \flowLimVec^{(1)}(\pareto))$ is scale-invariant with respect to $\pareto$.
\end{enumerate}

\noindent\textbf{Induction step:}
Assume the scale-invariance holds at stage $t$. We show it holds at stage $t+1$ for each of the four functions.

\textit{Flow capacity vector.} According to Equation~\eqref{flowCapFormula}, for each $e \in \EE$:
\begin{itemize}
    \item If $e \notin \cascadeSeqDet^{(t)}$, then $\flowLim{e}^{(t+1)}(\omega\pareto, \cascadeSeqDet) = \flowLim{e}^{(t)}(\omega\pareto, \cascadeSeqDet) = \omega \flowLim{e}^{(t)}(\pareto, \cascadeSeqDet) = \omega \flowLim{e}^{(t+1)}(\pareto, \cascadeSeqDet)$.
    \item If $e \in \cascadeSeqDet^{(t)}$ and $e$ failed completely, i.e. $e\in \cascadeSeqDet_R^{(t)}$, then $\flowLim{e}^{(t+1)}(\omega\pareto, \cascadeSeqDet) = 0$ for all $\omega$, so $\flowLim{e}^{(t+1)}(\omega\pareto, \cascadeSeqDet) = \omega \flowLim{e}^{(t+1)}(\pareto, \cascadeSeqDet) = 0$. Otherwise, i.e., if $e$ failed partially, then by scale-invariance of $\flowM^{(t)}$ and $\flowLimVec^{(t)}$, from Equation \eqref{exceedance} we obtain
\begin{equation}\exc{e}^{(t)}(\omega\pareto, \cascadeSeqDet) = \frac{\left|\left(\flowM^{(t)}(\omega\pareto, \cascadeSeqDet)\e{_{|\K|}}\right)_e\right|}{\flowLim{e}^{(t)}(\omega\pareto, \cascadeSeqDet)} = \frac{\omega\left|\left(\flowM^{(t)}(\pareto, \cascadeSeqDet)\e{_{|\K|}}\right)_e\right|}{\omega\flowLim{e}^{(t)}(\pareto, \cascadeSeqDet)} = \exc{e}^{(t)}(\pareto, \cascadeSeqDet). \label{exceedanceOmega}\end{equation}
As a result, 
\[\flowLim{e}^{(t+1)}(\omega\pareto, \cascadeSeqDet) = \omega\cdot \failFactor_e^{(t)}\left(\exc{e}^{(t)}(\pareto, \cascadeSeqDet)\right)\cdot \flowLim{e}^{(t)}(\pareto, \cascadeSeqDet) = \omega \flowLim{e}^{(t+1)}(\pareto, \cascadeSeqDet).\]

\end{itemize}

\textit{Sink and source matrices.} If the graph is not disconnected at $t+1$, then these matrices are unchanged and remain scale-invariant. Otherwise, for each connected component $\widetilde{\G} = (\widetilde{\V}, \widetilde{\EE})$ and $k\in \K$, using Equation~\eqref{demGenCascade} and scale-invariance at time $t$:
\[\eta_k^{(t)}(\widetilde{\V}, \omega\pareto, \cascadeSeqDet) = \frac{\sum_{v\in \widetilde{V}}\dem[v,k]^{(t)}(\omega \pareto, \cascadeSeqDet)}{\sum_{v\in \widetilde{V}}\gen[v,k]^{(t)}(\omega \pareto, \cascadeSeqDet)} = \frac{\omega \sum_{v\in \widetilde{V}}\dem[v,k]^{(t)}(\pareto, \cascadeSeqDet)}{\omega \sum_{v\in \widetilde{V}}\gen[v,k]^{(t)}(\pareto, \cascadeSeqDet)} = \eta_k^{(t)}(\widetilde{\V}, \pareto, \cascadeSeqDet).\]
Hence, from \eqref{demGenCascade} and scale invariance of $\demM^{(t)}$ and $\genM^{(t)}$ follows that $\demM^{(t+1)}$ and $\genM^{(t+1)}$ are also scale invariant with respect to $\pareto$. 

\textit{Flow matrix.} By Lemma~\ref{flowInv}, the scale-invariance of $\bm{U}^{(t+1)} = \demM^{(t+1)} - \genM^{(t+1)}$ and $\flowLimVec^{(t+1)}$ implies that $\flowM^{(t+1)}$ is scale-invariant  with respect to $\pareto$.

This completes the inductive step and the proof.
\end{proof}

The following lemma shows the corresponding SRC properties, conditional on the occurrence of a particular cascade $\cascadeSeqRand = \cascadeSeqDet$. 

\begin{lemma} \label{SRCCascade}
Suppose that $\genM^*$ and $\flowM^*$ are given by the minimal-cost source production and flow functions, respectively. Further assume that there are no production capacity constraints, i.e., $\genMLim = \infty$, and that either:
\begin{enumerate}[label=(\roman*)]
    \item the flow cost function is quadratic, as defined in \eqref{flowCost} with $\beta = 2$, or
    \item the flow cost is given by \eqref{wardrop} and there are no flow capacity constraints, i.e., $\lambda^{(1)} = \infty$.
\end{enumerate}
Let $\cascadeSeqDet \in \cascadeSet$ be a fixed cascade sequence. Then, for each time step $t \in \mathbb{N}$, the flow $\flowM^{(t)}$, flow capacity vector $\flowLimVec^{(t)}$, flow exceedance vector $\bm{\exc{}}^{(t)}$, source production matrix $\genM^{(t)}$, and sink requirement matrix $\demM^{(t)}$ are SRC functions of the vertex weight vector $\pareto$ on $\mathbb{R}_+^{\nv} \setminus \{ \bm{0} \}$, conditional on cascade $\cascadeSeqRand = \cascadeSeqDet$ occurring.
\end{lemma}

\begin{proof}
We proceed by induction on the cascade stage \( t \in \mathbb{N} \).

\noindent\textbf{Base case (\( t = 1 \)).} At the operational stage, the cascade has not yet started, so the functions are independent of the particular cascade sequence \( \cascadeSeqRand = \cascadeSeqDet \). By Lemmas~\ref{contFlow}, \ref{planningFlowInv}, and \ref{lemmaProductionOp}, we conclude that \( \flowM^{(1)}(\pareto, \cascadeSeqDet) \), \( \flowLimVec^{(1)}(\pareto, \cascadeSeqDet) \), \( \genM^{(1)}(\pareto, \cascadeSeqDet) \), and \( \demM^{(1)}(\pareto, \cascadeSeqDet) \) are SRC functions of \( \pareto \) on \( \mathbb{R}_+^{\nv} \setminus \{\bm{0}\} \). Moreover, since $\bm{\exc{}}^{(1)}(\pareto,\cascadeSeqDet)$, defined in \eqref{exceedance}, is a continuous function of $\flowM^{(1)}(\pareto,\cascadeSeqDet)$ and $\flowLimVec^{(1)}(\pareto, \cascadeSeqDet)$, it follows that $\bm{\phi{}}^{(1)}(\pareto,\cascadeSeqDet)$ is also SRC w.r.t.\ $\pareto$.

\vspace{0.3cm}
\noindent\textbf{Inductive step.} Suppose that for some \( t \in \mathbb{N} \), all five functions \( \flowM^{(t)}(\pareto, \cascadeSeqDet) \), \( \flowLimVec^{(t)}(\pareto, \cascadeSeqDet) \), \( \bm{\exc{}}^{(t)}(\pareto, \cascadeSeqDet) \), \( \genM^{(t)}(\pareto, \cascadeSeqDet) \), and \( \demM^{(t)}(\pareto, \cascadeSeqDet) \) are SRC. We prove that the same holds at stage \( t+1 \).

\emph{Flow capacity.}  Let \( (\pareto^l)_{l \in\N} \subset \mathbb{R}_+^{\nv} \) be a sequence such that \( \pareto^l \geq \pareto^\infty \) for all \( l \) and \( \pareto^l \to \pareto^\infty \neq \bm{0} \), and partition the set of failures into partial and complete failures $c^{(t)} = c_P^{(t)}\cup c_R^{(t)}$. If an edge $e$ does not fail in step $t+1$, i.e., \( e \notin \cascadeSeqDet^{(t)} \), then \( \flowLim{e}^{(t+1)}(\pareto^l, \cascadeSeqDet) = \flowLim{e}^{(t)}(\pareto^l, \cascadeSeqDet) \), and the SRC property follows by induction.  
If $e$ fails completely, i.e., $e\in \cascadeSeqDet_R^{(t)}$, then $\flowLim{e}^{(t+1)}(\pareto, c) =  0$ for any $\pareto$, hence it is also SRC w.r.t.\ $\pareto$.  
Finally, if $e$ fails partially, i.e.,  \( e \in \cascadeSeqDet_P^{(t)} \), then using the update rule \eqref{flowCapFormula}, we have that
\[
\flowLim{e}^{(t+1)}(\pareto^l, \cascadeSeqDet) = \failFactor_e^{(t)}\left( \exc{e}^{(t)}(\pareto^l, \cascadeSeqDet) \right) \cdot \flowLim{e}^{(t)}(\pareto^l, \cascadeSeqDet).
\]
 Applying the induction hypothesis twice and continuity of \( \failFactor_e^{(t)} \),
\begin{align*}
\lim_{l \to \infty} \flowLim{e}^{(t+1)}(\pareto^l, \cascadeSeqDet) &= \lim_{l \to \infty} \left(\failFactor_e^{(t)}\left( \exc{e}^{(t)}(\pareto^l, \cascadeSeqDet) \right) 
\cdot \flowLim{e}^{(t)}(\pareto^l, \cascadeSeqDet)\right) \\
&= \failFactor_e^{(t)}\left( \lim_{l \to \infty} \exc{e}^{(t)}(\pareto^l, \cascadeSeqDet) \right) 
\cdot \lim_{l \to \infty} \flowLim{e}^{(t)}(\pareto^l, \cascadeSeqDet) \\
&\overset{ind.~hyp.}= \failFactor_e^{(t)}\left( \frac{\lim_{l \to \infty} \flow[e]^{(t)}(\pareto^l, \cascadeSeqDet)}{\lim_{l\to \infty} \flowLim{e}^{(t)}(\pareto^l, \cascadeSeqDet)} \right) 
\cdot \flowLim{e}^{(t)}(\pareto^\infty, \cascadeSeqDet) \\
&\overset{ind.~hyp.}= \failFactor_e^{(t)}\left( \exc{e}^{(t)}(\pareto^\infty, \cascadeSeqDet) \right) \cdot \flowLim{e}^{(t)}(\pareto^\infty, \cascadeSeqDet) \\
&= \flowLim{e}^{(t+1)}(\pareto^\infty, \cascadeSeqDet).
\end{align*}
Thus, \( \flowLimVec^{(t+1)} \) is SRC.

\emph{Sink requirement and source production.}  
From \eqref{demGenCascade}, the updates to \( \demM^{(t+1)} \) and \( \genM^{(t+1)} \) are continuous functions of the previous-stage values $\demM^{(t)}$ and $\genM^{(t)}$, which are SRC w.r.t.\ $\pareto$. Since continuity preserves SRC, it follows that the functions \( \demM^{(t+1)} \) and \( \genM^{(t+1)} \) are SRC w.r.t.\ $\pareto$.

\emph{Flow matrix.}  
We have
\[
\flowM^{(t+1)}(\pareto, \cascadeSeqDet) = \flowM^*\left( \bm{U}^{(t+1)}(\pareto, \cascadeSeqDet), \flowLimVec^{(t+1)}(\pareto, \cascadeSeqDet) \right),
\]
where \( \bm{U}^{(t+1)} = \demM^{(t+1)} - \genM^{(t+1)} \). By Lemma~\ref{contFlow}, the flow mapping \( \flowM^* \) is continuous in both arguments. Since the inputs are SRC, the output is SRC as well. 

\emph{Exceedance.} Consider any edge $e\in \EE$, that has not failed completely at any prior or current stage. Then $e$ remains in the graph and its capacity is strictly positive at every stage. In particular, since each partial failure multiplies capacity by a factor in $[l_e,1]$, there exists $\varepsilon_e>0$ (e.g., $\varepsilon_e:=(l_e)^{t+1}\,\flowLim{e}^{(1)}$) such that $\flowLim{e}^{(t+1)}(\pareto^l,\cascadeSeqDet) \geq \varepsilon_e $. This implies that $\flowLim{e}^{(t+1)}(\pareto^\infty,\cascadeSeqDet) \geq \varepsilon >0$. Consequently, the exceedances $\exc{e}^{(t+1)}(\pareto^l, \cascadeSeqDet)$ and $\exc{e}^{(t+1)}(\pareto^\infty, \cascadeSeqDet)$ are well defined. Hence, $\bm{\phi{}}^{(t+1)}(\pareto,\cascadeSeqDet)$ is SRC, which again follows from the continuity of $\bm{\exc{}}^{(t+1)}$ in $\flowM^{(t+1)}$ and $\flowLimVec^{(t+1)}$ and the SRC properties of the flow matrix $\flowM^{(t+1)}(\pareto, \cascadeSeqDet)$ and the flow capacity vector $\flowLimVec^{(t+1)}(\pareto, \cascadeSeqDet)$.

This completes the induction, and the result follows.
\end{proof}
\subsection{Probabilistic properties of cascades} \label{subsec:probabilityCascadeProperties}
Up to this point, we have established that at every phase of the model, all functions of interest---namely, the flow matrix, flow capacity vector, source production matrix, and sink requirement matrix---are SRC and scale-invariant with respect to $\pareto$. We now turn to verifying that Statement~\hyperlink{assumptionPart2}{2} of Assumption~\ref{assumption} holds under the setting of Propositions~\ref{propClassFunc1} and \ref{propClassFunc2}. In particular, the next result shows that the probability of observing a particular cascade $\cascadeSeqRand = \cascadeSeqDet$, given the vertex weight vector $\pareto$, is SRC in~$\pareto$ and does not depend on its scale $\omega > 0$. The following lemma states this formally.
\begin{lemma}
    Suppose that $\genM^*$ and $\flowM^*$ are given by the minimal-cost source production and flow functions, respectively. Further assume that there are no production capacity constraints, i.e., $\genMLim = \infty$, and that either:
\begin{enumerate}[label=(\roman*)]
    \item the flow cost function is quadratic, as defined in \eqref{flowCost} with $\beta = 2$, or
    \item the flow cost is given by \eqref{wardrop} and there are no flow capacity constraints, i.e., $\lambda^{(1)} = \infty$.
\end{enumerate} Consider an arbitrary cascade sequence $\cascadeSeqDet\in \cascadeSet$. Then, scale $\omega>0$ does not influence the probability that cascade $\cascadeSeqDet$ occurs, given vertex weight vector $\omega\pareto$. In particular,
    \[\PR{\cascadeSeqRand = \cascadeSeqDet~|~ \omega \pareto} = \PR{\cascadeSeqRand = \cascadeSeqDet~|~\pareto} ~ \forall ~\omega>0.\]
    Moreover, the probability that cascade $\cascadeSeqRand = \cascadeSeqDet$ occurs, given $\pareto$ is an SRC function with respect to $\pareto$.
    \label{cascadeProb}
\end{lemma}

\begin{proof}
    Let $\omega>0$. Consider a particular cascade $\cascadeSeqDet = (\cascadeSeqDet^{(1)},\dots, \cascadeSeqDet^{(t_\cascadeSeqDet)})$, with $t_\cascadeSeqDet := |\cascadeSeqDet|$. Recall that $\cascadeSeqDet^{(t)} = \cascadeSeqDet_P^{(t)}\cup\cascadeSeqDet_R^{(t)}$ denotes the set of edges that have failed at time step $t\in \{1,\dots,t_\cascadeSeqDet\}$, where $\cascadeSeqDet_P^{(t)}$ and $\cascadeSeqDet_R^{(t)}$ denote the sets of partial and complete failures, respectively. To establish the independence of the probability of a particular cascade on the scaling factor $\omega$, by induction argument, it suffices to prove the following claim: conditioned on a fixed history of edge failures, the probability that any edge \( e \in \EE \) fails partially or completely at stage \( t \) remains unchanged under any positive scaling of the vertex weight vector. More precisely, for all \( t \in \{2, \dots, t_\cascadeSeqDet\} \), any failure history \( \{\cascadeSeqRand^{(1)} = \cascadeSeqDet^{(1)}, \dots, \cascadeSeqRand^{(t-1)} = \cascadeSeqDet^{(t-1)}\} \), and any scaling factor \( \omega > 0 \), we have
\begin{multline}
\PR{e \text{ fails partially at stage } t \mid \omega \pareto, \cascadeSeqRand^{(1)} = \cascadeSeqDet^{(1)}, \dots, \cascadeSeqRand^{(t-1)} = \cascadeSeqDet^{(t-1)})}
\\=
\PR{e \text{ fails partially at stage } t \mid \pareto,\cascadeSeqRand^{(1)} = \cascadeSeqDet^{(1)}, \dots, \cascadeSeqRand^{(t-1)} = \cascadeSeqDet^{(t-1)})}
\label{eq:failPartial_prob_equal}
\end{multline}
\begin{multline}
\PR{e \text{ fails completely at stage } t \mid \omega \pareto, \cascadeSeqRand^{(1)} = \cascadeSeqDet^{(1)}, \dots, \cascadeSeqRand^{(t-1)} = \cascadeSeqDet^{(t-1)})}
\\=
\PR{e \text{ fails completely at stage } t \mid \pareto,\cascadeSeqRand^{(1)} = \cascadeSeqDet^{(1)}, \dots, \cascadeSeqRand^{(t-1)} = \cascadeSeqDet^{(t-1)})}
\label{eq:failComplete_prob_equal}
\end{multline}

In our framework, edge failures are governed by exceedance probabilities: an edge \( e \) that has not yet failed its maximum $n_e$ times, may fail independently with probability \( p_{e,c}(\exc{e}^{(t)}) + p_{e,r}(\exc{e}^{(t)}) \), where \( \exc{e}^{(t)} \) denotes the exceedance experienced by edge \( e \) at stage \( t \). Let $\nu^{(t)}_e(c)$ denote the number of times edge $e$ has failed (either partially or completely) before time $t$ under cascade $c$. Thus, the conditional partial failure probability can be expressed as
\begin{multline*}
\PR{e \text{ fails partially at stage } t \mid \omega \pareto, \cascadeSeqRand^{(1)} = \cascadeSeqDet^{(1)}, \dots, \cascadeSeqRand^{(t-1)} = \cascadeSeqDet^{(t-1)})}\\
=
p_{e,c}\left(\exc{e}^{(t)}(\omega \pareto, \cascadeSeqDet^{(1)}, \dots,  \cascadeSeqDet^{(t-1)})\right) \cdot \vmathbb{1}\{\nu_e^{(t)} \neq n_e\}.
\end{multline*}
Recall from the proof of Lemma~\ref{scaleInvCascade} that the exceedance vector \( \bm{\exc{}}^{(r)} \) is invariant under scaling $\omega$ of \( \pareto \), see~\eqref{exceedanceOmega}. That is, 
\[
\bm{\exc{}}^{(r)}(\omega \pareto, \cascadeSeqDet^{(1)}, \dots,  \cascadeSeqDet^{(r-1)})
=
\bm{\exc{}}^{(r)}(\pareto, \cascadeSeqDet^{(1)}, \dots,  \cascadeSeqDet^{(t-1)}).
\]
Thus, it follows immediately that Equation~\eqref{eq:failPartial_prob_equal} holds. Using an analogous approach, we can show that Equation~$\eqref{eq:failComplete_prob_equal}$ holds as well.

Finally, note that the initial edge failures are selected at random according to the probability law $p^{(1)}(\cdot)$ and independently of \( \omega \). Therefore, the entire failure sequence remains distributionally invariant under positive rescaling of \( \pareto \), and the result follows by inductively applying the equalities in \eqref{eq:failPartial_prob_equal} and \eqref{eq:failComplete_prob_equal} across all stages of the cascade.

    It remains to show that $\PR{\cascadeSeqRand=\cascadeSeqDet~|~\pareto}$ is SRC. Note that $p_{e,c}(\cdot)$ and $p_{e,r}(\cdot)$ are by assumption continuous functions (see Step 4 in Section~\ref{emergencyPhase}) and from Lemma \ref{SRCCascade} it follows that $\exc{e}^{(t)}(\pareto, \cascadeSeqDet)$ is SRC w.r.t.\ $\pareto$ for all $t\in \N$ and all $e\in \EE$. Hence, since $\PR{\cascadeSeqRand = \cascadeSeqDet~|~\pareto}$ is a product of SRC functions, it is itself an SRC function of~$\pareto$.
\end{proof}
\section{Proofs of Propositions \ref{propClassFunc1} and \ref{propClassFunc2}} \label{subsec:proofsPropositions}
With all auxiliary results established in Section~\ref{proofs}, the proofs of Propositions~\ref{propClassFunc1} and~\ref{propClassFunc2} follow by verifying the conditions of Assumption~\ref{assumption}. These conditions rely on the scale-invariance, SRC, and boundedness of the cascade cost function, as well as the SRC and invariant behavior of the cascade probability. 

First, we prove Proposition~\ref{propClassFunc1} by consecutively applying the results derived so far.
\begin{proof}[Proof of Proposition \ref{propClassFunc1}.]
Consider the functions $\flowM^*$, $\genM^*$, $\genMLim$, and $\cascadeCost$ as described in Statements 1--4 in the proposition. In order to show that Assumption \ref{assumption} is satisfied, we first need to show that $\cascadeCost$ is $\rho$-scale-invariant and SRC w.r.t.\ $\pareto$, given that a cascade $\cascadeSeqRand = \cascadeSeqDet$ occurs. Recall that $\cascadeCost$ is the generalized cascade size, defined in~\eqref{generalizedFailureCost}, i.e., 
\begin{equation*}
\cascadeCost = \sum_{v \in \V} \sum_{k \in \K} \constCostFunction{v,k} \left( \dem[v,k]^{(1)} - \dem[v,k]^{(end)} \right)^\rho.
\end{equation*}
Consider a fixed cascade $\cascadeSeqDet \in \cascadeSet$. From Lemma~\ref{scaleInvCascade}, we know that the sink requirements matrix is scale-invariant at every stage of the cascade. Therefore, we find that for any $\omega>0$
\[\cascadeCost(\omega\pareto, \cascadeSeqDet) = \sum_{v \in \V} \sum_{k \in \K} \constCostFunction{v,k} \left( \dem[v,k]^{(1)}(\omega\pareto, \cascadeSeqDet) - \dem[v,k]^{|\cascadeSeqDet|}(\omega \pareto, \cascadeSeqDet) \right)^\rho = \omega^\rho \cascadeCost(\pareto, \cascadeSeqDet).\]
Hence, indeed, $\cascadeCost$ is $\rho$-scale-invariant. 

To prove the SRC property of $\cascadeCost$, we first observe that $\cascadeCost$ is a continuous function of $\demM^{(1)}$ and $\demM^{(end)}$. Moreover, due to Lemma~\ref{SRCCascade}, we have that the sink requirement matrices $\demM^{(1)}(\pareto,\cascadeSeqDet)$ and $\demM^{(end)}(\pareto,\cascadeSeqDet)$ are SRC w.r.t.\ $\pareto$. Since continuity preserves the SRC property, we conclude that $\cascadeCost$ is also an SRC function of $\pareto$ for, conditioned on a fixed cascade $\cascadeSeqRand = \cascadeSeqDet$. 

It remains to show that $\cascadeCost$ is appropriately bounded. For the conditional cascade cost $\cascadeCost$ given $\pareto $, denoted by $\cascadeCost(\pareto)$, we find that
\[\cascadeCost(\pareto)\leq \sum_{v\in \V}\sum_{k\in \K} \constCostFunction{v,k}\left(\dem[v,k]^{(1)}(\pareto)\right)^\rho =  \sum_{v\in \V}\sum_{k\in \K} \constCostFunction{v,k} q_{v,k}^\rho \paretoI{v}^\rho \leq \nk \max_{v\in\V, k\in \K} \constCostFunction{v,k}q_{v,k}^\rho \sum_{v\in \V} \paretoI{v}^\rho. \]
Thus, setting $\boundZ = \max_{v\in \V,k\in\K}\constCostFunction{v,k}q_{v,k}^\rho$ we obtain that
\[\cascadeCost(\pareto)\leq \boundZ\cdot \sum_{v\in \V}\paretoI{v}^\rho.\] 
This proves that Statement~\hyperlink{assumptionPart1}{1} of Assumption~\ref{assumption} is satisfied. Moreover, Statement~\hyperlink{assumptionPart2}{2} of the assumption is satisfied thanks to Lemma~\ref{cascadeProb}, which completes the proof.
\end{proof}  
We use a similar proof approach to prove Proposition~\ref{propClassFunc2}. However, unlike the previous case, this result also allows $\cascadeCost$ to be given by the cascade flow cost function. Accordingly, we must also establish additional properties of this cost function.
\begin{proof}[Proof of Proposition \ref{propClassFunc2}.]
Consider the functions $\flowM^*$, $\genM^*$, $\genMLim$, and $\cascadeCost$ as described in Steps 1--4 in the proposition. If $\cascadeCost$ is the generalized cascade cost, as in \eqref{generalizedFailureCost}, then the proof is analogue to the proof of Proposition~\ref{propClassFunc1}, hence we omit this part of the proof here and only focus on the case of the cascade flow cost, defined in Equation~\eqref{deltaFlowCost}. First, we prove the 1-scale-invariance. Lemma \ref{scaleInvCascade} shows that $\flowM^{(t)}(\omega\pareto, \cascadeSeqDet) = \omega \flowM^{(t)}(\pareto, \cascadeSeqDet)$ and $\flowLimVec^{(t)}(\omega\pareto, \cascadeSeqDet) = \omega \flowLimVec^{(t)}(\pareto, \cascadeSeqDet)$ for every cascade stage $t$. Hence, 
\begin{align*}\cascadeCost(\omega\pareto, \cascadeSeqDet) &= \sum_{e\in \EE}\left(d_e |\flow[e]^{(end)}(\omega \pareto, \cascadeSeqDet)| +\frac{1}{\beta} b_e \flowLim{e}(\omega \pareto, \cascadeSeqDet)^{(end)}\left(\exc{e}^{(end)}(\omega \pareto, \cascadeSeqDet)\right)^\beta \right)\\
&\qquad -\sum_{e\in \EE}\left(d_e |\flow[e]^{(1)}(\omega \pareto, \cascadeSeqDet)| +\frac{1}{\beta} b_e \flowLim{e}(\omega \pareto, \cascadeSeqDet)^{(1)}\left(\exc{e}^{(1)}(\omega \pareto, \cascadeSeqDet)\right)^\beta \right)\\
&= \sum_{e\in \EE}\left(d_e \omega |\flow[e]^{(end)}(\pareto, \cascadeSeqDet)| +\frac{1}{\beta} b_e  \omega \flowLim{e}(\pareto, \cascadeSeqDet)^{(end)}\left(\exc{e}^{(end)}(\pareto, \cascadeSeqDet)\right)^\beta \right)\\
&\qquad -\sum_{e\in \EE}\left( d_e \omega |\flow[e]^{(1)}(\pareto, \cascadeSeqDet)| +\frac{1}{\beta} b_e \omega \flowLim{e}(\pareto, \cascadeSeqDet)^{(1)}\left(\exc{e}^{(1)}(\pareto, \cascadeSeqDet)\right)^\beta \right) = \omega \cascadeCost(\pareto, \cascadeSeqDet),\end{align*}
showing the $1$-scale-invariance of the cascade flow cost. 

The SRC property of $\cascadeCost$ follows from Lemma~\ref{SRCCascade}. Specifically, the lemma yields that, given cascade $\cascadeSeqRand = \cascadeSeqDet$ occurs, the functions $\flowM^{(t)}$, $\flowLimVec^{(t)}$, and $\bm{\exc{}}^{(t)}$ are SRC w.r.t.\ $\pareto$ for every cascade stage $t$. We observe that $\cascadeCost$ is continuous w.r.t.\ the aforementioned functions. Since continuity preserves the SRC property, we obtain that  $\cascadeCost(\pareto, \cascadeSeqDet)$ is SRC w.r.t.\ $\pareto$, given that cascade $\cascadeSeqRand = \cascadeSeqDet$ occurs. 

It remains to show that $\cascadeCost$ is bounded. We leverage the fact that for every commodity $k\in \K$, at any given time the flow of this commodity on edge $e\in \EE$ is bounded by the total requirement of this commodity. In particular,
\[|\flow[e,k]| \leq \sum_{v\in \V} q_{v,k}\paretoI{v}.\]
As a result, 
\[|\flow[e]|\leq \sum_{k\in \K}\sum_{v\in \V} q_{v,k}\paretoI{v}\leq \max_{v\in \V, k\in \K} \{q_{v,k}\} \sum_{v\in \V} \paretoI{v}.\]
Next, we bound the flow capacity at cascade stage $r$. Recall that edges that fully failed are removed from the graph. Therefore, the capacity of an edge $e$ that is still a part has been decreased at most by a factor $l_{\min}^r$. Moreover, the operational flow capacity is bounded from below by $\varepsilon_{\min}\sum_{v\in \V} \paretoI{v}$ and from above by $\flowLim{e}^{(0)}$. Therefore,
\[l_{\min}^{r} \varepsilon_{\min}\sum_{v\in \V}\paretoI{v} \leq \flowLim{e}^{(l)}\leq \max\left\{\tau |\flow[e]^{(0)}|, \varepsilon_{\min} \sum_{v\in \V} \paretoI{v}\right\}\leq \max\{\varepsilon_{\min},\tau \max_{v\in\V, k\in \K} \{q_{v,k}\}\}\sum_{v\in \V} \paretoI{v}.\]
Let $\kappa = \max\{\varepsilon_{\min},\tau \max_{v\in\V, k\in \K} \{q_{v,k}\}\}$. Note that each cascade has at most $r^* = \sum_{e\in \EE} n_e$ steps. Using this and applying the above bounds, we obtain that the flow cost at the end of the cascade can be bounded from above as follows
\begin{align*}c_f(\flowM^{(end)}, \flowLimVec^{(end)}) &\leq \sum_{e\in \EE}\left( d_e \kappa \sum_{v\in \V}\paretoI{v} + \frac{1}{\beta}b_e \kappa \sum_{v\in \V} \paretoI{v} \left(\max_{v\in \V, k\in \K} \{q_{v,k}\}/l_{\min}^{r^*}\right)^\beta\right)\leq \boundZ \sum_{v\in \V} \paretoI{v},\end{align*}
with $\boundZ = \nee \cdot \kappa\cdot \max_{e\in \EE} \left\{d_e + \frac{b_e}{\beta}\left(\max_{v\in \V, k\in \K} \{q_{v,k}\}/l_{\min}^{r^*}\right)^\beta\right\}$.
%Therefore, using H\"older's inequality and the fact that $\pareto>0$ and $\beta>1$, we find
%\[|\flow[e]|^\beta \leq \max_{v\in \V, k\in \K} \{q_{v,k}^\beta\} \left(\sum_{v\in \V} \paretoI{v}\right)^\beta \leq \max_{v\in \V, k\in \K} \{q_{v,k}^\beta\} \cdot \nv^{\beta-1} \sum_{v\in \V} \paretoI{v}^\beta.\]
%Thus, we find that $\Delta c_f(\pareto)\leq \boundZ\sum_{v\in \V}\paretoI{v}^\beta$, with $\boundZ = 2\cdot\max_{v\in \V, k\in \K} \{q_{v,k}^\beta\}\cdot  \nv^{\beta-1} \sum_{e\in \EE}( a_e b_e)/(\beta a_e^\beta)$ and 
Finally, since flow cost is non-negative, we obtain that
\[\cascadeCost = c_f(\flowM^{(end)}, \flowLimVec^{(end)}) - c_f(\flowM^{(1)}, \flowLimVec^{(1)})\leq c_f(\flowM^{(end)}, \flowLimVec^{(end)})\leq \boundZ\sum_{v\in \V} \paretoI{v}.\]
Thus, we conclude that Statement~\hyperlink{assumptionPart1}{1} of Assumption~\ref{assumption}
is satisfied. Moreover, Statement~\hyperlink{assumptionPart2}{2} of the assumption is satisfied due to Lemma \ref{cascadeProb}, which concludes the proof.
\end{proof}  

\section{Non-quasi-convex graph flows}\label{appendixExample}
In this section, we discuss an example of a graph that yields non-quasi-convex flows. As noted in Section~\ref{subsec:SourceFunctions}, this poses a challenge because, without quasi-convexity, the uniqueness of optimal solutions to the upper-level optimization problem~\eqref{optProd} is no longer guaranteed. This issue motivated the restriction to flow cost functions with parameter $\beta = 2$ in Proposition~\ref{propClassFunc1}.
 \begin{example}[Non-quasi-convexity in the $\bm{K_4}$ graph.]\label{exampleNonQuasi}
Consider the complete graph on four vertices $K_4$\added{ and the cost function in \eqref{flowCost} with $\beta = 3$, $a_e = b_e = 1$ for all $e\in \EE$, i.e., $c_f =\sum_{e\in \EE} \frac{1}{3} |f_e|^3$. For this choice of parameters, the flow cost function does not depend on the flow capacity $\flowLimVec,$ hence we drop this dependency in our notation}. For each edge, we assign a direction towards the vertex with the smaller index, which yields the following incidence matrix:
\[
\incM = \bordermatrix{
      &\textit{(1,2)}   & \textit{(1,3)} & \textit{(1,4)} &  \textit{(2,3)} & \textit{(2,4)} &\textit{(3,4)}\cr
\text{1} & 1 & 1 & 1 & 0 & 0 & 0 \cr
\text{2} & -1 & 0 & 0 & 1 & 1 & 0\cr
\text{3} & 0 & -1 & 0 & -1 & 0 & 1 
\cr
\text{4} &0 & 0 & -1 & 0 & -1 & -1 }.
\]
\deleted{We consider the cost function in \eqref{flowCost} with $\beta = 3$, $a_e = b_e = 1$ for all $e\in \EE$, i.e., $c_f =\sum_{e\in \EE} \frac{1}{3} |f_e|^3$. For this choice of parameters, the flow cost function does not depend on the flow capacity $\flowLimVec,$ hence we drop this dependency in our notation. }

Suppose that $\nk = 1$, i.e., there is only one commodity, and that the sink requirement matrix is given by $\demM = (1,0,0,0)^T$. We aim to show that, in this setting, $|\flowTot^*| = |\flowM^*|$ is not quasi-convex, i.e., \eqref{eq:quasiFlowTot} does not hold. We choose $\genM_1 = (0,1,0,0)^T$ and $\genM_2 = (0,0,1,0)^T$, which yields $\bm{U}_1 := \demM - \genM_1 = (1,-1, 0,0)^T$ and $\bm{U}_2 := \demM - \genM_2 = (1,0, -1,0)^T$. We study the flow $\flowM^*(\bm{U}(\eta))$, where $\bm{U}(\eta):= (1-\eta)\bm{U}_1 + \eta \bm{U_2}$, with $\eta\in [0,1]$. This choice of $\bm{U}(\eta)$ means that one unit of the commodity has to be transported to vertex 1, where a fraction $1-\eta$ of this commodity originates at vertex $2$ and a fraction $\eta$ originates at vertex 3. 

The optimal flows for $\eta\in \{0,\frac{1}{2},1\}$ are given in Table~\ref{tab:optFlowsEx}. Using these results, we observe that Equation~\eqref{eq:quasiFlowTot} does not hold for $\eta = \frac{1}{2}$ and edge $e = \textit{(1,4)}$, since
\begin{equation*}
|F_e(\bm{U}(\eta))| = (\sqrt{5} - 1)/4 \approx 0.309 > 0.292 \approx 1 - \sqrt{2}/2 =  \max\left\{|F_e(\bm{U}(0))|, |F_e(\bm{U}(1))|\right\}.    
\end{equation*}
\renewcommand{\arraystretch}{1.2}
\begin{table}[ht]
    \centering
    \begin{tabular}{|c|c|c|c|c|c|c|}
    \hline
      \diagbox[height = 0.7cm]{ $\eta$}{\textit{edge}}&\textit{(1,2)}   & \textit{(1,3)} & \textit{(1,4)} &  \textit{(2,3)} & \textit{(2,4)} &\textit{(3,4)} \\
      \hline
      $0$ & $\sqrt{2} - 1$ & $1 -\sqrt{2}/2$ &  $1 -\sqrt{2}/2$ & $-1 +\sqrt{2}/2$ & $-1 +\sqrt{2}/2$ & $0$\\
      \hline
      $1/2$ & $(5 - \sqrt{5})/8$ & $(5 - \sqrt{5})/8$& $(\sqrt{5} - 1)/4$&$ 0$& $(1 - \sqrt{5})/8$& $(1-\sqrt{5})/8$\\
      \hline
    $1$ & $1 -\sqrt{2}/2$& $\sqrt{2} - 1$&   $1 -\sqrt{2}/2$&$1 -\sqrt{2}/2$&$0$&$-1 +\sqrt{2}/2$\\
    \hline
    \end{tabular}
    \caption{Values of the optimal flow $\flowM^*(\bm{U}(\eta))$ for each edge and $\eta\in \{0,\frac{1}{2},1\}$, see Figure~\ref{fig:example}.}
    \label{tab:optFlowsEx}
\end{table}

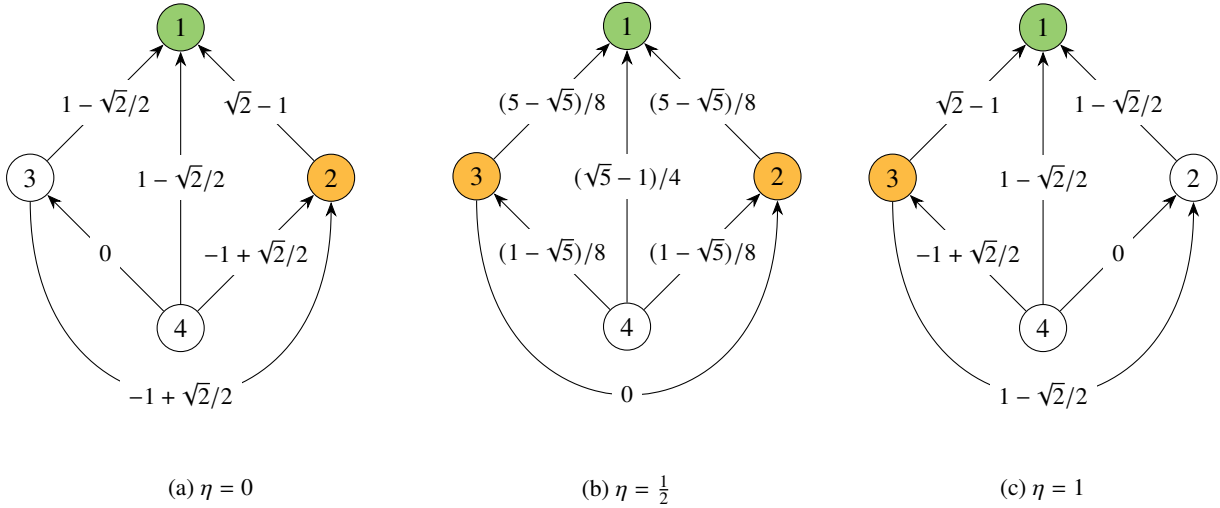
\begin{figure}[ht]
\begin{subfigure}[ht]{0.33\textwidth}
    \begin{tikzpicture}
  % Define the style for the nodes
  \tikzset{vertex/.style = {shape=circle,draw,minimum size=1.5em}}
  % Define the style for the edges
  \tikzset{edge/.style = {->,> = latex'}}

  % Nodes
  \node[vertex, fill = YellowGreen] (1) at (90:2) {1};
  \node[vertex, fill = Dandelion] (2) at (0:2) {2};
  \node[vertex] (3) at (180:2) {3};
  \node[vertex] (4) at (270:2) {4};

  % Edges
  \draw[{Stealth[length=2mm, width=1.5mm]}-] (1) -- node[fill =white]{\footnotesize$\sqrt{2} - 1$} (2); % Edge 1
  \draw[{Stealth[length=2mm, width=1.5mm]}-] (1) -- node[fill =white]{\footnotesize$1 - \sqrt{2}/2$} (3); % Edge 2
  \draw[{Stealth[length=2mm, width=1.5mm]}-] (1) -- node[fill = white] {\footnotesize $1 - \sqrt{2}/2$} (4); % Edge 3
  \draw[{Stealth[length=2mm, width=1.5mm]}-] (2) to [out=-90,in=-90, looseness =2.2] node[fill =white]{\footnotesize$-1 + \sqrt{2}/2$} (3); % Edge 4
  \draw[{Stealth[length=2mm, width=1.5mm]}-] (2) --  node[fill = white]{\footnotesize$-1 + \sqrt{2}/2$}(4); % Edge 5
  \draw[{Stealth[length=2mm, width=1.5mm]}-] (3) -- node[fill =white]{\footnotesize 0} (4); % Edge 6

\end{tikzpicture}
    \caption{$\eta = 0$}
    \label{fig:eta0}
\end{subfigure}
\begin{subfigure}[ht]{0.33\textwidth}
    \centering
    \begin{tikzpicture}
  % Define the style for the nodes
  \tikzset{vertex/.style = {shape=circle,draw,minimum size=1.5em}}
  % Define the style for the edges
  \tikzset{edge/.style = {->,> = latex'}}

  % Nodes
  \node[vertex, fill = YellowGreen] (1) at (90:2) {1};
  \node[vertex, fill = Dandelion] (2) at (0:2) {2};
  \node[vertex, fill = Dandelion] (3) at (180:2) {3};
  \node[vertex] (4) at (270:2) {4};

  % Edges
  \draw[{Stealth[length=2mm, width=1.5mm]}-] (1) -- node[fill =white]{\footnotesize$(5 - \sqrt{5})/8$} (2); % Edge 1
  \draw[{Stealth[length=2mm, width=1.5mm]}-] (1) -- node[fill =white]{\footnotesize $ (5 - \sqrt{5})/8$} (3); % Edge 2
  \draw[{Stealth[length=2mm, width=1.5mm]}-] (1) -- node[fill = white] {\footnotesize $(\sqrt{5}-1)/4$} (4); % Edge 3
  \draw[{Stealth[length=2mm, width=1.5mm]}-] (2) to [out=-90,in=-90, looseness =2.2] node[fill =white]{\footnotesize$0$} (3); % Edge 4
  \draw[{Stealth[length=2mm, width=1.5mm]}-] (2) --  node[fill = white]{\footnotesize$(1-\sqrt{5})/8$}(4); % Edge 5
  \draw[{Stealth[length=2mm, width=1.5mm]}-] (3) -- node[fill =white]{\footnotesize $(1 -\sqrt{5})/8$} (4); % Edge 6
\end{tikzpicture}
    \caption{$\eta = \frac{1}{2}$}
    \label{fig:eta12}
\end{subfigure}
\begin{subfigure}[ht]{0.33\textwidth}
    \centering
    \begin{tikzpicture}
  % Define the style for the nodes
  \tikzset{vertex/.style = {shape=circle,draw,minimum size=1.5em}}
  % Define the style for the edges
  \tikzset{edge/.style = {->,> = latex'}}

  % Nodes
  \node[vertex, fill = YellowGreen] (1) at (90:2) {1};
  \node[vertex] (2) at (0:2) {2};
  \node[vertex, fill = Dandelion] (3) at (180:2) {3};
  \node[vertex] (4) at (270:2) {4};

  % Edges
  \draw[{Stealth[length=2mm, width=1.5mm]}-] (1) -- node[fill =white]{\footnotesize$1 - \sqrt{2}/2$} (2); % Edge 1
  \draw[{Stealth[length=2mm, width=1.5mm]}-] (1) -- node[fill =white]{\footnotesize $\sqrt{2} - 1$} (3); % Edge 2
  \draw[{Stealth[length=2mm, width=1.5mm]}-] (1) -- node[fill = white] {\footnotesize$1 - \sqrt{2}/2$} (4); % Edge 3
  \draw[{Stealth[length=2mm, width=1.5mm]}-] (2) to [out=-90,in=-90, looseness =2.2] node[fill =white]{\footnotesize $1 - \sqrt{2}/2$} (3); % Edge 4
  \draw[{Stealth[length=2mm, width=1.5mm]}-] (2) --  node[fill = white]{\footnotesize0}(4); % Edge 5
  \draw[{Stealth[length=2mm, width=1.5mm]}-] (3) -- node[fill =white]{\footnotesize $-1 + \sqrt{2}/2$} (4); % Edge 6

\end{tikzpicture}
    \caption{$\eta = 1$}
    \label{fig:eta1}
\end{subfigure}
\caption{Optimal flow $\flowM^*(\bm{U}(\eta))$ on the $K_4$ graph for different values of parameter~$\eta$. In each case, one unit of flow needs to be transferred to vertex 1 (green). Depending on the value of $\eta$, the flow originates from vertex 2 and/or 3 (orange). Negative flows represent movement opposite to the direction of the edge.}
\label{fig:example}
\end{figure}
We note that our numerical analysis has also shown non-quasi-convex behavior in $K_4$ for other values of $\beta>2$. This shows that even though the minimum cost flow problem is convex, this property is typically not inherited by the optimizer, not even in a weaker form of quasi-convexity. 

In the remainder of this example, we show full calculations of the optimal flows presented in Table~\ref{tab:optFlowsEx}. First, we note that for every parameter $\eta\in [0,1]$, the corresponding optimal flow is unique, because it is the solution of a convex optimization problem with a strictly convex objective function. The optimal flows for $\eta \in \{0, \frac{1}{2}, 1\}$ are depicted in Figure~\ref{fig:example}. We observe that due to the symmetry of the graph, we can obtain the flows for $\eta = 1$ from the flows for $\eta = 0$, which we derive first.

\textit{Case 1: $\eta = 0$.} We observe that the flow from vertex 2 to vertex 1 can take one of the five paths: $p_1 = 2\rightarrow 1$, $p_2 = 2\rightarrow 3\rightarrow 1$, $p_3 = 2\rightarrow 4\rightarrow 1$, $p_4 = 2\rightarrow 3\rightarrow 4\rightarrow 1$, and $p_5 = 2\rightarrow 4\rightarrow 3\rightarrow 1$. The optimal flow is unique; therefore, because of the symmetry of the graph and the cost function, the amount of flow on the paths $p_2$ and $p_3$ and on the paths $p_4$ and $p_5$ is (pairwise) the same. Note that the only paths that utilize the edge $(3,4)$ are $p_4$ and $p_5$, and, since their flows are equal and counteracting, we know that $F_{(3,4)}^*(\bm{U}(0)) = 0$. To compute the remaining flows, let $x = F_{(1,2)}^*(\bm{U}(0))$. Edges $(1,3)$ and $(1,4)$ are utilized by paths $p_2 + p_4$ and $p_3 + p_5$, respectively. Hence, due to the symmetry of $p_2$ and $p_3$ as well as $p_4$ and $p_5$, we can conclude that the flow on (1,3) and (1,4) must have the same magnitude, i.e.,  $\flowM_{(1,3)}^*(\bm{U}(0)) = \flowM_{(1,4)}^*(\bm{U}(0)).$  In addition, since the net requirement of vertex 1 is equal to 1, $\flowM_{(1,2)}^*(\bm{U}(0))  + \flowM_{(1,3)}^*(\bm{U}(0)) +  \flowM_{(1,4)}^*(\bm{U}(0))  = 1$, it follows that $\flowM_{(1,3)}^*(\bm{U}(0)) = \flowM_{(1,4)}^*(\bm{U}(0)) = \frac{1-x}{2}$.  
Similarly, since vertices 3 and 4 have net requirement equal to 0, i.e., $U_3(0) = U_4(0) = 0$,  their flow balance must be 0, i.e., $$0 = \flowM^*_{(2,3)}(\bm{U}(0)) + \flowM^*_{(1,3)}(\bm{U}(0)) +\flowM^*_{(3,4)}(\bm{U}(0))  = \flowM^*_{(2,3)}(\bm{U}(0)) + \flowM^*_{(1,3)}(\bm{U}(0)),$$ $$0 = \flowM^*_{(2,4)}(\bm{U}(0)) + \flowM^*_{(1,4)}(\bm{U}(0)) +\flowM^*_{(3,4)}(\bm{U}(0))  = \flowM^*_{(2,4)}(\bm{U}(0)) + \flowM^*_{(1,4)}(\bm{U}(0)).$$ Then, altogether we find that \[F_{(1,3)}^*(\bm{U}(0)) = F_{(1,4)}^*(\bm{U}(0)) = -F_{(2,3)}^*(\bm{U}(0)) = -F_{(2,4)}^*(\bm{U}(0)) = \tfrac{1 - x}{2}.\]
Having this, we can determine $x$, by finding the minimum-flow cost. The flow cost is given by
\[c_f(\flowM^*(\bm{U}(0))) = x^3 + 4\cdot \left(\tfrac{1 - x}{2}\right)^3 \]
%Taking derivative w.r.t. $x$ and equating it to 0, yields
%\[x^2 - \frac{1}{2}(1 - x)^2 = 0 \iff x^2 + 2x - 1 = 0.\]
and its unique positive minimizer is given by $x = \sqrt{2} - 1$. This yields the optimal flow $$\flowM^*(\bm{U}(0)) = (\sqrt{2} - 1, 1 -\sqrt{2}/2,  1 -\sqrt{2}/2,-1 +\sqrt{2}/2,-1 +\sqrt{2}/2,0)^T,$$ as illustrated in Figure~\ref{fig:example}. By the symmetry of the graph, we can also immediately conclude that $$\flowM^*(\bm{U}(1)) = (1 -\sqrt{2}/2, \sqrt{2} - 1,   1 -\sqrt{2}/2,1 -\sqrt{2}/2,0,-1 +\sqrt{2}/2)^T.$$

\textit{Case 2: $\eta = \frac{1}{2}$.} We again leverage the symmetry of the problem to compute the minimal-cost flows. First, we notice that vertices $2$ and $3$ have the same net requirement equal to $-\frac{1}{2}$, making them structurally indistinguishable in the graph, because of the symmetries. This implies that $F_{(2,3)}^*(\bm{U}(\frac{1}{2})) = 0$, as the optimal flow on any path through edge (2,3) would be counteracted by the same magnitude flow in the opposite direction. Next, let $x = F_{(1,2)}^*(\bm{U}(\frac{1}{2}))$. Then, by symmetry, $$F_{(1,3)}^*\left(\bm{U}\left(\tfrac{1}{2}\right)\right) = x, \quad-F_{(2,4)}^*\left(\bm{U}\left(\tfrac{1}{2}\right)\right) = -F_{(3,4)}^*\left(\bm{U}\left(\tfrac{1}{2}\right)\right) = \frac{1}{2} - x,\quad \text{and} \quad F_{(1,4)}^*\left(\bm{U}\left(\tfrac{1}{2}\right)\right) = 1 - 2x,$$ where the last equality arises from the fact that $F_{(1,2)}^*(\bm{U}(\frac{1}{2})) + F_{(1,3)}^*(\bm{U}(\frac{1}{2})) + F_{(1,4)}^*(\bm{U}(\frac{1}{2})) = 1$. With this, we obtain the following flow cost
\[c_f\left(\flowM^*\left(\bm{U}\left(\tfrac{1}{2}\right)\right)\right) = 2\cdot x^3 + 2\cdot \left(\tfrac{1}{2} - x\right)^3 + \left(1 - 2x\right)^3.\] 
Since $x = (5 - \sqrt{5})/8$ is the unique positive minimizer of the above cost function, this yields $$\flowM^*(\bm{U}(\tfrac{1}{2})) = ((5 - \sqrt{5})/8, (5 - \sqrt{5})/8, (\sqrt{5} - 1)/4, 0, (1 - \sqrt{5})/8, (1-\sqrt{5})/8)^T.$$
\hfill$\lozenge$
\end{example}

\section{Proof of Lemma~\ref{DCPowerFlow}} \label{app:KKTproof}
\begin{proof}
    We prove the lemma by showing that the DC power flow $\flowM^* = \bm{V}\bm{U}$ satisfies the Karush–Kuhn–Tucker (KKT) conditions of the minimal-cost flow problem specified in the lemma. First, we observe that the cost function can be written as 
\[\frac{1}{2}\left(\bm{S}^{-1/2}\flowM\right)^T\left(\bm{S}^{-1/2}\flowM\right)\]
because $\bm{S}$ is a diagonal matrix with positive entries.
Hence, the corresponding KKT conditions of \eqref{optFlow} are given as follows: 
\begin{subequations}
    \begin{gather}
        \bm{S}^{-1}\flowM + \incM^T\mu = 0,\label{eq1}\\
        \incM \flowM = \bm{U},\label{eq2}
    \end{gather}
\end{subequations}
for some $\mu \in \R^{\nv}$. Recall that $\bm{V} = \bm{S}\incM^T\bm{L}^+$ and $\bm{L} = \incM \bm{S}\incM^T$. Choosing $\flowM = \bm{V}\bm{U}$ and $\mu = - \bm{L}^+\bm{U}$, we find that
\[\bm{S}^{-1}\bm{V}\bm{U} - \incM^T\bm{L}^+\bm{U}= \incM^T\bm{L}^+\bm{U} - \incM^T\bm{L}^+\bm{U}  = 0 \]
and \[\incM\bm{V}\bm{U} = \incM\bm{S}\incM^T\left(\incM\bm{S}\incM^T\right)^+\bm{U} = \bm{U}.\]
Note that the last equality is true because \begin{itemize}
    \item $\incM\bm{S}\incM^T\left(\incM\bm{S}\incM^T\right)^+$ is the orthogonal projection onto the image of $\incM\bm{S}\incM^T$,
    \item the image of $\incM\bm{S}\incM^T$ is given by $\{\bm{x}\in\R^{\nv} : \e{_{\nv}}^Tx = 0 \}$ since $\text{rank}(\incM\bm{S}\incM^T) = \text{rank}(\incM)$ = $\nv - 1$ and $\incM\bm{S}\incM^T\e{_{\nv}} = 0$,
    \item $\e{_{\nv}}^T\bm{U} = 0$ by construction, which implies that each column of $\bm{U}$ lies in the image of  $\incM\bm{S}\incM^T$ and therefore remains unchanged by the projection.
\end{itemize}
Hence, we conclude that the proposed choice of $\flowM$ and $\mu$ satisfies \eqref{eq1} and \eqref{eq2}.  Since \eqref{optFlow} is an equality-constrained convex problem, Slater's condition is satisfied, implying strong duality for \eqref{optFlow} \citep{boyd2004convex}. Last but not least, the cost function is strictly convex, which means that $\eqref{optFlow}$ has a single optimal solution. These three properties together imply that the DC power flow $\flowM^* = \bm{V}\bm{U}$ is the optimal solution of the problem in question~\citep{boyd2004convex}.
\end{proof}

\end{document}